\documentclass[letterpaper,11pt]{article}

\usepackage[runin]{abstract}
\newcommand\ddfrac[2]{\frac{\displaystyle #1}{\displaystyle #2}}

\usepackage{geometry}
\usepackage{titlesec}
\usepackage{graphicx}
\usepackage{epsfig}
\usepackage{amsmath}
\usepackage{amsthm}
\usepackage{amsfonts}
\usepackage{amssymb}
\usepackage{empheq}
\usepackage{algorithmicx}
\usepackage{algorithm}
\usepackage{algpseudocode}
\usepackage{float}
\usepackage[authoryear]{natbib}
\usepackage{url}

\usepackage[font=sf]{caption}

\newcommand{\vertiii}[1]{{\left\vert\kern-0.25ex\left\vert\kern-0.25ex\left\vert #1 
    \right\vert\kern-0.25ex\right\vert\kern-0.25ex\right\vert}}

\titleformat*{\section}{\Large\bfseries}
\titleformat*{\subsection}{\large\sc}
\titleformat*{\subsubsection}{\itshape}

\usepackage{mathtools}

\usepackage{epigraph}

\begin{document}

\title{{\bf Poising on Ariadne's thread: An algorithm for computing a maximum clique in polynomial time}}

\author{\large{Ioannis Avramopoulos and the Ghost of Helags}}

\date{}

\maketitle

\thispagestyle{empty} 

\newtheorem{definition}{Definition}
\newtheorem{proposition}{Proposition}
\newtheorem{theorem}{Theorem}
\newtheorem*{theorem*}{Theorem}
\newtheorem{corollary}{Corollary}
\newtheorem{lemma}{Lemma}
\newtheorem{axiom}{Axiom}
\newtheorem{thesis}{Thesis}

\vspace*{-0.2truecm}

\begin{abstract}
In this paper, we present a polynomial-time algorithm for the maximum clique problem, which implies {\bf P = NP}. Our algorithm is based on a continuous game-theoretic representation of this problem and at its heart lies a discrete-time dynamical system. The rule of our dynamical system depends on a parameter such that if this parameter is equal to the maximum-clique size, the iterates of our dynamical system are guaranteed to converge to a maximum clique.
\end{abstract}

\section{Introduction}

\setlength{\epigraphwidth}{0.34\textwidth}
\epigraph{``You want forever, always or never.''}{--- \textup{The Pierces}}

There was a general belief that {\bf NP}-complete problems are computationally intractable in that they required exponential time to solve in the worst case. In this paper, we prove a polynomial upper bound on such problems by giving a polynomial-time algorithm for the {\em maximum clique problem} (\citep{Pardalos2, BBPP, QWu} are introductions to this problem). Briefly, the {\em clique problem} is, given a positive integer $k$, to decide if an undirected graph $G$ has a clique of size $k$. This decision problem is {\bf NP}-complete \citep{Karp}. In this paper, we present a polynomial-time algorithm for the corresponding {\bf NP}-hard optimization problem (given an undirected graph $G$, find a maximum clique in $G$). Our approach draws on the power of characterizations. The maximum clique problem has many equivalent formulations as an integer programming problem or as a continuous non-convex optimization problem \citep{Pardalos2}. Our approach to computing a maximum clique is based on the latter continuous formulation.

\subsection{The maximum clique problem as quadratic optimization}

Let $G$ be an undirected graph and let $A$ be its adjacency matrix. \cite{Motzkin-Strauss} relate the solutions of the optimization problem
\begin{align*}
\max \left\{ X \cdot AX | X \in \Delta \right\},
\end{align*}
where
\begin{align*}
\Delta = \left\{ X \in \mathbb{R}^n \bigg| \sum_{i=1}^n X(i) = 1, X(i) \geq 0, i = 1, \ldots, n \right\},
\end{align*}
and $n$ is the number of vertices of $G$, with the maximum-clique size (clique number) of $G$. In particular, they show that if $X^*$ is a global maximizer, then 
\begin{align*}
X^* \cdot A X^* = 1 - \frac{1}{\omega(G)}
\end{align*}
where $\omega(G)$ is the clique number. They further show that the uniform strategy over a maximum clique is a global maximizer. But in their formulation other global maximizers may exist (even further stationary points do not necessarily coincide with maximal cliques \citep{Jagota}). \cite{Bomze2} shows that if $X^*$ is a maximizer of the optimization problem
\begin{align*}
\max \left\{ X \cdot C X | X \in \Delta \right\},
\end{align*}
where $C = A + 1/2 I$ and $I$ is the identity matrix, then
\begin{align*}
X^* \cdot C X^* = 1 - \frac{1}{2 \omega(G)}
\end{align*}
He further shows that uniform strategies over maximum cliques are the unique global maximizers (avoiding the spurious solutions of the Motzkin-Strauss formulation).

\subsection{The maximum clique problem as equilibrium computation}

It is often beneficial to study quadratic programs of the previous form as {\em doubly symmetric bimatrix games,} that is $2$-player games where the payoff matrix of each player is the transpose of that of the other and the payoff matrix is also symmetric. In a doubly symmetric game whose payoff matrix is C, there is a one-to-one correspondence between strict local maxima of $X \cdot CX$ over the probability simplex and {\em evolutionarily stable strategies} (for example, see \citep{Weibull}). 

The evolutionarily stable strategy (ESS) has its origins in mathematical biology \citep{TheLogicOfAnimalConflict, Evolution} but it admits a characterization \citep{HSS} more prolific in our setting: An ESS is an isolated equilibrium point that exerts an ``attractive force'' in a neighborhood (especially under the {\em replicator dynamic} in continuous\citep{TaylorJonker} or discrete \citep{Baum-Eagon} form). Starting from an interior strategy, the replicator dynamic in the doubly symmetric bimatrix whose payoff matrix is $C= A + 1/2 I$ is ensured to converge to a {\em maximal clique} (see \citep{Bomze2, Pelillo}) that is not necessarily maximum.

\subsection{Our maximum clique computation algorithm}

To compute a maximum clique in polynomial time further ideas are needed. Our approach in this paper is based on a game-theoretic construction due to \cite{Nisan-ESS}. By adapting the aforementioned result of \cite{Motzkin-Strauss} and building on \citep{Etessami}, \cite{Nisan-ESS} constructs a game-theoretic {\em transformation}  of the clique problem that receives as input an undirected graph and gives as output a doubly symmetric bimatrix game whose payoff matrix is akin to $C$. We refer to the payoff matrix of the game Nisan designed as the {\em Nisan-Bomze payoff matrix} and denote it by $C^+$. The primary question driving Nisan's inquiry is the computational complexity of recognizing an ESS: A distinctive property of the Nisan game is that a certain pure strategy, called strategy $0$ in his paper (and also denoted $E_0$ in this paper), is an ESS if and only if a parameter of the Nisan-Bomze payoff matrix (which we call the {\em Nisan parameter}) exceeds the clique number. Nisan shows that the problem of recognizing if $E_0$ is an ESS is {\bf coNP}-complete. The Nisan game is the conceptual basis of our maximum-clique computation algorithm.

Our early experience with maximum-clique computation in the Nisan game was negative. If the Nisan parameter is greater than the clique number, $E_0$ is a global ESS (GESS), which implies that $E_0$ is the unique symmetric equilibrium strategy. But if the Nisan parameter is equal to the clique number, $E_0$ remains a {\em global neutrally stable strategy} but other equilibria appear, namely, one equilibrium for each maximum clique (which we refer to as {\em maximum-clique equilibria}) and their corresponding convex combinations with $E_0$. Maximum-clique equilibria and their convex combinations with $E_0$ are {\em neutrally stable strategies}. An approach we followed to compute a maximum clique was to try to compute a neutrally stable strategy other than $E_0$. (See Section \ref{equalizers_section} for definitions of evolutionary and neutral stability as well as their global versions). The equilibria of the Nisan game when the Nisan parameter is equal to the clique number are an {\em evolutionarily stable set,} which implies that such equilibria are neutrally stable strategies that are attractive under the {\em replicator dynamic} (in continuous or discrete form). Given any neutrally stable strategy other than $E_0$ a maximum-clique can be readily recovered. Our approach to compute a neutrally stable strategy other than $E_0$ was to try to exploit that neutrally stable strategies attract {\em multiplicative weights}.\footnote{A generalization of the discrete-time replicator dynamic.} But our effort to provably stay out of the region of attraction of $E_0$ was futile. 

We then discovered that the problem of computing a maximum clique has a ``backdoor'' in the Nisan game. To unlock this backdoor we isolated the attractive force of $E_0$ by intersecting the evolution space of the Nisan game with a hyperplane perpendicular to $E_0$ and using the intersection of this hyperplane with the evolution space of the Nisan game as the evolution space of our equilibrium computation algorithm. This approach can be implemented by restricting the probability mass of strategy $E_0$ to a fixed value $\epsilon \in (0, 1)$ and by adapting the multiplicative weights algorithm. The multiplicative weights algorithm we adapted is Hedge \citep{FreundSchapire1, FreundSchapire2}. Denoting the Nisan-Bomze payoff matrix by $C$, Hedge assumes the following expression in the Nisan game:
\begin{align*}
T_i(X) = X(i) \cdot \frac{\exp\left\{ \alpha E_i \cdot CX \right\}}{ \sum_{j=1}^n X(j) \exp \left\{ \alpha E_j \cdot CX \right\} }, \quad i = 0, 1, \ldots, n
\end{align*}
where $\alpha$ is a parameter called the {\em learning rate,} which has the role of a {\em step size} in our equilibrium computation setting. If we restrict the probability mass of $E_0$ to the fixed value $0< \epsilon < 1$ our dynamical system assumes the following expression:
\begin{align*}
T_0(X) &= \epsilon\\
T_i(X) &= X(i) \cdot \frac{(1-\epsilon) \exp\left\{ \alpha E_i \cdot CX \right\}}{\sum_{j=1}^n X(j) \exp \left\{ \alpha E_j \cdot CX \right\} }, \quad i = 1, \ldots, n.
\end{align*}
The latter dynamical system is not ensured to converge to a maximum-clique equilibrium---instead, it may converge to a maximal clique. Here comes one critical idea in this vein: In the evolution space wherein this system is acting, maximum-clique equilibria have a distinctive property, namely, assuming the Nisan parameter is equal to the clique number, if $X^*$ is a maximum clique equilibrium, then
\begin{align*}
\max_{i=1}^n \left\{ (CX^*)_i \right\} - X^* \cdot CX^* = 0
\end{align*}
whereas if $X^*$ is a maximal clique equilibrium, then 
\begin{align*}
\max_{i=1}^n \left\{ (CX^*)_i \right\} - X^* \cdot CX^* < 0.
\end{align*}
To exploit this phenomenon, in an effort to enforce convergence of our dynamical system to a maximum clique we perturbed the components $E_i \cdot CX$ of the gradient $CX$ of the objective function $X \cdot CX$ with the components of the gradient of a logarithmic barrier function, a technique that is akin to {\em interior point optimization methods} barring we did not vanish the perturbation as time progresses. But we had trouble capturing the precise effect of the logarithmic barrier function on the behavior of our dynamical system analytically, and we decided to restrict the evolution space of our dynamical system even further. This latter approach is followed in this paper.

In this paper, our dynamical system evolves in a ``lower feasibility set'' (a slice of the game). This ensures that if the Nisan parameter $k$ is equal to greater than the clique number (denoted $\omega(G)$) no equilibria of $C$ appear within this feasibility set other than maximum-clique equilibria which appear when $k = \omega(G)$. Our algorithm is guided by this property starting with a large value of $k$ and incrementally decreasing $k$ until a maximum-clique equilibrium can be computed. We have derived a condition that enables us to determine when the search for a particular value of $k$ should be abandoned, that $k$ should decrease, and the search should continue using a smaller value.

\if(0)

\subsection{Our contributions}

Our main contribution is that {\bf P = NP}. Some of our detailed contributions are as follows:
\begin{enumerate}

\item We use a barrier function to restrict evolution of a dynamical system inside a desirable subset of the system's blanket evolution space via a ``growth transformation'' (see \citep{Baum-Eagon, Baum-Sell, Gopalakrishnan}). To the best of our knowledge this is the first paper where growth transformations are used in this fashion. 

\item We obtain fixed-point and equilibrium approximation bounds for the empirical average of the iterates our dynamical system generates. (In particular, for an associated sequence of ``multipliers.'') To the best of our knowledge these are the first bounds on equilibrium computation using an iterative dynamical systems approach based on multiplicative weights even in the setting of maximal clique computation wherein the problem of converging to a maximal clique via replicator dynamics has been studied quite thoroughly (see \citep{Bomze2, Bomze, Pelillo} and references therein).

\item We transform the fixed-point approximation scheme presented in this paper to a polynomial-time equilibrium computation algorithm using the concepts of {\em minimum gap} and {\em probability and payoff sectors}. 

\item We prove Hedge is a growth transformation for homogeneous polynomials with nonnegative coefficients for {\em all} positive values of the learning rate parameter, which separates the problem of selecting the learning rate from the problem of ensuring convergence to a fixed point. However, our algorithm uses the learning rate to ensure convergence to a maximum clique.

\end{enumerate}

\fi

\subsection{Our proof techniques}

In our algorithm, which we call {\em Ariadne,} the computation of a maximum clique is guided by the iterations of a dynamical system. In fact, there are two versions this dynamical system, both of which are discontinuous (but admit a continuous Lyapunov function). The secondary system is activated when the Nisan parameter is equal to the clique number. (Once the secondary dynamical system is activated, we learn the value of the clique number, but execution needs to continue to compute a maximum clique.) To prove that this combination of dynamical systems guides Ariadne toward a maximum clique, we prove asymptotic convergence to a maximum-clique equilibrium. The proof rests on the fact that these dynamical systems are {\em growth transformations} (see \citep{Baum-Eagon, Baum-Sell, Gopalakrishnan}) for a potential function. 

The growth transformations proposed by previous authors (in the aforementioned references) are based on the {\em discrete-time replicator dynamic}. In this paper, our growth transformations are based on Hedge, which facilitates the analysis deriving equilibrium approximation bounds. Our proof that Hedge is a growth transformation makes use of a limiting argument: We derive Hedge as the limit of more elementary maps and use elementary {\em functional analysis} for our conclusion.

However, perhaps the most important analytic technique introduced in this paper is the derivation of equilibrium approximation bounds using the {\em Chebyshev order inequality} (known also as the {\em generalized Chebyshev sum inequality}). To derive equilibrium approximation bounds for our dynamical system, we first derive an order preservation principle for our dynamical system which we then ``plug in'' the Chebyshev order inequality to obtain polynomial bounds on the number of iterations to approximately converge to a non-equilibrium fixed point or a maximum-clique equilibrium. 

Our main results in this latter direction are Lemma \ref{forever} (which is based on the inverse function theorem and the Hartman-Stampacchia theorem) and Theorem \ref{equilibrium_error_nonuniform} in Appendix \ref{Fixed_point_bound}, the latter applying to any symmetric bimatrix game and, therefore, its applicability is more general than the doubly symmetric games that are analyzed in this paper. In general symmetric bimatrix games, we cannot expect blanket convergence to a symmetric equilibrium strategy starting from any (interior) initialization: \cite{Daskalakis-SAGT} show that in Shapley's $3 \times 3$ symmetric bimatrix game the dynamics defined by using Hedge in each player position and computing the empirical average of each player's iterated sequence of strategies that ensues from the interaction diverge (under assumptions on the learning rate) for nonuniform initializations of play. For example, consider the symmetric game $(C, C^T)$, where 
\begin{align*}
C = \frac{1}{2} \left[ \begin{array}{cccccc}
0 & 0 & 0 & 0 & 1 & 2 \\
0 & 0 & 0 & 2 & 0 & 1 \\
0 & 0 & 0 & 1 & 2 & 0 \\
0 & 1 & 2 & 0 & 0 & 0 \\
2 & 0 & 1 & 0 & 0 & 0 \\
1 & 2 & 0 & 0 & 0 & 0 \\
\end{array} \right],
\end{align*}
which an extended form of Shapley's game.\footnote{We would like to thank an anonymous reviewer of a related submission for pointing out this example.} Figure \ref{Daskalakisetal} illustrates divergence of the sequence of averages from the uniform equilibrium starting from initial condition $(0.1, 0.2, 0.3, 0.2, 0.1, 0.1)$. This divergence phenomenon cannot manifest in Ariadne (because matrices are doubly symmetric and Lemma \ref{convergence} ensures convergence of the iterates to a maximum clique equilibrium).

\begin{figure}[tb]
\centering
\includegraphics[width=14cm]{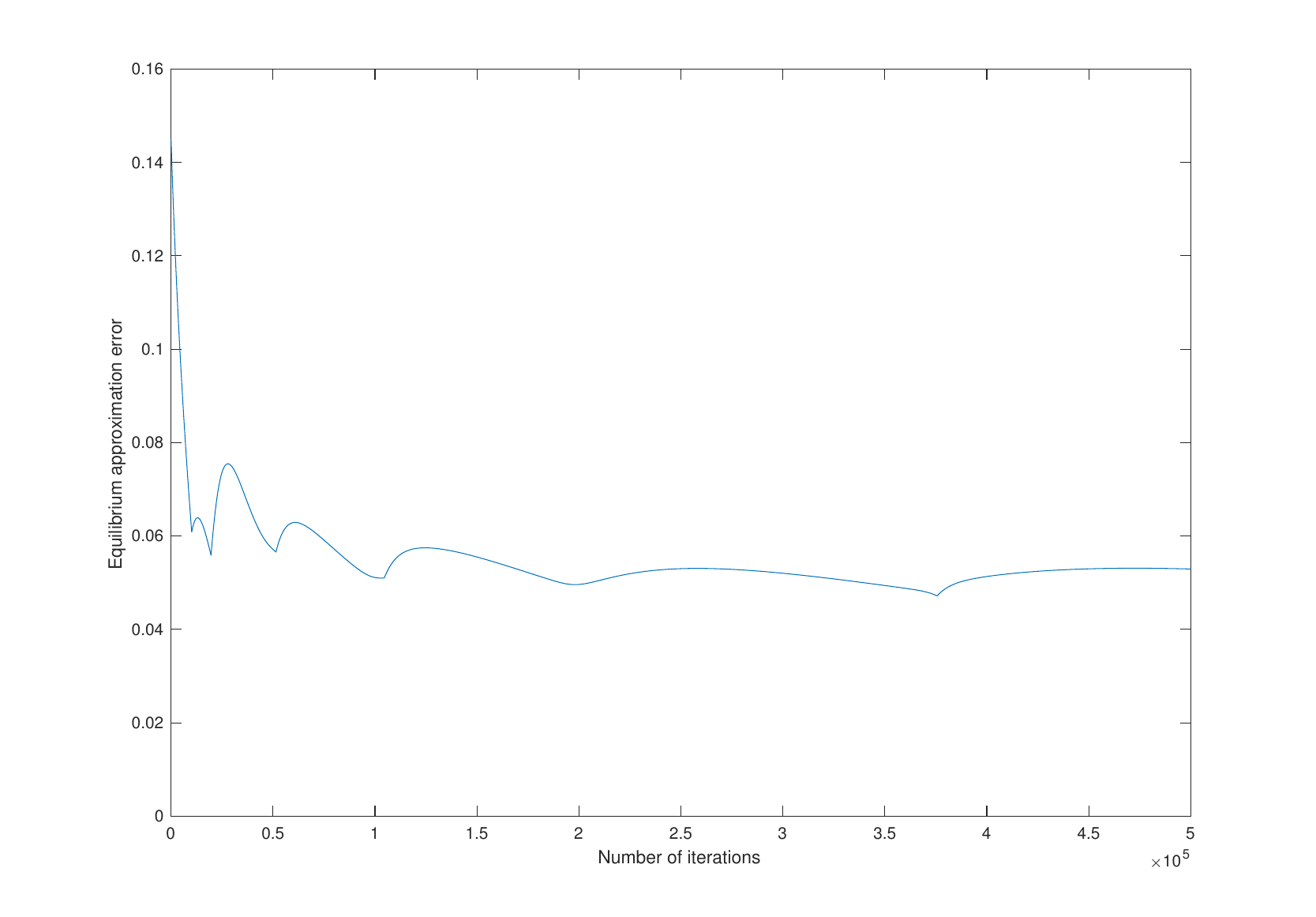}
\caption{\label{Daskalakisetal}
Divergence of the empirical average of the iterates of Hedge in extended Shapley's game starting from a non-uniform initial condition for a small learning rate.}
\end{figure}

\subsection{Overview of the rest of the paper}

The rest of this paper is organized as follows: In Section \ref{equalizers_section}, we present game-theoretic background. In Section \ref{Hedge}, we define the archetypical form of our dynamical system's map, characterize its fixed points (they coincide with the fixed points of the replicator dynamic), and give a characterization of the system's map using optimization theory.  In Section \ref{Algorithm}, we present the backbone of our algorithm, its various components, and their analysis. A main feature of our algorithm is that we use a barrier function to restrict evolution of a dynamical system inside a desirable subset of the system's blanket evolution space via a growth transformation. Using such barrier functions we are able to perform non-convex global optimization. To the best of our knowledge this is the first paper where growth transformations are used in this fashion. In Section \ref{Algorithm2}, we build on top of our algorithm's backbone to derive a dynamical system which has favorable properties regarding analytical tractability of computational complexity bounds. Using this latter formulation, in Section \ref{Complexity}, we prove that the Ariadne's complexity is polynomial. In Section \ref{conclusion}, we conclude along with discussing possible directions for future work. Finally, in the Appendix, we prove that Hedge is a growth transformation for positive values of the learning rate parameter in homogeneous polynomials with nonnegative coefficients (subject to constraints on the coefficients). We furthermore derive an inequality on the fixed-point approximation error of our dynamical system. To the extent of our knowledge, this inequality is the first equilibrium approximation bound in non-convex problems using Hedge (or other multiplicative weight algorithms such as the discrete-time replicator dynamic). Then we show that the composition of the relative entropy with Hedge is a convex function of the learning rate, and use this property to devise a lemma which in turn gives an upper bound on the learning rate that our maximum-clique computation algorithm leverages in a stage of its execution. Subsequently we derive upper bounds on the learning rate such that non-equilibrium fixed points are repelling under Hedge, in that the induced dynamics eliminate the possibility of convergence to a non-equilibrium fixed point. We, finally, give pseudocode for our clique computation algorithm.

\section{Preliminary background on Nash equilibria}
\label{equalizers_section}

\subsection{Bimatrix games and symmetric bimatrix games}

A $2$-player (bimatrix) game in normal form is specified by a pair of $n \times m$ matrices $A$ and $B$, the former corresponding to the {\em row player} and the latter to the {\em column player}. A {\em mixed strategy} for the row player is a probability vector $P \in \mathbb{R}^n$ and a mixed strategy for the column player is a probability vector $Q \in \mathbb{R}^m$. The {\em payoff} to the row player of $P$ against $Q$ is $P \cdot A Q$ and that to the column player is $P \cdot B Q$. Denote the space of probability vectors for the row player by $\mathbb{P}$ and for the column player by $\mathbb{Q}$. A Nash equilibrium of the bimatrix game $(A, B)$ is a pair of mixed strategies $P^*$ and $Q^*$ such that all unilateral deviations from these strategies are not profitable, that is, for all $P \in \mathbb{P}$ and $Q \in \mathbb{Q}$, we simultaneously have that
\begin{align}
P^* \cdot AQ^* &\geq P \cdot AQ^*\label{eqone}\\
P^* \cdot BQ^* &\geq P^* \cdot BQ.\label{eqtwo}
\end{align}
(For example, see \citep{BVS_AGT}.) $A, B$ are called {\em payoff matrices}. We denote the set of Nash equilibria of the bimatrix game $(A, B)$ by $NE(A, B)$. If $B = A^T$, where $A^T$ is the transpose of matrix $A$, the bimatrix game is called a {\em symmetric bimatrix game}. Let $(C, C^T)$ be a symmetric bimatrix game. We denote the space of symmetric bimatrix games by $\mathbb{C}$. $\mathbb{\hat{C}}$ denotes the space of payoff matrices whose entries lie in the range $[0, 1]$. Pure strategies are denoted either as $i$ or as $E_i$, where $E_i$ is a probability vector whose mass is concentrated in position $i$. $\mathbb{X}(C)$ denotes the space of mixed strategies of $(C, C^T)$ (a probability simplex). We call $(P^*, Q^*) \in NE(C, C^T)$ a symmetric equilibrium if $P^* = Q^*$. If $(X^*, X^*)$ is a symmetric equilibrium, we call $X^*$ a symmetric equilibrium strategy. It follows from \eqref{eqone} and \eqref{eqtwo} that a symmetric (Nash) equilibrium strategy $X^* \in \mathbb{X}(C)$ satisfies
\begin{align*}
\forall X \in \mathbb{X}(C) : (X^* - X) \cdot CX^* \geq 0.
\end{align*} 
$NE^+(C)$ denotes the symmetric equilibrium strategies of $(C, C^T)$. We denote the (relative) interior of $\mathbb{X}(C)$ by $\mathbb{\mathring{X}}(C)$ (every pure strategy in $\mathbb{\mathring{X}}(C)$ has probability mass). Let $X \in \mathbb{X}(C)$. We define the {\em support} or {\em carrier} of $X$ by
\begin{align*}
\mathcal{C}(X) \equiv \{ i \in \mathcal{K}(C) | X(i) > 0\}.
\end{align*}

A {\em doubly symmetric bimatrix game} \citep[p.26]{Weibull} is a symmetric bimatrix game whose payoff matrix, say $C$, is symmetric, that is $C = C^T$. Symmetric equilibria in doubly symmetric games are KKT points of a {\em standard quadratic program} (cf. \citep{Bomze}):
\begin{align*}
\mbox{ maximize } &X \cdot CX\\
\mbox{ subject to } &X \in \mathbb{X}(C).
\end{align*}
$X \cdot CX$ is the {\em potential function} of the game.

\subsection{Equalizers: Definition and basic properties}

\begin{definition}
$X^* \in \mathbb{X}(C)$ is called an {\em equalizer} if
\begin{align*}
\forall X \in \mathbb{X}(C) : (X^* - X) \cdot CX^* = 0. 
\end{align*} 
We denote the set of equalizers of $C$ by $\mathbb{E}(C)$.
\end{definition}

Note that $\mathbb{E}(C) \subseteq NE^+(C)$. Equalizers generalize interior symmetric equilibrium strategies, as every such strategy is an equalizer, but there exist symmetric bimatrix games with a non-interior equalizer (for example, if a column of $C$ is constant, the corresponding pure strategy of $C$ is an equalizer of $C$). Note that an equalizer can be computed in polynomial time by solving the linear (feasibility) program (LP)
\begin{align*}
(CX)_1 = \cdots = (CX)_n, \quad \sum_{i=1}^n X(i) = 1, \quad X \geq 0,
\end{align*}
which we may equivalently write as
\begin{align*}
CX = c \mathbf{1}, \quad \mathbf{1}^T X = 1, \quad X \geq 0,
\end{align*}
where $\mathbf{1}$ is a column vector of ones of appropriate dimension. We may write this problem as a standard LP as follows: Let
\begin{align*}
A \doteq \left[ \begin{array}{cc}
C & - \mathbf{1} \\
\mathbf{1}^T & 0 \end{array} \right]
\mbox{ and }
Y \doteq \left[ \begin{array}{c}
X \\
c \end{array} \right].
\end{align*}
then
\begin{align*}
A Y = \left[ \begin{array}{cc}
C & - \mathbf{1} \\
\mathbf{1}^T & 0 \end{array} \right]
\left[ \begin{array}{c}
X \\
c \end{array} \right] =
\left[ \begin{array}{c}
CX - c \mathbf{1} \\
\mathbf{1}^T X \end{array} \right],
\end{align*}
and the standard form of our LP, assuming $C > 0$, is
\begin{align}
\left[ \begin{array}{c}
CX - c \mathbf{1} \\
\mathbf{1}^T X \end{array} \right] = \left[ \begin{array}{c}
\mathbf{0} \\
1 \end{array} \right], X \geq 0, c \geq 0\label{my_LP}
\end{align}
where $\mathbf{0}$ is a column vector of zeros of appropriate dimension. We immediately obtain that:

\begin{lemma}
$\mathbb{E}(C)$ is a convex set.
\end{lemma}

\begin{proof}
The set of feasible/optimal solutions of a linear program is a convex set. Let
\begin{align*}
\mathbb{Y}^* = \left\{ [X^T \mbox{ } c]^T | [X^T \mbox{ } c]^T \mbox{ is a feasible solution to \eqref{my_LP}} \right\}.
\end{align*}
Then $\mathbb{Y}^*$ is convex and therefore the set
\begin{align*}
\mathbb{X}^* = \left\{ X | [X^T \mbox{ } c]^T \mbox{ is a feasible solution to \eqref{my_LP}} \right\}
\end{align*}
is also convex since $c$ is unique provided the LP is feasible.
\end{proof}

We can actually show something stronger:

\begin{lemma}
\label{lman}
If $X_1^*, X_2^* \in \mathbb{E}(C)$ then $\Lambda(X_1^*, X_2^*) \subseteq \mathbb{E}(C)$, where
\begin{align*}
\Lambda(X_1^*, X_2^*) = \left\{ (1-\lambda) X_1^* + \lambda X^*_2 \in \mathbb{X}(C) | \lambda \in \mathbb{R} \right\}.
\end{align*}
\end{lemma}

\begin{proof}
Assume $X_1^*, X_2^* \in \mathbb{E}(C)$. Then, by the definition of an equalizer,
\begin{align*}
\forall X \in \mathbb{X}(C) : X^*_1 \cdot CX^*_1 &= X \cdot CX^*_1 \mbox{ and }\\
\forall X \in \mathbb{X}(C) : X^*_2 \cdot CX^*_2 &= X \cdot CX^*_2.
\end{align*}
Let
\begin{align*}
Y^* = (1-\lambda) X_1^* + \lambda X^*_2, \mbox{ } \lambda \in \mathbb{R}.
\end{align*}
Then
\begin{align*}
Y^* \cdot CY^* &= (1-\lambda) \left( (1-\lambda) X^*_1 + \lambda X^*_2 \right) \cdot CX^*_1 + \lambda \left( (1-\lambda) X^*_1 + \lambda X^*_2 \right) \cdot CX^*_2\\
  &= (1-\lambda) X^*_1 \cdot CX^*_1 + \lambda X^*_2 \cdot CX^*_2\\
  &= (1-\lambda) X \cdot CX^*_1 + \lambda X \cdot CX^*_2\\
  &= X \cdot C \left((1-\lambda) X^*_1 + \lambda X^*_2 \right)\\
  &= X \cdot CY^*.
\end{align*}
Since $X$ is arbitrary, the proof is complete.
\end{proof}

\if(0)

\subsection{Approximate and well-supported approximate equilibria}
\label{approx_equilibria}

As mentioned earlier, conditions \eqref{eqone} and \eqref{eqtwo} simplify as follows for a symmetric equilibrium strategy $X^*$:
\begin{align*}
\forall X \in \mathbb{X}(C) : (X^* - X) \cdot CX^* \geq 0.
\end{align*}
An $\epsilon$-approximate symmetric equilibrium, say $X^*$, is defined as follows:
\begin{align*}
\forall X \in \mathbb{X}(C) : (X^* - X) \cdot CX^* \geq -\epsilon.
\end{align*}
We may equivalently write the previous expression as
\begin{align*}
(CX^*)_{\max} - X^* \cdot CX^* \leq \epsilon,
\end{align*}
where
\begin{align*}
(CX^*)_{\max} = \max\{ Y \cdot CX^* | Y \in \mathbb{X}(C) \}.
\end{align*}

Let us now give an important result on approximate equilibria. To that end, we need a definition:

\begin{definition}
\label{cool}
$(X^*, Y^*)$ is an $\epsilon$-well-supported Nash equilibrium of $(A, B)$ if
\begin{align*}
&E_i \cdot A Y^* > E_k \cdot A Y^* + \epsilon \Rightarrow X^*(k) = 0 \mbox{ and }\\
&X^* \cdot B E_j > X^* \cdot B E_k + \epsilon \Rightarrow Y^*(k) = 0.
\end{align*}
\end{definition}

Definition \ref{cool} is due to \cite{Daskalakis}. We note that an $\epsilon$-well-supported Nash equilibrium of $(A, B)$ is necessarily an $\epsilon$-approximate equilibrium of $(A, B)$ but the converse is not generally true. However, given an approximate equilibrium we can obtain a well-supported equilibrium:

\begin{proposition}
\label{cgood}
Let $(A, B)$ be such that $0 \leq A, B \leq 1$. Given an $\epsilon^2/8$-approximate Nash equilibrium of $(A, B)$, where $0 \leq \epsilon \leq 1$, we can find an $\epsilon$-well-supported Nash equilibrium in polynomial time.
\end{proposition}

The previous proposition is due to \citep{CDT} motivated by a related result in \citep{Daskalakis}. We have the following characterization of well-supported equilibria:

\begin{proposition}
\label{my_5_cents}
$(X^*, Y^*)$ is an $\epsilon$-well-supported Nash equilibrium of $(A, B)$ if and only if
\begin{align*}
X^*(i) > 0 &\Rightarrow E_i \cdot AY^* \geq \max_{k = 1}^m E_k \cdot A Y^* - \epsilon \mbox{ and }\\
Y^*(j) > 0 &\Rightarrow X^* \cdot BE_j \geq \max_{\ell = 1}^n X^* \cdot BE_\ell - \epsilon.
\end{align*}
\end{proposition}

\begin{proof}
The statement of the lemma is just the contrapositive of Definition \ref{cool}.
\end{proof}

In a symmetric bimatrix game, the previous proposition simplifies as:
\begin{proposition}
\label{my_10_cents}
$X^*$ is an $\epsilon$-well-supported symmetric equilibrium strategy of $(C, C^T)$ if and only if it is an $\epsilon$-approximate symmetric equilibrium strategy and
\begin{align*}
(\hat{C}X)_{\max} - (\hat{C}X)_{\min} \leq \epsilon
\end{align*}
where $\hat{C}$ is the carrier of $X^*$.
\end{proposition}

\fi

\subsection{Approximate and well-supported approximate equilibria}
\label{approx_equilibria}

As mentioned earlier, conditions \eqref{eqone} and \eqref{eqtwo} simplify as follows for a symmetric equilibrium strategy $X^*$:
\begin{align*}
\forall X \in \mathbb{X}(C) : (X^* - X) \cdot CX^* \geq 0.
\end{align*}
An $\epsilon$-approximate symmetric equilibrium, say $X^*$, is defined as follows:
\begin{align*}
\forall X \in \mathbb{X}(C) : (X^* - X) \cdot CX^* \geq -\epsilon.
\end{align*}
We may equivalently write the previous expression as
\begin{align*}
(CX^*)_{\max} - X^* \cdot CX^* \leq \epsilon,
\end{align*}
where
\begin{align*}
(CX^*)_{\max} = \max\{ Y \cdot CX^* | Y \in \mathbb{X}(C) \}.
\end{align*}

Let us now give an important result on approximate equilibria. To that end, we need a definition:

\begin{definition}
\label{cool}
$(X^*, Y^*)$ is an $\epsilon$-well-supported Nash equilibrium of $(A, B)$ if
\begin{align*}
&E_i \cdot A Y^* > E_k \cdot A Y^* + \epsilon \Rightarrow X^*(k) = 0 \mbox{ and }\\
&X^* \cdot B E_j > X^* \cdot B E_k + \epsilon \Rightarrow Y^*(k) = 0.
\end{align*}
\end{definition}

Definition \ref{cool} is due to \cite{Daskalakis}. We note that an $\epsilon$-well-supported Nash equilibrium of $(A, B)$ is necessarily an $\epsilon$-approximate equilibrium of $(A, B)$ but the converse is not generally true. However, given an approximate equilibrium we can obtain a well-supported equilibrium:

\begin{proposition}
\label{cgood}
Let $(A, B)$ be such that $0 \leq A, B \leq 1$. Given an $\epsilon^2/8$-approximate Nash equilibrium of $(A, B)$, where $0 \leq \epsilon \leq 1$, we can find an $\epsilon$-well-supported Nash equilibrium in polynomial time.
\end{proposition}

The previous proposition is due to \citep{CDT} motivated by a related result in \citep{Daskalakis}. We have the following characterization of well-supported equilibria:

\begin{proposition}
\label{my_5_cents}
$(X^*, Y^*)$ is an $\epsilon$-well-supported Nash equilibrium of $(A, B)$ if and only if
\begin{align*}
X^*(i) > 0 &\Rightarrow E_i \cdot AY^* \geq \max_{k = 1}^m E_k \cdot A Y^* - \epsilon \mbox{ and }\\
Y^*(j) > 0 &\Rightarrow X^* \cdot BE_j \geq \max_{\ell = 1}^n X^* \cdot BE_\ell - \epsilon.
\end{align*}
\end{proposition}

\begin{proof}
The statement of the lemma is just the contrapositive of Definition \ref{cool}.
\end{proof}

In a symmetric bimatrix game, the previous proposition simplifies as:
\begin{proposition}
\label{my_10_cents}
$X^*$ is an $\epsilon$-well-supported symmetric equilibrium strategy of $(C, C^T)$ if and only if it is an $\epsilon$-approximate symmetric equilibrium strategy and
\begin{align*}
(\hat{C}X)_{\max} - (\hat{C}X)_{\min} \leq \epsilon
\end{align*}
where $\hat{C}$ is the carrier of $X^*$.
\end{proposition}

\subsection{Evolutionary stability}

An equilibrium notion in symmetric bimatrix games (and, therefore, also in doubly symmetric bimatrix games) of primary interest in this paper is the GESS (global evolutionarily stable strategy), which is a global version of the ESS \citep{TheLogicOfAnimalConflict, Evolution}. Of primary interest are aslo related equilibrium notions such as the NSS (neutrally stable strategy) and GNSS (global NSS). These concepts admit the following definitions:

\begin{definition}
Let $C \in \mathbb{C}$. We say $X^* \in \mathbb{X}(C)$ is an ESS, if 
\begin{align*}
\exists O \subseteq \mathbb{X}(C) \mbox{ } \forall X \in O / \{ X^* \} : X^* \cdot CX > X \cdot CX.
\end{align*}
Here $O$ is a neighborhood of $X^*$. If $O$ coincides with $\mathbb{X}(C)$, we say $X^*$ is a GESS. If the above inequality is weak we have an NSS and a GNSS respectively.
\end{definition}

The aforementioned definition of an ESS was originally obtained as a characterization \citep{HSS}. Note that an NSS, and, therefore, also an ESS, is necessarily a symmetric equilibrium strategy. The ESS and NSS admit the following characterizations. These characterizations correspond to how they were initially defined.

\begin{proposition}
\label{ess_1}
$X^*$ is an ESS of $(C, C^T)$ if and only if the following conditions hold simultaneously
\begin{align*}
X^* \cdot C X^* &\geq X \cdot C X^*, \mbox{ } \forall X \in \mathbb{X}(C), \mbox{ and }\\
X^* \cdot C X^* &= X \cdot C X^* \Rightarrow X^* \cdot C X > X \cdot C X, \mbox{ } \forall X \in \mathbb{X}(C)\mbox{ such that } X \neq X^*.
\end{align*}
An NSS correspond to a weak inequality.
\end{proposition}

The characterization of the ESS and NSS in Proposition \ref{ess_1} does not readily yield the global versions of GESS and GNSS that are of primary interest in this paper. Note finally that:

\begin{lemma}
\label{111a}
If $X^*$ is an equalizer, then $X^*$ is an ESS if and only if it is a GESS.
\end{lemma}

\begin{proof}
Straightforward from Proposition \ref{ess_1}.
\end{proof}

\begin{lemma}
\label{111b}
If $X^*$ is an equalizer, then $X^*$ is an NSS if and only if it is a GNSS.
\end{lemma}

\section{A dynamical systems approach to maximum-clique computation}
\label{Hedge}

At the heart of our maximum-clique computation algorithm lies a dynamical system based on Hedge \citep{FreundSchapire1, FreundSchapire2} that induces the following map in our setting:
\begin{align*}
T_i(X) = X(i) \cdot \frac{\exp\left\{ \alpha E_i \cdot CX \right\}}{ \sum_{j=1}^n X(j) \exp \left\{ \alpha E_j \cdot CX \right\} } \equiv X(i) \cdot \frac{\exp\left\{ \alpha (CX)_i \right\}}{ \sum_{j=1}^n X(j) \exp \left\{ \alpha (CX)_j \right\} } \quad i = 1, \ldots, n,
\end{align*}
where $C$ is the payoff matrix of a symmetric bimatrix game, $n$ is the number of pure strategies, $E_i$ is the probability vector corresponding to pure strategy $i$, and $X(i)$ is the probability mass of pure strategy $i$. Parameter $\alpha$ is called the {\em learning rate}, which has the role of a {\em step size}. Our algorithm also generates iterates using the {\em discrete-time replicator dynamic,} that is, the map
\begin{align*}
J_i(X) = X(i) \cdot \frac{(CX)_i}{ X \cdot CX } \quad i = 1, \ldots, n,
\end{align*}
as necessary. It is easy to show that the fixed points $X$ of $J$ satisfy
\begin{align*}
\forall i \in \mathcal{C}(X) : (CX)_i = X \cdot CX
\end{align*}
a condition that is equivalent to
\begin{align*}
\forall i, j \in \mathcal{C}(X) : (CX)_i = (CX)_j.
\end{align*}
The fixed points of $T$ coincide with the fixed points of $J$:

\begin{lemma}
\label{fixed_points_Hedge}
$X$ is a fixed point of $T$ if and only if $X$ is a pure strategy or otherwise
\begin{align*}
\forall i, j \in \mathcal{C}(X) : (CX)_{i} = (CX)_{j}. 
\end{align*}
\end{lemma}

\begin{proof}
First we show sufficiency, that is, if for all $i, j \in \mathcal{C}(X)$, $(CX)_{i} = (CX)_{j}$,
then $T(X) = X$:
Some of the coordinates of $X$
are zero and some are positive. Clearly, 
the zero coordinates will not become positive after applying $T$. 
Now, notice that, for all $i \in \mathcal{C}(X)$, $\exp\{\alpha (CX)_i\} = \sum_{j = 1}^n X(j) \exp\{\alpha (CX)_j\}$. Therefore, $T(X) = X$.

Now we show necessity, that is,
if $X$ is a fixed point of $T$, 
then for all $i$ and for all $i, j \in \mathcal{C}(X)$, $(CX)_i = (CX)_j$:
Let $\hat{X}(i) = T_i(x)$.
Because $X$ is a fixed point, $\hat{X}(i) = X(i)$. Therefore,
{\allowdisplaybreaks
\begin{align}
\hat{X}(i) &= X(i)\notag\\
\frac{X(i) \exp{\{\alpha (CX)_i\}}}{\sum_{j} X(j) \exp{\{\alpha (CX)_j\}}} &= X(i)\notag\\
\exp{\{\alpha (CX)_i\}} &= \sum_{j} X(j) \exp{\{\alpha (CX)_j\}},\notag
\end{align}}
which implies 
\[\exp{\{\alpha ((CX)_i - (CX)_j)\}} = 1, X(i) > 0,\]
and, thus,
\[(CX)_i = (CX)_j, X(i) > 0.\]
This completes the proof.
\end{proof}

Hedge can be derived as the dual of the optimization problem\footnote{See \citep{Bowen} for a related result.}
\begin{align*}
\mbox{ minimize } &RE(Y, X) - \alpha Y \cdot CX\\
\mbox{ subject to } &Y \in \mathbb{X}(C).
\end{align*}
Let us prove this: The Lagrangian is
\begin{align}
L(Y, \lambda) = RE(Y, X) - \alpha Y \cdot CX + \lambda (\mathbf{1}^T Y - 1),\label{langrangian}
\end{align}
where we assume that the constraint $Y \geq 0$ is implicit. The dual function $L_D(\lambda)$ is obtained by minimizing the Lagrangian $L(Y, \lambda)$:
\begin{align*}
L_D(\lambda) = \inf_{Y \in \mathbb{R}^n} \{RE(Y, X) - \alpha Y \cdot CX + \lambda (\mathbf{1}^T Y - 1)\}.
\end{align*}
The Lagrangian is minimized when the gradient is zero. Observing to that end that
\begin{align*}
\frac{\partial RE(Y, X)}{\partial Y(i)} = \frac{\partial }{\partial Y(i)} \left( Y(i) \ln \frac{Y(i)}{X(i)} \right) = 1 + \ln \frac{Y(i)}{X(i)}
\end{align*}
and
\begin{align*}
\frac{\partial }{\partial Y(i)} (Y \cdot CX) = (CX)_i,
\end{align*}
we obtain that the gradient is zero when
\begin{align*}
1 + \ln \frac{Y(i)}{X(i)} - \alpha (CX)_i + \lambda = 0.
\end{align*}
Solving for $Y(i)$ in the previous expression, we obtain
\begin{align}
Y(i) = X(i) \exp\{ - 1 - \lambda + \alpha (CX)_i  \} \quad i = 1, \ldots, n.\label{firm}
\end{align}
Note that, by the previous expression, the constraint $Y \geq 0$ is automatically satisfied. Substituting now the previous expression for $Y$ in \eqref{langrangian}, we obtain the dual function
\begin{align*}
L_D(\lambda) = &\sum_{i=1}^n X(i) \exp\{ - 1 - \lambda + \alpha (CX)_i \} (- 1 - \lambda + \alpha (CX)_i) -\\
  &- \alpha \sum_{i=1}^n X(i) \exp\{ - 1 - \lambda + \alpha (CX)_i  \} (CX)_i + \\
  &+ \lambda \left(\sum_{i=1}^n X(i) \exp\{ - 1 - \lambda + \alpha (CX)_i \} - 1\right),
\end{align*}
which simplifies to
\begin{align*}
L_D(\lambda) = &- \sum_{i=1}^n X(i) \exp\{ - 1 - \lambda + \alpha (CX)_i \} -\lambda.\\
\end{align*}
This function is concave in $\lambda$. Since the dual function is concave, to find the optimal $\lambda$ we simply need to set the derivative of $L_D(\lambda)$ (with respect to $\lambda$) equal to $0$. To that end, we have
\begin{align*}
\frac{d L_D(\lambda)}{d \lambda} = \sum_{i=1}^n X(i) \exp\{ - 1 - \lambda + \alpha (CX)_i \} - 1 = 0,
\end{align*}
which implies
\begin{align*}
\exp\{ - 1 - \lambda \} = \frac{1}{\sum_{i=1}^n X(i) \exp\{ \alpha (CX)_i \}}.
\end{align*}
Substituting in \eqref{firm} we obtain
\begin{align*}
Y(i) = X(i) \frac{ \exp\{\alpha (CX)_i\}}{\sum_{i=1}^n X(i) \exp\{ \alpha (CX)_i \}} \quad i=1, \ldots, n
\end{align*}
as claimed.\\

Solving the maximum clique problem using Hedge is an approach also taken by \cite{Pelillo}, where Hedge is referred to as ``exponential replicator dynamic'' in that paper. Hedge is reported in that paper to be dramatically faster than the discrete-time replicator dynamic and even more accurate. However, a blanket application of this dynamic can compute a maximal (instead of maximum) clique. The techniques considered by \cite{Pelillo} to enhance the efficacy of the approach do not provably compute a maximum clique (as we do in this paper).

\section{Ariadne: The primary sequence of iterates}
\label{Algorithm}

\setlength{\epigraphwidth}{0.75\textwidth}
\epigraph{``It seems that for the maximum clique problem a {\em good} formulation of the problem is of crucial importance in solving the problem.''}{--- \textup{P. M. Pardalos and J. Xue}}

In this section, we define the ``backbone'' of our maximum-clique computation algorithm.

\subsection{The Nisan game as the evolution space of our dynamical system}

The evolution space of our dynamical system is a subset of the evolution space of the Nisan game, but before defining what this evolution space is, let us start by defining the Nisan game first. Given an undirected graph $G(V, E)$, where $|V| = n$, and an integer $1 < k \leq n$, consider the following $(n+1) \times (n+1)$ symmetric matrix $C^+$: $C^+$'s rows and columns correspond to the vertices of $V$, numbered $1$ to $n$, with an additional row and column, numbered $0$.
\begin{itemize}

\item For $1 \leq i \neq j \leq n$: $C^+(i, j) = 1$ if $(i, j) \in E$ and $C^+(i, j) = 0$ if $(i, j) \not\in E$.

\item For $1 \leq i \leq n$: $C^+(i, i) = 1/2$.

\item For $1 \leq i \leq n$: $C^+(0, i) = C^+(i, 0) = 1 - \frac{1}{2k}$.

\item $C^+(0, 0) = 1 - \frac{1}{2k} \equiv C_{00}$.

\end{itemize}
That is, $C^+$ consists of a symmetric adjacency matrix of $0$'s and $1$'s with the value $1/2$ on the main diagonal and an extra strategy whose corresponding payoff entries are identical and equal to the potential value of a clique of size $k$. We refer to this matrix as the {\em Nisan-Bomze payoff matrix}. 

Cliques can be identified with their {\em characteristic vectors,} that is, uniform strategies over their corresponding carrier (a property retained from \cite{Motzkin-Strauss}). Every characteristic vector (of a clique) is a fixed point of $T$ (cf. Section \ref{Hedge}). We call $k$ the {\em Nisan parameter}. 

Considering the doubly symmetric game whose payoff matrix is $C^+$, one of Nisan's main results \citep{Nisan-ESS} is that strategy $0$ (which we also denote by $E_0$) is an ESS if and only if the maximum clique of $G$ is less than $k$. Note that $E_0$ is an equalizer and, therefore, it is an ESS if and only if it is a GESS (cf. Lemma \ref{111a}). If the Nisan parameter is such that $E_0$ is a GESS it is easily shown that it is the unique equilibrium of the game. If the Nisan parameter is equal to the clique number, then other equilibria appear, namely, an equilibrium for every maximum clique (which is a global maximizer of the quadratic potential such as $E_0$ is) and a corresponding equilibrium line (of global maximizers) with terminal points $E_0$ and the respective {\em maximum-clique equilibrium}.

We use $C$ to denote the $n \times n$ matrix obtained from $C^+$ by excluding strategy $0$. We denote the probability simplex of $C$ by $\mathbb{Y}$. The payoff matrix whereby iterates are generated is obtained from $C$ by adding a positive constant matrix (for example, a matrix all of whose entries are equal to one) and scaling with a positive scalar (for example, two) such that the maximum payoff entry over the minimum payoff entry is equal to a constant (for example, two). The proof our algorithm runs in polynomial time requires this technical manipulation. In the sequel, we assume $C$ has been transformed in this fashion: One has been added to every element and the matrix has been divided by two. We also assume that parameter $C_{00}$ has been accordingly adjusted.

There is a way to generalize the previous construction. To that end, let $0 < \omega < 1$ and define a matrix $C^+_\omega$ such that:
\begin{itemize}

\item For $1 \leq i \neq j \leq n$: $C^+_\omega(i, j) = 1$ if $(i, j) \in E$ and $C^+_\omega(i, j) = 0$ if $(i, j) \not\in E$.

\item For $1 \leq i \leq n$: $C^+_\omega(i, i) = \omega$.

\item For $1 \leq i \leq n$: $C^+_\omega(0, i) = C^+_\omega(i, 0) = 1 - \frac{1-\omega}{k}$.

\item $C^+_\omega(0, 0) = 1 - \frac{1-\omega}{k}$.

\end{itemize}
We may refer to the game corresponding to this payoff matrix as the {\em generalized Nisan game} that has properties analogous to the Nisan game (that is, the generalized Nisan game corresponding to $\omega = 1/2$). The benefit of considering the generalized Nisan game is that if $C$ (without adding the constant matrix) is not invertible (which may happen if the corresponding adjacency matrix has the eigenvalue $-1/2$ as follows from elementary matrix theory), there exists $\omega$ such that $C_\omega$, the $n \times n$ matrix obtained from $C^+_\omega$ by excluding strategy $0$, is invertible. The invertibility of $C$ is essential in Lemma \ref{forever}. To avoid cluttering the notation, we assume that $C$ is invertible.

\subsection{Evolution inside a desirable ``feasibility set''}

Our algorithm uses up to three dynamical systems, namely, a {\em preliminary,} a {\em primary,} and a {\em secondary.} The preliminary dynamical system is used to initialize the primary (which is our algorithm's ``heart''). The secondary system is activated upon the iterates of the primary dynamical system entering a neighborhood of a maximum-clique equilibrium (when the Nisan parameter is equal to the clique number). To ensure the iterates compute a maximum clique when the Nisan parameter becomes equal to the clique number, we restrict the evolution space of our primary and secondary systems. We call the restricted evolution space the ``lower feasibility set'' (noting that the iterates of the primary dynamical system may temporarily escape to the ``upper feasibility set''). The complement of the lower feasibility set consists of the ``upper feasibility set'' and the ``infeasibility set'' (see Figure \ref{feasibility_sets}). Our algorithm ensures that the iterates of the primary and secondary dynamical systems are initialized (using the preliminary system) and remain in the lower feasibility set (barring possible transient excursions of the primary and secondary dynamical systems to the ``upper feasibility set''). Let us define these sets precisely:

\begin{figure}[tb]
\centering
\includegraphics[width=14cm]{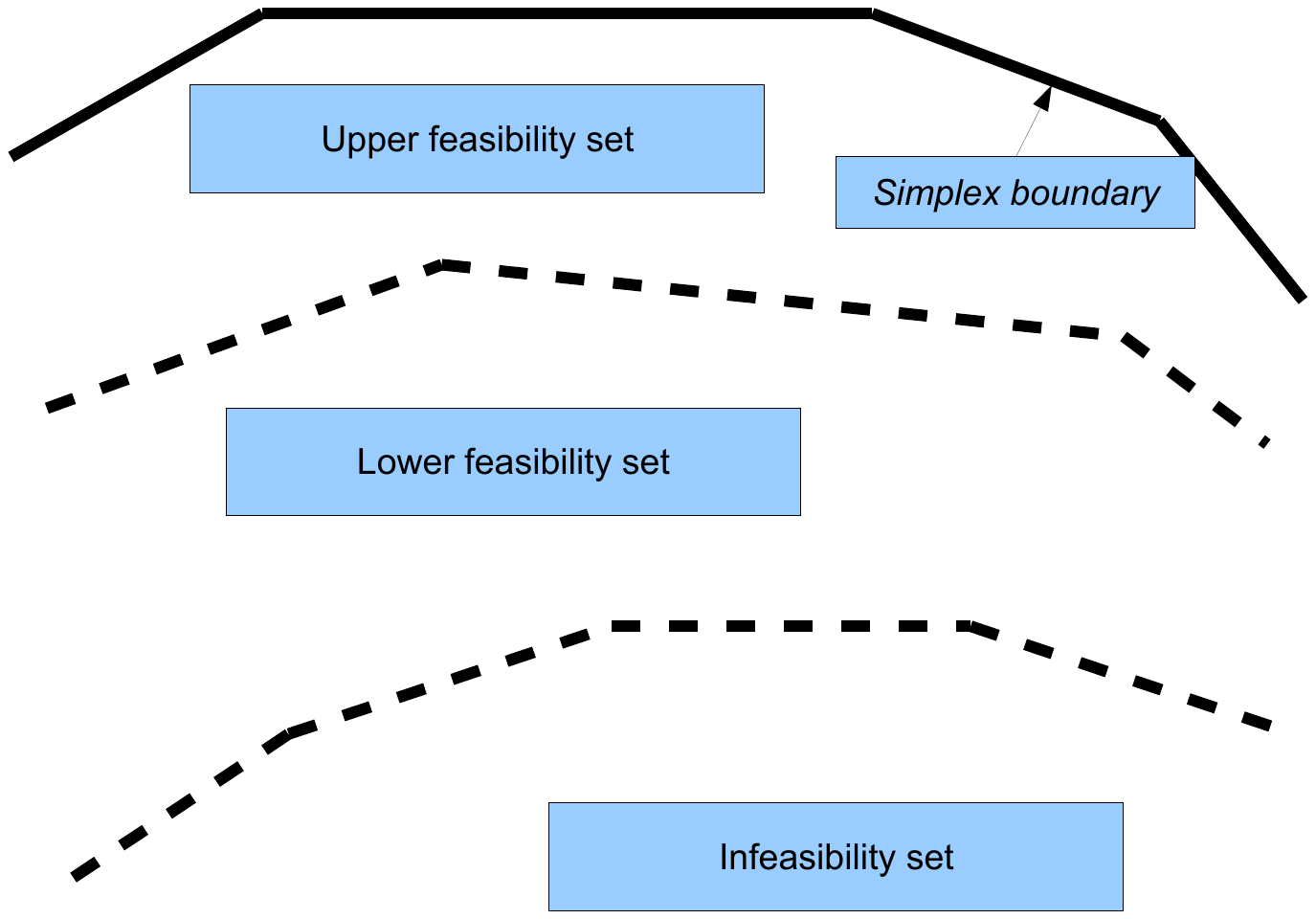}
\caption{\label{feasibility_sets}
The feasibility and infeasibility sets.}
\end{figure}

\subsubsection{Definition of the feasibility and infeasibility sets}

We denote the lower feasibility set by $\mathbb{F_L}$, the upper feasibility set by $\mathbb{F_U}$ and the infeasibility set by $\mathbb{I}$ and define them as
\begin{align*}
\begin{cases} \mathbb{F_U} &\equiv \left\{ X \in \mathbb{Y} \big| \max_{i=1}^n \left\{ (CX)_i \right\} > C_{00} + \varepsilon \right\}\\
\mathbb{F_L} &\equiv \left\{ X \in \mathbb{Y} \big| C_{00} - \varepsilon' \leq \max_{i=1}^n \left\{ (CX)_i \right\} \leq C_{00} + \varepsilon \right\}\\
\mathbb{I} &\equiv \left\{ X \in \mathbb{Y} \big| \max_{i=1}^n \left\{ (CX)_i \right\} < C_{00} - \varepsilon' \right\}\end{cases}.
\end{align*} 
All sets are polytopes and the infeasibility set $\mathbb{I}$ is also convex. We will specify values for $\varepsilon$ and $\varepsilon'$ in the sequel. Let us note for now that these parameters are chosen such that only maximum clique equilibria and no other equilibrium fixed points capable of attracting the iterates of our dynamical system can be located in the lower feasibility set $\mathbb{F_L}$. $X \in \mathbb{Y}$ is called strictly upper feasible if 
\begin{align*}
\max_{i=1}^n \left\{ (CX)_i \right\} > C_{00} + \varepsilon
\end{align*}
and strictly lower feasible if
\begin{align*}
C_{00} - \varepsilon' < \max_{i=1}^n \left\{ (CX)_i \right\} < C_{00} + \varepsilon.
\end{align*}

\if(0)

\begin{lemma}
\label{first_feasibility_lemma}
Suppose $X^*$ is a maximal-clique equilibrium and the Nisan parameter $k$ satisfies $k > X^* \cdot CX^* = \omega(G)$. Then $X^* \not\in \mathbb{F_L}$. Furthermore if $k = X^* \cdot CX^*$, $X^*$ is strictly lower feasible.
\end{lemma}

\begin{proof}
If $k > X^* \cdot CX^*$, then 
\begin{align*}
X^* \cdot CX^* = \max_{i=1}^n \left\{ (CX^*)_i \right\} < C_{00} - \varepsilon
\end{align*}
and the definition of the feasibility set implies $X^* \not\in \mathbb{F_L}$ as claimed. If $k = X^* \cdot CX^*$, then
\begin{align*}
X^* \cdot CX^* = \max_{i=1}^n \left\{ (CX^*)_i \right\} = C_{00},
\end{align*}
which implies $X^*$ is strictly lower feasible.
\end{proof}

\fi

\subsubsection{Preliminaries of the method by which iterates remain within the lower feasibility set}

The primary mechanism by which we restrict evolution of the iterates to the interior of the lower feasibility set $\mathbb{F_L}$ is by means of two {\em barrier functions}.\footnote{See \citep[Chapter 4]{Bertsekas} on {\em Lagrange multiplier algorithms} (see also \citep{Bertsekas-Lagr}) where barrier functions are discussed in their elementary form in conjunction with {\em interior point algorithms}.} The first primary barrier function is a function $\mathsf{G} : \mathbb{Y} \rightarrow \mathbb{R}$ where
\begin{align*}
\mathsf{G}(X) = \ddfrac{X \cdot CX - \mathsf{C}}{ \left( \max_{i=1}^n \left\{ (CX)_i \right\}  - \mathsf{C}_\ell \right) \prod_{i=1}^n \left( \mathsf{C}_u - (CX)_i \right) } \quad X \in \mathbb{F_L}
\end{align*}
where $\mathsf{C} = C_{00}$, $\mathsf{C}_\ell = C_{00} - \varepsilon'$, and $\mathsf{C}_\ell = C_{00} + \varepsilon$. Using the product in the denominator of the barrier function, instead of
\begin{align*}
\mathsf{C}_u - \max_{i=1}^n \left\{ (CX)_i \right\},
\end{align*}
is a trick that facilitates our subsequent analysis and appears in a related form in \citep[Proposition 3.3.10]{Bertsekas}. The second primary barrier function is
\begin{align*}
\mathsf{G}'(X) = \ddfrac{\left( X \cdot CX - \mathsf{C} \right) \left( \frac{1}{2} X \cdot X \right)}{ \left( \max_{i=1}^n \left\{ (CX)_i \right\}  - \mathsf{C}_\ell \right) \prod_{i=1}^n \left( \mathsf{C}_u - (CX)_i \right) } \quad X \in \mathbb{F_L}
\end{align*}
with the same parameters as above. We call $\mathsf{G}$-feasible any $X \in \mathbb{Y}$ which is strictly lower feasible that is
\begin{align*}
\mathsf{C}_\ell < \max_{i=1}^n \left\{ (CX)_i \right\} < \mathsf{C}_u.
\end{align*}
We assume that the initial condition of the primary dynamical system is $\mathsf{G}$-feasible and we will prove that the remaining iterates remain so, that is,
\begin{align*}
\forall k \geq 0 : \mathsf{C}_\ell < \max_{i=1}^n \left\{ (CX^k)_i \right\} < \mathsf{C}_u,
\end{align*}
barring excursions to the upper feasibility set. In the previous expression, we start counting iterations from $0$ for simplicity. The secondary dynamical system is discussed in Section \ref{secondary_system}.

\subsection{How iterates (typically) remain in the lower feasibility set}

The discrete-time replicator dynamic $J$ is a {\em growth transformation} for any polynomial with nonnegative coefficients \citep{Baum-Eagon, Baum-Sell}. That is, the discrete-time replicator dynamic strictly ascends polynomials with nonnegative coefficients except at fixed points wherein the value of the polynomial remains constant. This implies that in a doubly symmetric bimatrix game $C$, $\forall X \in \mathbb{X}(C)$ such that $X$ is not a fixed point, this dynamic ascends the potential function $\mathsf{P} : \mathbb{X}(C) \rightarrow \mathbb{R}$, where $\mathsf{P}(X) = X \cdot CX$. In the Appendix, we prove that $T$ (cf. Section \ref{Hedge}) is a growth transformation for $\mathsf{P}$ $\forall \alpha > 0$ provided $0 \leq C \leq 1$ (element-wise). We ensure that the iterates of our dynamical system remain $\mathsf{G}$-feasible and to that end we design corresponding growth transformations for $\mathsf{G}$ and $\mathsf{G'}$. To that end, we extend a result by \cite{Gopalakrishnan}, namely, that rational functions admit growth transformations, which can be obtained by growth transformations for corresponding polynomials. One of their results is that:

\begin{proposition}
\label{Gopalakrishnan_proposition}
Given a domain
\begin{align*}
D = \left\{ X \in \mathbb{R} \bigg| \sum_{j=1}^{q_i} X_{ij} = 1 \quad X_{ij} \geq 0 \quad i=1, \ldots, p \quad j =1, \ldots, q_i \right\}
\end{align*}
and a rational function $R(X) = S_1(X) \big/ S_2(X)$ where $S_1(X)$ and $S_2(X)$ are polynomials with real coefficients and $S_2(X)$ has only positive values in $D$, for any $X \in D$, there exists a polynomial $P_q(X)_{q \leftarrow X}$ parametrized by $X$ such that
\begin{align*}
P_q(Y)_{q \leftarrow X} > P_q(X)_{q \leftarrow X} \Rightarrow R(Y) > R(X)
\end{align*}
and for this it is enough to set 
\begin{align*}
P_q(X)_{q \leftarrow X} = S_1(X) - R(q)_{q \leftarrow X} S_2(X).
\end{align*}
\end{proposition}

At the heart of Ariadne lies a (primary) dynamical system based on Hedge that is the composition of two maps, one growth transformation for the first primary barrier function and another for the second primary barrier function. Parameters are adjusted such that each of these maps ascends $\mathsf{P}$ (we will show that this is always possible). The secondary system (Section \ref{secondary_system}) also ascends $\mathsf{P}$.

\begin{lemma}
\label{gt_lemma}
Given a $\mathsf{G}$-feasible $X$, there exists an $n \times n$ positive matrix $\bar{C}_X$, which we call the operative matrix at $X$, that depends on $X$, such that
\begin{align*}
T_{\mathsf{G}}(X)_i = X(i) \ddfrac{\exp\left\{ \alpha (\bar{C}_X X)_i \right\}}{\sum_{j=1}^n X(j) \exp\left\{ \alpha (\bar{C}_X X)_j \right\}} \quad i =1, \ldots, n 
\end{align*}
is a growth transformation for $\mathsf{G}$. Furthermore if the initial condition is $\mathsf{G}$-feasible and the learning rate $\alpha$ is then chosen from iteration to iteration to be equal to
\begin{align*}
(\exp\{\alpha\} - 1)^2 = \frac{1}{4} \delta^2,
\end{align*}
where $\delta$ is the minimum of the distance between $X$ and the upper feasibility set and the distance between $X$ and the infeasibility set, then repeatedly applying $T_{\mathsf{G}}$, the iterates $\{ X^k \}$ remain $\mathsf{G}$-feasible.
\end{lemma}

\begin{proof}
To prove the second part of the lemma, we prove that $\hat{X} \in \mathbb{F_L}$ where $\hat{X} = T_{\mathsf{G}}(X)$. Along the way, we prove the first part of the lemma. We may write the barrier function as
\begin{align*}
\mathsf{G}(X) &= \ddfrac{X \cdot CX - \mathsf{C}}{ \left( \max_{i=1}^n \left\{ (CX)_i \right\}  - \mathsf{C}_\ell \right) \prod_{i=1}^n \left( \mathsf{C}_u - (CX)_i \right) }\\
  &= \ddfrac{X \cdot CX - \mathsf{C}}{ \left( \max_{Y \in \mathbb{Y}} \left\{ Y \cdot CX \right\}  - \mathsf{C}_\ell \right) \prod_{i=1}^n \left( \mathsf{C}_u - (CX)_i \right) }\\
  &= \ddfrac{X \cdot CX - \mathsf{C}}{ \left( Y \cdot CX  - \mathsf{C}_\ell \right) \prod_{i=1}^n \left( \mathsf{C}_u - (CX)_i \right) }\\
  &= \mathsf{G}(X | Y)
\end{align*}
where $Y$ is a best response to $X$.\footnote{If the best response $Y$ is not unique, a variety of rules can be used to compute a best response such as selecting any one of them---see \citep[pp. 245-7]{ConvexAnalysis} for the computation of the {\em subdifferential}.} To prove that $X \in \mathbb{F_L}$ implies $\hat{X} \in \mathbb{F_L}$ we prove that 
\begin{align}
\mathsf{G}(\hat{X}) = \mathsf{G}(\hat{X} | \hat{Y}) > \mathsf{G}(\hat{X} | Y) > \mathsf{G}(X | Y) = \mathsf{G}(X).\label{GT}
\end{align}
Our proof extends an idea of \cite{Gopalakrishnan} (cf. Proposition \ref{Gopalakrishnan_proposition}) who reduce the problem of devising a growth transformation for a rational function to one of devising a growth transformation for a corresponding polynomial. In our particular problem,  their methodology stipulates that a growth transformation for the polynomial 
\begin{align*}
P_{\mathsf{G}}(X | Y) = \left( X \cdot CX - \mathsf{C} \right) - \left. \mathsf{G}(q | Y) \right|_{q \leftarrow X} \left( Y \cdot CX  - \mathsf{C}_\ell \right) \prod_{\ell =1}^n \left( \mathsf{C}_u - (CX)_\ell \right),
\end{align*}
where $Y$ is a best response to $X$, is also a growth transformation for $\mathsf{G}(X|Y)$. Growth transformations for $P_{\mathsf{G}}(X | Y)$ are based on its gradient, which assumes the expression
\begin{align*}
\ddfrac{\partial P_{\mathsf{G}}(X | Y)}{\partial X(i)} = (CX)_i - \mathsf{G}(X) \left( \prod_{\ell=1}^n \left( \mathsf{C}_u - (CX)_\ell \right) \right) (CY)_i +
\end{align*}
\begin{align}
+ G(X) \left( \max_{i=1}^n \left\{ (CX)_i \right\} - \mathsf{C}_\ell \right) \sum_{m=1}^n \left( \prod_{\ell =1, \ell \neq m}^n \left( \mathsf{C}_u - (CX)_\ell \right) \right) C_{im}\label{avc}
\end{align}
and which is equal to the gradient of
\begin{align*}
Q(X | Y) = X \cdot CX - \left. \mathsf{G}(q) \left( \prod_{l=1}^n \left( \mathsf{C}_u - (Cq)_l \right) \right) \right|_{q \leftarrow X} Y \cdot CX +
\end{align*}
\begin{align}
+ \left. G(q) \left( \max_{i=1}^n \left\{ (Cq)_i \right\} - \mathsf{C}_\ell \right) \sum_{m=1}^n \left( \prod_{l=1, l \neq m}^n \left( \mathsf{C}_u - (Cq)_l \right) \right) \right|_{q \leftarrow X} (CX)_m.\label{acv}
\end{align}
To find a growth transformation for $P_{\mathsf{G}}(X | Y)$ it suffices to find a growth transformation for $Q(X | Y)$. The advantage of this equivalence is that $Q(X|Y)$ can be expressed in the form of a homogeneous quadratic function using a trick by \cite{Bomze}. We may thus write
\begin{align*}
Q(X|Y) = X \cdot \bar{C} X,
\end{align*}
where $\bar{C}$ is a square matrix. We may add a positive constant to the entries of this matrix and then normalize with a positive scalar to obtain a positive matrix $\bar{C}_X$, which we call the operative matrix at $X$. In Lemma \ref{Baum_Eagon_extension_exponential_function} in the appendix, we prove that applying $T$ on $X$ using $\bar{C}_X$, unless $X$ is a fixed point corresponding to $\bar{C}_X$, we obtain $\hat{X}$ such that $\hat{X} \cdot \bar{C}_X \hat{X} > X \cdot \bar{C}_X X$. This implies 
\begin{align}
\mathsf{G}(\hat{X} | Y) = \mathsf{G}(T_{\mathsf{G}} (X) | Y) > \mathsf{G}(X | Y).\label{asdkfjsfkdjsfajhf}
\end{align}
Keeping now $\hat{X}$ fixed and considering the polynomial $P_{\mathsf{G}}(\hat{X}|Y)$ in the variable $Y$ this time, we have
\begin{align*}
P_{\mathsf{G}}(\hat{X} | Y) = \left( \hat{X} \cdot C\hat{X} - \mathsf{C} \right) - \left. \mathsf{G}(\hat{X} | q) \right|_{q \leftarrow Y} \left( Y \cdot C\hat{X}  - \mathsf{C}_\ell \right) \prod_{\ell =1}^n \left( \mathsf{C}_u - (C\hat{X})_\ell \right).
\end{align*}
Maximizing with respect to $Y$ weakly increases $P_{\mathsf{G}}(\hat{X} | Y)$ and, therefore, also weakly increases $\mathsf{G}(\hat{X} | Y)$, implying that
\begin{align*}
\mathsf{G}(\hat{X} | \hat{Y}) = \max_{Y \in \mathbb{Y}} \left\{ \mathsf{G}(\hat{X} | Y) \right\} \geq \mathsf{G}(\hat{X} | Y)
\end{align*}
and, therefore, combining with \eqref{asdkfjsfkdjsfajhf}, proving our claim \eqref{GT} (which implies that $T_{\mathsf{G}}$ is a growth transformation for the barrier function $\mathsf{G}$). It remains prove that this ensures the iterates of our dynamical system cannot leap across the negative infinity barrier. There two ways things could go wrong. The first is an even number of terms in the denominator's product become negative. To prove that this is not possible, consider the barrier function
\begin{align*}
\mathsf{G}'(X) = \ddfrac{X \cdot CX - \mathsf{C}}{ \left( \max_{i=1}^n \left\{ (CX)_i \right\}  - \mathsf{C}_\ell \right) \left( \mathsf{C}_u - \max_{i=1}^n \left\{ (CX)_i \right\} \right) \prod_{i=1, i \neq \max}^n | \mathsf{C}_u - (CX)_i | },
\end{align*}
where $|\cdot|$ is the absolute value. Any growth transformation for $\mathsf{G}'$ satisfies the property that if $X$ is $\mathsf{G}$-feasible, then $\hat{X}$ is also $\mathsf{G}$-feasible (provided the learning rate is small enough). But $T_{\mathsf{G}}$ is a growth transformation for $\mathsf{G}'$, which implies that if $X^0$ is $\mathsf{G}$-feasible, repeatedly applying $T_{\mathsf{G}}$ (using a small enough learning rate), the iterates $\{ X^k \}$ remain $\mathsf{G}$-feasible. How small should the learning be? If the learning rate is chosen as in the statement of the lemma, Lemma \ref{convexity_lemma_2} and Pinsker's inequality imply that
\begin{align}
\frac{1}{2} \|X - T(X)\|^2 \leq RE(X, T(X)) \leq \alpha (\exp\{\alpha\} - 1) \leq (\exp\{\alpha\} - 1)^2\label{vgood_inequality}
\end{align}
which implies iterates cannot ``jump across'' the negative infinity barrier, completing the proof.
\end{proof}

An analogous result holds for the second primary barrier function.\\ 

Computing $\delta$ requires solving two convex optimization problems, namely,
\begin{align*}
\mbox{ minimize } &\| Y - X \|\\
\mbox{ subject to } &(CY)_{\max} \leq \mathsf{C}_\ell\\
 &Y \in \mathbb{Y},
\end{align*}
which computes the distance to the infeasibility set, and
\begin{align*}
\mbox{ maximize } &\rho\\
\mbox{ subject to } &\|Y - X\| \leq \rho\\
 &(CY)_{\max} \leq C_u\\
 &Y \in \mathbb{Y},
\end{align*}
which computes the distance to the upper feasibility set and setting $\delta$ to be equal to the minimum of these distances. These problems need only be solved approximately, as the algorithm is not sensitive to an exact solution, but performing this computation in every iteration is time-consuming. Our algorithm thus uses an alternative method to find an appropriate value for the learning rate.

\subsubsection{Partitioning $\mathbb{F_L}$ to ensure the barrier functions are bounded away from $-\infty$}

\begin{figure}[tb]
\centering
\includegraphics[width=14cm]{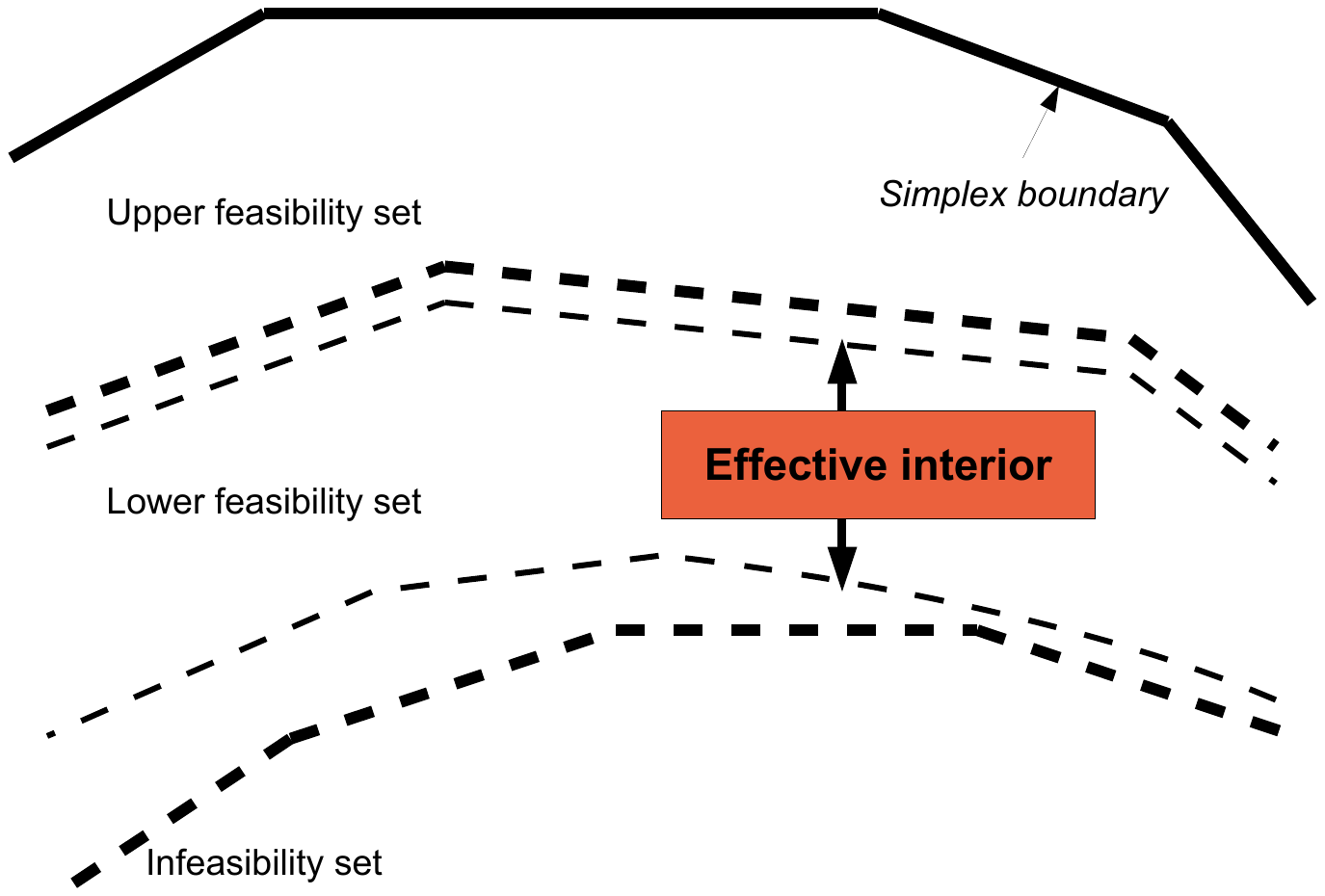}
\caption{\label{effective_interior}
The effective interior and the lower and upper boundaries of the lower feasibility set.}
\end{figure}

Our algorithm is designed such that iterates ascend the potential function $\mathsf{P}$ rather than $\mathsf{G}$ and $\mathsf{G'}$. As $\mathsf{P}$ increases, $\mathsf{G}$ and $\mathsf{G'}$ may descrease. The correctness of our algorithm relies on a lower bound on these barrier functions away from $-\infty$ for if the barrier functions can descend to $-\infty$, iterates may get stuck due to a diminishing learning rate rather than convergence to an equilibrium. To attain this lower bound we further partition the lower feasibility set into three sets, namely, the {\em lower boundary}, the {\em effective interior}, and the {\em upper boundary} (see Figure \ref{effective_interior}). The lower boundary consists of all strategies in the lower feasibility set such that  
\begin{align*}
\mathsf{G}^*(X) = \ddfrac{X \cdot CX - \mathsf{C}}{\max_{i=1}^n \left\{ (CX)_i \right\} - \mathsf{C}_\ell} \leq \mathsf{G}^*_0
\end{align*}
for some $\mathsf{G}^*_0$ small, certainly small enough such that the lower feasibility set is not partitioned, for example,
\begin{align*}
\mathsf{G}^*_0 = - \ddfrac{n}{C_{00} - C_\ell}.
\end{align*}
Let us further call the set of strategies of the lower feasibility set such that
\begin{align*}
\max_{i=1}^n \left\{ (CX)_i \right\} \geq \mathsf{C}_u - \epsilon
\end{align*}
for some $\epsilon > 0$ small, for example, 
\begin{align*}
\epsilon = \frac{1}{n} (\mathsf{C}_u - C_{00}),
\end{align*} 
the {\em upper boundary} of the lower feasibility set. The set of strategies such that
\begin{align*}
\mathsf{G}^*(X) > \mathsf{G}_0^* \quad \mbox{ and } \quad \max_{i=1}^n \left\{ (CX)_i \right\} < \mathsf{C}_u - \epsilon.
\end{align*}
is called the {\em effective interior}.

\subsubsection{Selecting parameters to ensure iterates remain in the effective interior of the lower feasibility set as the potential function $\mathsf{P}$ monotonically increases}
\label{barrier_lower_bound}

We are going to consider two bounds for the learning rate: an upper bound $\alpha_h$ (say equal to $1.0$) and a lower bound $\alpha_\ell$ (say equal to $0.001$). We use the upper bound $\alpha_h$ to ``wind'' parameter $\mathsf{C}$ such that $\mathsf{P}$ increases. To make sure that our dynamical system monotonically increases the potential function $\mathsf{P}$ from iteration to iteration, we treat each $\mathsf{C}$ (corresponding to the first and second barrier functions) as a parameter (to that effect) making sure that it is always greater than or equal to $\mathsf{P}(X)$. That using this principle in the design of our algorithm can be effective in monotonically increasing $\mathsf{P}$ is argued by the next lemmas:

\begin{lemma}
\label{123abc}
For all $X$ in the lower feasibility set and for all $\hat{\alpha} > 0$, there exists $\mathsf{\hat{C}} > \mathsf{P}(X)$ such that for all $\alpha \leq \hat{\alpha}$ and for all $\mathsf{C}$ such that $P(X) < \mathsf{C} \leq \mathsf{\hat{C}}$ we have that $\mathsf{P}(T_{\mathsf{G}}(X)) > \mathsf{P}(X)$.
\end{lemma}

\begin{proof}
Expression \eqref{avc} in Lemma \ref{gt_lemma} implies that if $\mathsf{C} = \mathsf{P}(X)$, $T_{\mathsf{G}}(X) = T(X)$, and, therefore, by Lemma \ref{Baum_Eagon_extension_exponential_function} (in the Appendix), we have that $\forall \alpha > 0 : \mathsf{P}(T_{\mathsf{G}}(X)) = \mathsf{P}(T(X)) > \mathsf{P}(X)$. Since $\mathsf{P}(T_{\mathsf{G}}(X))$ is a continuous function of $\mathsf{C}$ and $\alpha$, the intermediate value theorem completes the proof.
\end{proof}

\begin{lemma}
\label{123abcd}
For all $X$ in the lower feasibility set and for all $\hat{\alpha} > 0$, there exists $\mathsf{\hat{C}} > \mathsf{P}(X)$ such that for all $\alpha \leq \hat{\alpha}$ and for all $\mathsf{C}$ such that $P(X) < \mathsf{C} \leq \mathsf{\hat{C}}$ we have that  $\mathsf{P}(T_{\mathsf{G'}}(X)) > \mathsf{P}(X)$.
\end{lemma}

\begin{proof}
Analogous to the proof of Lemma \ref{123abc}.
\end{proof}

\begin{lemma}
\label{123abcde}
For all $X$ in the lower feasibility set and for all $\hat{\alpha} > 0$, there exists $\mathsf{\hat{C}} > \mathsf{P}(X)$ and $\mathsf{\hat{C'}} > \mathsf{P}(T_{\mathsf{G}}(X))$ such that for all $\alpha \leq \hat{\alpha}$ and for all $\mathsf{C}$ such that $P(X) < \mathsf{C} \leq \mathsf{\hat{C}}$ and for all $\mathsf{C'}$ such that $P(T_{\mathsf{G}}(X)) < \mathsf{C'} \leq \mathsf{\hat{C'}}$ we have that $\mathsf{P}(T_{\mathsf{G}}(X)) > \mathsf{P}(X)$ and $\mathsf{P}((T_{\mathsf{G'}} \circ T_{\mathsf{G}})(X)) > \mathsf{P}(T_{\mathsf{G}}(X)) > \mathsf{P}(X)$.
\end{lemma}

\begin{proof}
We apply the reasoning in Lemma \ref{123abc} first followed by the reasoning in Lemma \ref{123abcd}.
\end{proof}

\noindent
In the succeeding discussion, whenever we refer to the map $T_{\mathsf{G}'} \circ T_{\mathsf{G}}$ we assume both maps use the same $\alpha$ and {\em corresponding} values for parameter $\mathsf{C}$ (in that the update rule that generates these values is applied an identical number of times) and that these parameters are configured such that either map ascends $\mathsf{P}$. To that effect, given $X$ and $\alpha$, our algorithm selects $\mathsf{C}$ starting with $\mathsf{C} = C_{00}$, computes $\hat{X} = T_{\mathsf{G}}(X)$ and $\hat{X'} = T_{\mathsf{G}'}(\hat{X})$ and sets 
\begin{align*}
\mathsf{C} \leftarrow \frac{1}{2} \left( \mathsf{C} + \mathsf{P}(X) \right) \mbox{ and } \mathsf{C'} \leftarrow \frac{1}{2} \left( \mathsf{C'} + \mathsf{P}(\hat{X}) \right)
\end{align*} 
upon failure to meet our objective that $\mathsf{P}$ increases for either map, that is, upon failure that
\begin{align*}
\mathsf{P}(\hat{X}) > \mathsf{P}(X) \quad \mbox{ and } \quad \mathsf{P}(\hat{X'}) > \mathsf{P}(\hat{X}).
\end{align*}
This process is executed in an iterative fashion and Lemma \ref{123abcde} guarantees that it terminates. However, after a constant number of iterations (for example, ten or twenty) our algorithm sets 
\begin{align}
\mathsf{C} \leftarrow \mathsf{P}(X) \mbox{ and } \mathsf{C'} \leftarrow \mathsf{P}(\hat{X}).\label{fgfgfg}
\end{align} 
This has the effect that $\mathsf{P}$ increases: To see this, considering the growth transformation for the first primary barrier function $\mathsf{G}$, let us recall equation \eqref{acv}:
\begin{align*}
Q(X) = X \cdot CX - \left. \mathsf{G}(q) \left( \prod_{l=1}^n \left( \mathsf{C}_u - (Cq)_l \right) \right) \right|_{q \leftarrow X} Y \cdot CX +
\end{align*}
\begin{align*}
+ \left. G(q) \left( \max_{i=1}^n \left\{ (Cq)_i \right\} - \mathsf{C}_\ell \right) \sum_{m=1}^n \left( \prod_{l=1, l \neq m}^n \left( \mathsf{C}_u - (Cq)_l \right) \right) \right|_{q \leftarrow X} (CX)_m.
\end{align*}
We may rewrite this expression as 
\begin{align*}
Q(X) = X \cdot CX - \ddfrac{X \cdot CX - \mathsf{C}}{\max_{i=1}^n \left\{ (CX)_i \right\} - \mathsf{C}_\ell} Y \cdot CX + \sum_{m=1}^n \ddfrac{X \cdot CX - \mathsf{C}}{\mathsf{C}_u - (CX)_m } (CX)_m.
\end{align*}
The effect of \eqref{fgfgfg} is then that $Q(X) = X \cdot CX$, which guarantees that $\mathsf{P}$ increases. An analogous situation emerges for the growth transformation of the second primary barrier function.\\

\noindent
Assuming the current iterate $X$ is in the effective interior of the lower feasibility set, if $\hat{X'} = (T_{\mathsf{G}'} \circ T_{\mathsf{G}})(X)$ using $\alpha_h$ is also in the effective interior of the lower feasibility set, we keep $\hat{X'}$ as the next iterate. Otherwise, we use the bisection method on the learning rate $\alpha$ to find the value $\hat{\alpha}$ such that $\hat{X'}$ lands on the boundary of the effective interior. If $\hat{\alpha}$ is greater than the lower bound $\alpha_\ell$, we set $\alpha \leftarrow \hat{\alpha} / 2$ and invoke $(T_{\mathsf{G}'} \circ T_{\mathsf{G}})(X)$. Then $\hat{X'}$ is ensured to land in the effective interior. If $\hat{\alpha}$ is equal to or smaller than the lower bound $\alpha_\ell$, then we perform the update that yields $\hat{X'}$ and immediately after we invoke one of two {\em complementary mechanisms} (described below). Note that this process cannot guarantee in itself that $\mathsf{P}(\hat{X}) > \mathsf{P}(X)$ in the event that the bisection method on $\alpha$ has to be invoked. To ensure $\mathsf{P}(\hat{X}) > \mathsf{P}(X)$, as the bisection method unfolds, upon detection of an (intermediate) point $\tilde{X}$ such that $\mathsf{P}(\tilde{X}) \leq \mathsf{P}(X)$, we restart the bisection method: Although it is straightforward to update $\mathsf{C}$ using the previous update rule since $X$ remains fixed, as $\mathsf{C}$ changes, the point $\hat{X}$ where $T_{\mathsf{G'}}$ is invoked must also change since $\hat{X}$ is a function of parameter $\mathsf{C}$ of $T_{\mathsf{G}}$. Leaving the details for the pseudocode (Section \ref{pseudocode}), we note that the proof that this algorithm runs in polynomial time is an implication of terminating the halving scheme after a constant number of iterations and using \eqref{fgfgfg} in the update. Note finally that the correctness of our algorithm (and its polynomial running time) is insensitive to the choice of the lower bound $\alpha_\ell$. However, if $\alpha_\ell$ is large, the map $T_{\mathsf{G}'} \circ T_{\mathsf{G}}$ doesn't have ``breathing space'' in that it is forced to yield to one of the complementary mechasnisms in almost every iteration, which may affect the running-time performance. Let us now specify the aforementioned pair of complementary mechanisms in detail:

\subsubsection{The first complementary mechanism}

In the event that 
\begin{align*}
\mathsf{G}^*(X) = \ddfrac{X \cdot CX - \mathsf{C}}{\max_{i=1}^n \left\{ (CX)_i \right\} - \mathsf{C}_\ell} = \mathsf{G}^*_0
\end{align*}
(or approximately so) our algorithm subsequently solves the convex optimization problem
\begin{align*}
\mbox{ minimize } &RE(Y, X) - \alpha Y \cdot \bar{C}X\\
\mbox{ subject to } &Y \cdot CE_m = X \cdot CE_m \\
 &Y \in \mathbb{Y},
\end{align*}
where $\bar{C}$ is the corresponding operative matrix of $\mathsf{G}$ at $X$ and $E_m$ is a best response to $X$. Since $CE_m$ is a column of $C$, the column corresponding to pure strategy $m$, the constraint $Y \cdot CE_m = X \cdot CE_m \equiv c$ corresponds to a hyperplane in $\mathbb{R}^n$ that intersects the simplex at $X$. Therefore, if $X$ is an interior point, there exists a continuum of points that satisfy the constraint. If $Y$ is a solution to this optimization problem, then
\begin{align*}
\max_{i=1}^n \left\{ (CY)_i \right\} \geq Y \cdot CE_m = X \cdot CE_m = \max_{i=1}^n \left\{ (CX)_i \right\}.
\end{align*}
To prove that this method is effective, it remains to prove that $Y \cdot \bar{C} Y > X \cdot \bar{C} X$. Writing the KKT conditions for the previous problem, we obtain
\begin{align*}
1 + \ln \frac{Y(i)}{X(i)} - \alpha (\bar{C}X)_i + \lambda - \mu (CE_m)_i = 0 \quad i=1, \ldots, n
\end{align*}
Straight algebra gives that
\begin{align*}
Y(i) = X(i) \ddfrac{\exp\left\{ \alpha (\bar{C}X)_i + \mu (CE_m)_i \right\}}{\sum_{i=1}^n X(i) \exp\left\{ \alpha (\bar{C}X)_i + \mu (CE_m)_i \right\}} \quad i = 1, \ldots, n
\end{align*}
satisfies these conditions for all $\mu$ (and, therefore, for the optimal $\mu$). Lemma \ref{Baum_Eagon_extension_exponential_function} together with the aforementioned trick by Bomze imply then that 
\begin{align*}
Y \cdot \bar{C} Y + \frac{\mu}{\alpha} (CY)_m > X \cdot \bar{C} X + \frac{\mu}{\alpha} (CX)_m
\end{align*}
and, therefore, that 
\begin{align*}
Y \cdot \bar{C} Y > X \cdot \bar{C} X.
\end{align*} 
We denote the map obtained from the previous optimization problem as $T^{OPT}_{\mathsf{G}}$. We note that $T^{OPT}_{\mathsf{G}}$ ascends $\mathsf{G}^*$, which implies that the window of consecutive iterations required to invoke $T^{OPT}_{\mathsf{G}}$ has length {\em one}. This observation is important in computing the fixed points of our dynamical system. Note that by Berge's maximum theorem and the strict convexity of the objective function (cf. \citep[p.239]{Sundaram}) $Y$ is a continuous function of $\alpha$. There are a variety of methods for selecting $\alpha$: In our algorithm $\alpha$ is chosen such that in the next iteration $T_{\mathsf{G}'} \circ T_{\mathsf{G}}$ is invoked. To that end, it is sufficient that
\begin{align*}
\mathsf{G}^*(X) > \mathsf{G}_0^* \quad \mbox{ and } \quad \max_{i=1}^n \left\{ (CX)_i \right\} < \mathsf{C}_u - \epsilon.
\end{align*}

\subsubsection{Proof that the first complementary mechanism is polynomial}

Let us now show that the convex optimization problem in the first complementary mechanism, namely,
\begin{align*}
\mbox{ minimize } &RE(Y, X) - \alpha Y \cdot \bar{C}X\\
\mbox{ subject to } &Y \cdot CE_m = X \cdot CEm\\
  &Y \in \mathbb{X}(C)
\end{align*}
can be solved in polynomial time. To that end let us compute the dual. The Lagrangian is
\begin{align}
L(Y, \lambda, \mu) = RE(Y, X) - \alpha Y \cdot \bar{C}X + \lambda (\mathbf{1}^T Y - 1) + \mu (Y \cdot CE_m - X \cdot CEm),\label{langrangian2}
\end{align}
where we assume that the constraint $Y \geq 0$ is implicit. The dual function $L_D(\lambda)$ is obtained by minimizing the Lagrangian $L(Y, \lambda, \mu)$:
\begin{align*}
L_D(\lambda, \mu) = \inf_{Y \in \mathbb{R}^n} \{RE(Y, X) - \alpha Y \cdot \bar{C}X + \lambda (\mathbf{1}^T Y - 1) + \mu (Y \cdot CE_m - X \cdot CEm) \}.
\end{align*}
The Lagrangian is minimized when the gradient is zero and the gradient is zero when
\begin{align*}
1 + \ln \frac{Y(i)}{X(i)} - \alpha (\bar{C}X)_i + \lambda + \mu (CE_m)_i = 0.
\end{align*}
Solving for $Y(i)$ in the previous expression, we obtain
\begin{align}
Y(i) = X(i) \exp\{ - 1 - \lambda + \alpha (\bar{C}X)_i - \mu (CE_m)_i \} \quad i = 1, \ldots, n.\label{firm2}
\end{align}
Note that, by the previous expression, the constraint $Y \geq 0$ is automatically satisfied. Substituting now the previous expression for $Y$ in \eqref{langrangian2}, we obtain the dual function
\begin{align*}
L_D(\lambda, \mu) = &\sum_{i=1}^n X(i) \exp\{ - 1 - \lambda + \alpha (\bar{C}X)_i - \mu (CE_m)_i \} (- 1 - \lambda + \alpha (\bar{C}X)_i - \mu (CE_m)_i) -\\
  &- \alpha \sum_{i=1}^n X(i) \exp\{ - 1 - \lambda + \alpha (\bar{C}X)_i - \mu (CE_m)_i  \} (\bar{C}X)_i + \\
  &+ \lambda \left(\sum_{i=1}^n X(i) \exp\{ - 1 - \lambda + \alpha (\bar{C}X)_i - \mu (CE_m)_i\} - 1\right) \\
  &+ \mu \left( \sum_{i=1}^n X(i) \exp\{ - 1 - \lambda + \alpha (\bar{C}X)_i - \mu (CE_m)_i  \} (CE_m)_i - X \cdot CE_m \right)
\end{align*}
which simplifies to
\begin{align*}
L_D(\lambda, \mu) = - \sum_{i=1}^n X(i) \exp\{ - 1 - \lambda + \alpha (\bar{C}X)_i - \mu (CE_m)_i \} -\lambda - \mu X \cdot CE_m
\end{align*}
Since the dual function is concave, to find the optimal $\lambda$ and $\mu$ we simply need to set the derivative of $L_D(\lambda, \mu)$ (with respect to $\lambda$ and $\mu$) equal to $0$. To that end, we have
\begin{align*}
\frac{d L_D(\lambda, \mu)}{d \lambda} &= \sum_{i=1}^n X(i) \exp\{ - 1 - \lambda + \alpha (\bar{C}X)_i - \mu (CE_m)_i  \} - 1 = 0
\end{align*}
and
\begin{align}
\frac{d L_D(\lambda, \mu)}{d \mu} &= \sum_{i=1}^n X(i) \exp\{ - 1 - \lambda + \alpha (\bar{C}X)_i - \mu (CE_m)_i \} (CE_m)_i  - X \cdot CE_m = 0\label{xyz}
\end{align}
which implies
\begin{align*}
\exp\{ - 1 - \lambda \} = \frac{1}{\sum_{i=1}^n X(i) \exp\{ \alpha (\bar{C}X)_i - \mu (CE_m)_i \}}.
\end{align*}
Substituting in \eqref{firm2} we obtain
\begin{align*}
Y(i) = X(i) \frac{ \exp\{\alpha (\bar{C}X)_i - \mu (CE_m)_i\}}{\sum_{i=1}^n X(i) \exp\{ \alpha (\bar{C}X)_i - \mu (CE_m)_i \}} \quad i=1, \ldots, n
\end{align*}
as claimed. Substituting in \eqref{xyz} we obtain
\begin{align}
\ddfrac{\sum_{i=1}^n X(i) \exp\{ \alpha (\bar{C}X)_i - \mu (CE_m)_i \} (CE_m)_i}{\sum_{j=1}^n X(j) \exp\{ \alpha (\bar{C}X)_j - \mu (CE_m)_j \}}  - X \cdot CE_m = 0\label{bcd}
\end{align}  
We are looking for the value of $\mu$ that solves this equation. To that end, by the concavity of the dual function, we have
\begin{align*}
\frac{d^2 L_D(\lambda, \mu)}{d \lambda^2} &= - \sum_{i=1}^n X(i) \exp\{ - 1 - \lambda + \alpha (\bar{C}X)_i - \mu (CE_m)_i  \} < 0
\end{align*}  
\begin{align*}
\frac{d^2 L_D(\lambda, \mu)}{d \mu^2} &= - \sum_{i=1}^n X(i) \exp\{ - 1 - \lambda + \alpha (\bar{C}X)_i - \mu (CE_m)_i \} ((CE_m)_i)^2 < 0
\end{align*}
\begin{align*}
\frac{d^2 L_D(\lambda, \mu)}{d \lambda d \mu} &= - \sum_{i=1}^n X(i) \exp\{ - 1 - \lambda + \alpha (\bar{C}X)_i - \mu (CE_m)_i  \} (CE_m)_i
\end{align*}
\begin{align*}
\frac{d^2 L_D(\lambda, \mu)}{d \lambda d \mu} &= - \sum_{i=1}^n X(i) \exp\{ - 1 - \lambda + \alpha (\bar{C}X)_i - \mu (CE_m)_i \} (CE_m)_i
\end{align*}
and, therefore,
\begin{align*}
\left( \sum_{i=1}^n X(i) \exp\{ - 1 - \lambda + \alpha (\bar{C}X)_i - \mu (CE_m)_i  \} \right) \left( \sum_{i=1}^n X(i) \exp\{ - 1 - \lambda + \alpha (\bar{C}X)_i - \mu (CE_m)_i \} ((CE_m)_i)^2 \right) -
\end{align*}
\begin{align*}
- \left( \sum_{i=1}^n X(i) \exp\{ - 1 - \lambda + \alpha (\bar{C}X)_i - \mu (CE_m)_i \} (CE_m)_i \right)^2 > 0
\end{align*}
Letting
\begin{align*}
f(\mu) \equiv \ddfrac{\sum_{i=1}^n X(i) \exp\{ \alpha (\bar{C}X)_i - \mu (CE_m)_i \} (CE_m)_i}{\sum_{j=1}^n X(j) \exp\{ \alpha (\bar{C}X)_j - \mu (CE_m)_j \}}
\end{align*}
and taking the derivative with respect to $\mu$ we obtain in the numerator the negative of the previous positive expression. Therefore $f$ is strictly monotonically decreasing and we can apply the bisection method to solve \eqref{bcd}. By elementary numerical analysis the bisection method halves the error in the every iteration and, therefore, its complexity is linear in the number of precision digits.

\subsubsection{The second complementary mechanism}

In the event that
\begin{align*}
\max_{i=1}^n \left\{ (CX)_i \right\} = \mathsf{C}_u - \epsilon,
\end{align*}
(or approximately so) our algorithm iterates $J$. $J$ ascends $\mathsf{P}$ (cf. Proposition \ref{Baum_Eagon}) and, as shown in the sequel (cf. Lemma \ref{fundamental_lemma}), even if it escapes the lower feasibility set (to enter the upper feasibility set) it will return to the lower feasibility set at a strategy $X$, where
\begin{align*}
\mathsf{G}^*(X) \geq \mathsf{G}_0^* \quad \mbox{ and } \quad \max_{i=1}^n \left\{ (CX)_i \right\} \leq \mathsf{C}_u - \epsilon
\end{align*}
by an appropriate selection of an intermediate point in the secant line connecting the last iterate with the second-to-last iterate (\cite{Baum-Sell} show that all such intermediate points ascend the potential function). We denote the map that selects an intermediate point in the secant line by $J_{1/2}$. There are a variety of methods for selecting the intermediate point in the last iteration: Our algorithm selects a point such that $T_{\mathsf{G}'} \circ T_{\mathsf{G}}$ is invoked next. To that end, it is sufficient that the iterate lands in the effective interior, that is,
\begin{align*}
\mathsf{G}^*(X) > \mathsf{G}_0^* \quad \mbox{ and } \quad \max_{i=1}^n \left\{ (CX)_i \right\} < \mathsf{C}_u - \epsilon.
\end{align*}

\subsection{Fixed points of our dynamical system}
\label{fixed_points}

\begin{lemma}
\label{fixed_points_lemma_general}
Let $T_{\mathsf{G}} : \mathbb{Y} \rightarrow \mathbb{Y}$ and $\hat{J} : \mathbb{Y} \rightarrow \mathbb{Y}$ be such that $\hat{J}$ is invertible and
\begin{align*}
\hat{J}(X) \neq X \bigwedge T_{\mathsf{G}}(X) \neq X \Rightarrow (\hat{J} \circ T_{\mathsf{G}})(X) \neq X.
\end{align*} 
Then the fixed points of $T_{\mathsf{G}}$ are a subset of the fixed points of $\hat{J}$. 
\end{lemma}

\begin{proof}
Our first claim is that
\begin{align*}
\hat{J}(X) = X \bigwedge T_{\mathsf{G}}(X) = X \Leftrightarrow (\hat{J} \circ T_{\mathsf{G}})(X) = X.
\end{align*}
The forward direction, that is, that 
\begin{align*}
\hat{J}(X) = X \bigwedge T_{\mathsf{G}}(X) = X \Rightarrow (\hat{J} \circ T_{\mathsf{G}})(X) = X.
\end{align*}
is straightforward. The prove our claim it suffices to show that
\begin{align*}
\neg \left( \hat{J}(X) = X \bigwedge T_{\mathsf{G}}(X) = X \right) \Rightarrow (\hat{J} \circ T_{\mathsf{G}})(X) \neq X.
\end{align*}
This breaks down to showing that
\begin{align*}
\hat{J}(X) \neq X \bigwedge T_{\mathsf{G}}(X) = X \Rightarrow (\hat{J} \circ T_{\mathsf{G}})(X) \neq X,
\end{align*}
\begin{align*}
\hat{J}(X) = X \bigwedge T_{\mathsf{G}}(X) \neq X \Rightarrow (\hat{J} \circ T_{\mathsf{G}})(X) \neq X,
\end{align*}
and 
\begin{align*}
\hat{J}(X) \neq X \bigwedge T_{\mathsf{G}}(X) \neq X \Rightarrow (\hat{J} \circ T_{\mathsf{G}})(X) \neq X.
\end{align*}
But, by the definition of the synthesis of two maps,
\begin{align*}
(\hat{J} \circ T_{\mathsf{G}})(X) \equiv \hat{J}(T_{\mathsf{G}}(X))
\end{align*}
and it is a matter of straight algebra to verify that all three cases go through (the third case being identical to the assumption in the statement of the lemma). Therefore,
\begin{align*}
\{ X | (\hat{J} \circ T_{\mathsf{G}})(X) = X \} = \{ X | \hat{J}(X) = X \} \cap \{ X | T_{\mathsf{G}}(X) = X \}.
\end{align*}
That is, the intersection of the fixed points of $\hat{J}$ and those of $T_{\mathsf{G}}$ equals the set of fixed points of $\hat{J} \circ T_{\mathsf{G}}$.\\ 

We would like to show that
\begin{align*}
\{ X | T_{\mathsf{G}}(X) = X \} \subset \{ X | \hat{J}(X) = X \}.
\end{align*}
By the first part of the proof it suffices to show that
\begin{align*}
\{ X | T_{\mathsf{G}}(X) = X \} \subset \{ X | (\hat{J} \circ T_{\mathsf{G}})(X) = X \}.
\end{align*}
Let $X$ be such that
\begin{align*}
(\hat{J} \circ T_{\mathsf{G}})(X) \neq X.
\end{align*}
To prove the lemma, it suffices to show that
\begin{align*}
T_{\mathsf{G}}(X) \neq X.
\end{align*}
Let $\hat{J}^{-1}(X)$ denote the inverse of $X$ under $\hat{J}$.\footnote{In a typical application of this lemma, $\hat{J}$ is the discrete-time replicator dynamic $J$. Note that $J$ is a diffeormorphism \citep[Theorem 4]{LosertAkin} (Theorem 4 in that paper assumes $C > 0$) and, therefore, invertible.} If $X$ is not a fixed point of $\hat{J}$, we obtain
\begin{align*}
(\hat{J}^{-1} \circ \hat{J} \circ T_{\mathsf{G}})(X) \neq \hat{J}^{-1}(X) \neq X,
\end{align*}
where the last inequality follows by $\hat{J}$ and $\hat{J}^{-1}$ having identical fixed points since
\begin{align*}
\hat{J}(X) = X \Leftrightarrow (\hat{J}^{-1} \circ \hat{J})(X) = \hat{J}^{-1}(X) \Leftrightarrow X = \hat{J}^{-1}(X),
\end{align*}
which implies
\begin{align*}
T_{\mathsf{G}}(X) \neq X
\end{align*}
as claimed. If $X$ is a fixed point of $\hat{J}$, we obtain $\hat{J}^{-1}(X) = X$ and, therefore, that
\begin{align*}
(\hat{J}^{-1} \circ \hat{J} \circ T_{\mathsf{G}})(X) \neq X
\end{align*}
which implies
\begin{align*}
T_{\mathsf{G}}(X) \neq X
\end{align*}
as claimed. This completes the proof.
\end{proof}

\begin{lemma}
\label{fixed_points_lemma}
Suppose $T_{\mathsf{G}}$ / $T_{\mathsf{G'}}$ / $T_{\mathsf{G}}^{OPT}$ monotonically ascend the potential function $\mathsf{P}(X)$. Then the set of fixed points of $T_{\mathsf{G}}$ / $T_{\mathsf{G'}}$ / $T_{\mathsf{G}}^{OPT}$ is a subset of the set of fixed points of the replicator dynamic.
\end{lemma}

\begin{proof}
Let $J : \mathbb{Y} \rightarrow \mathbb{Y}$ be the discrete-time replicator dynamic, that is,
\begin{align*}
J(X)_i = X(i) \ddfrac{(CX)_i}{X \cdot CX} \quad i =1, \ldots, n,
\end{align*}
where $C > 0$. $J$ is invertible as implied by \citep[Theorem 4]{LosertAkin}. Consider the synthesis $J \circ T_{\mathsf{G}}$ of $J$ and $T_{\mathsf{G}}$. Both $J$ and $T_{\mathsf{G}}$ monotonically ascend $\mathsf{P}(X)$. This implies that
\begin{align*}
J(X) \neq X \bigwedge T_{\mathsf{G}}(X) \neq X \Rightarrow (J \circ T_{\mathsf{G}})(X) \neq X.
\end{align*} 
 Lemma \ref{fixed_points_lemma_general} completes the proof. The proof for $T_{\mathsf{G'}}$ and $T_{\mathsf{G}}^{OPT}$ is analogous.
\end{proof}

\begin{theorem}
\label{fundamental_theorem_on_fixed_points}
Suppose $T_{\mathsf{G}}$ and $T_{\mathsf{G'}}$ monotonically ascend the potential function $\mathsf{P}(X)$. Then the fixed points of $T_{\mathsf{G}} \circ T_{\mathsf{G}'}$ are pure strategies and uniform equalizers.
\end{theorem}

\begin{proof}
Referring back to \eqref{avc}, the fixed points of $T_{\mathsf{G}}$ satisfy $\forall i,j \in \mathcal{C}(X)$ that
\begin{align*}
(CX)_i - \mathsf{G}(X) \left( \prod_{\ell=1}^n \left( \mathsf{C}_u - (CX)_\ell \right) \right) (CY)_i +
\end{align*}
\begin{align*}
+ G(X) \left( \max_{i=1}^n \left\{ (CX)_i \right\} - \mathsf{C}_\ell \right) \sum_{m=1}^n \left( \prod_{\ell =1, \ell \neq m}^n \left( \mathsf{C}_u - (CX)_\ell \right) \right) C_{im} =
\end{align*}
\begin{align*}
= (CX)_j - \mathsf{G}(X) \left( \prod_{\ell=1}^n \left( \mathsf{C}_u - (CX)_\ell \right) \right) (CY)_j +
\end{align*}
\begin{align*}
+ G(X) \left( \max_{i=1}^n \left\{ (CX)_i \right\} - \mathsf{C}_\ell \right) \sum_{m=1}^n \left( \prod_{\ell =1, \ell \neq m}^n \left( \mathsf{C}_u - (CX)_\ell \right) \right) C_{jm}
\end{align*}
By Lemma \ref{fixed_points_lemma}, they also satisfy
\begin{align*}
- \mathsf{G}(X) \left( \prod_{\ell=1}^n \left( \mathsf{C}_u - (CX)_\ell \right) \right) (CY)_i +
\end{align*}
\begin{align*}
+ G(X) \left( \max_{i=1}^n \left\{ (CX)_i \right\} - \mathsf{C}_\ell \right) \sum_{m=1}^n \left( \prod_{\ell =1, \ell \neq m}^n \left( \mathsf{C}_u - (CX)_\ell \right) \right) C_{im} =
\end{align*}
\begin{align*}
= - \mathsf{G}(X) \left( \prod_{\ell=1}^n \left( \mathsf{C}_u - (CX)_\ell \right) \right) (CY)_j +
\end{align*}
\begin{align*}
+ G(X) \left( \max_{i=1}^n \left\{ (CX)_i \right\} - \mathsf{C}_\ell \right) \sum_{m=1}^n \left( \prod_{\ell =1, \ell \neq m}^n \left( \mathsf{C}_u - (CX)_\ell \right) \right) C_{jm}.
\end{align*}
The fixed points of $T_{\mathsf{G'}}$ satisfy $\forall i,j \in \mathcal{C}(X)$ that
\begin{align*}
\left( \frac{1}{2} X \cdot X \right) (CX)_i + (X \cdot CX - \mathsf{C}) X(i) - \mathsf{G}(X) \left( \prod_{\ell=1}^n \left( \mathsf{C}_u - (CX)_\ell \right) \right) (CY)_i +
\end{align*}
\begin{align*}
+ G(X) \left( \max_{i=1}^n \left\{ (CX)_i \right\} - \mathsf{C}_\ell \right) \sum_{m=1}^n \left( \prod_{\ell =1, \ell \neq m}^n \left( \mathsf{C}_u - (CX)_\ell \right) \right) C_{im} =
\end{align*}
\begin{align*}
= \left( \frac{1}{2} X \cdot X \right)(CX)_j + (X \cdot CX - \mathsf{C}) X(j) - \mathsf{G}(X) \left( \prod_{\ell=1}^n \left( \mathsf{C}_u - (CX)_\ell \right) \right) (CY)_j +
\end{align*}
\begin{align*}
+ G(X) \left( \max_{i=1}^n \left\{ (CX)_i \right\} - \mathsf{C}_\ell \right) \sum_{m=1}^n \left( \prod_{\ell =1, \ell \neq m}^n \left( \mathsf{C}_u - (CX)_\ell \right) \right) C_{jm}
\end{align*}
By Lemma \ref{fixed_points_lemma}, they also satisfy
\begin{align*}
(X \cdot CX - \mathsf{C}) X(i) - \mathsf{G}(X) \left( \prod_{\ell=1}^n \left( \mathsf{C}_u - (CX)_\ell \right) \right) (CY)_i +
\end{align*}
\begin{align*}
+ G(X) \left( \max_{i=1}^n \left\{ (CX)_i \right\} - \mathsf{C}_\ell \right) \sum_{m=1}^n \left( \prod_{\ell =1, \ell \neq m}^n \left( \mathsf{C}_u - (CX)_\ell \right) \right) C_{im} =
\end{align*}
\begin{align*}
= (X \cdot CX - \mathsf{C}) X(j) - \mathsf{G}(X) \left( \prod_{\ell=1}^n \left( \mathsf{C}_u - (CX)_\ell \right) \right) (CY)_j +
\end{align*}
\begin{align*}
+ G(X) \left( \max_{i=1}^n \left\{ (CX)_i \right\} - \mathsf{C}_\ell \right) \sum_{m=1}^n \left( \prod_{\ell =1, \ell \neq m}^n \left( \mathsf{C}_u - (CX)_\ell \right) \right) C_{jm}
\end{align*}
By the first part of the proof of Lemma \ref{fixed_points_lemma_general}, the fixed points of $T_{\mathsf{G}'} \circ T_{\mathsf{G}}$ are the intersection of the set of fixed points of $T_{\mathsf{G}}$ and the set of fixed points of  $T_{\mathsf{G'}}$. Therefore, the fixed points of $T_{\mathsf{G}'} \circ T_{\mathsf{G}}$ satisfy
\begin{align*}
(X \cdot CX - \mathsf{C}) X(i) = (X \cdot CX - \mathsf{C}) X(j)
\end{align*}
and cancelling the factors, we obtain the lemma.
\end{proof}

\begin{theorem}
\label{overall_fixed_points}
The fixed points of the primary dynamical system are pure strategies and uniform equalizers.
\end{theorem}

\begin{proof}
The primary dynamical system is a sequence each element of which is $T_{\mathsf{G}'} \circ T_{\mathsf{G}}$, $T_{\mathsf{G}}^{OPT}$, $J$, or $J_{1/2}$. (Note that all such elements ascend the potential function $\mathsf{P}$ and that by Lemma \ref{fixed_points_lemma_general} the fixed points of $J_{1/2}$ are fixed points of the replicator dynamic.) An example of a window of this sequence is:
\begin{align*}
\cdots \quad T_{\mathsf{G}'} \circ T_{\mathsf{G}} \quad T_{\mathsf{G}'} \circ T_{\mathsf{G}} \quad T_{\mathsf{G}}^{OPT} \quad T_{\mathsf{G}'} \circ T_{\mathsf{G}} \quad J \quad J \quad J_{1/2} \quad T_{\mathsf{G}'} \circ T_{\mathsf{G}} \quad \cdots
\end{align*}
and another example is
\begin{align*}
\cdots \quad T_{\mathsf{G}'} \circ T_{\mathsf{G}} \quad T_{\mathsf{G}'} \circ T_{\mathsf{G}} \quad T_{\mathsf{G}}^{OPT} \quad T_{\mathsf{G}'} \circ T_{\mathsf{G}} \quad J_{1/2} \quad T_{\mathsf{G}'} \circ T_{\mathsf{G}} \quad \cdots
\end{align*}
Our algorithm ensures that $T_{\mathsf{G}}^{OPT}$ and $J$ (or $J_{1/2}$) are necessarily separated by a window of $T_{\mathsf{G}'} \circ T_{\mathsf{G}}$ and such that, following an element equal to $T_{\mathsf{G}'} \circ T_{\mathsf{G}}$, the number of elements that are equal to either $T_{\mathsf{G}}^{OPT}$ or $J / J_{1/2}$ are finite. This implies that the fixed points of the primary dynamical system are pure strategies and uniform equalizers since the fixed points of a window
\begin{align*}
T_{\mathsf{G}'} \circ T_{\mathsf{G}} \quad T_{\mathsf{G}}^{OPT} \quad \cdots \quad T_{\mathsf{G}'} \circ T_{\mathsf{G}} \quad T_{\mathsf{G}}^{OPT}
\end{align*}
are the intersection of the fixed points of $T_{\mathsf{G}'} \circ T_{\mathsf{G}}$ and $T_{\mathsf{G}}^{OPT}$ (cf. Lemma \ref{fixed_points_lemma}) and the fixed points of a window
\begin{align*}
T_{\mathsf{G}'} \circ T_{\mathsf{G}} \quad J \quad \cdots \quad J
\end{align*}
or 
\begin{align*}
T_{\mathsf{G}'} \circ T_{\mathsf{G}} \quad J \quad \cdots \quad J \quad J_{1/2}
\end{align*}
or 
\begin{align*}
T_{\mathsf{G}'} \circ T_{\mathsf{G}} \quad J_{1/2}
\end{align*}
are the intersection of the fixed points of $T_{\mathsf{G}'} \circ T_{\mathsf{G}}$ and $J / J_{1/2}$ by the first part of the proof of Lemma \ref{fixed_points_lemma_general}.
\end{proof}


\subsection{Leapfrogging non-equilibrium fixed points and a fundamental property}

\begin{definition}
We say that the probability vectors $p$ and $q$ in $\mathbb{R}^n$ have the same ranking if 
\begin{align*}
\forall i, j \in \{1, \ldots, n\}: p(i) \geq p(j) \mbox{ if and only if } q(i) \geq q(j).
\end{align*}
\end{definition}

\begin{definition}
The probability sector of a fixed point, say $X^*$, is the set of all strategies that have the same ranking as $X^*$.
\end{definition}

The iterates of our dynamical system may converge to a non-equilibrium fixed point (for example, a clique) in the effective interior of the lower feasibility set unless our algorithm takes action to prevent this possibility. Note that on the event of convergence to a fixed point, the iterates enter and forever remain in its probability sector. To avoid such undesirable convergence, we rest on a property of non-equilibrium fixed points, namely, that they are necessarily ``interior'' points of the lower feasibility set, in that parameters can be configured such that they lie strictly below the upper feasibility set and strictly above the infeasibility set. This is shown in the following lemmas:

\begin{lemma}
\label{convexityy_lemma}
Let $C$ be an arbitrary square payoff matrix. Then the function $F : \mathbb{X}(C) \rightarrow \mathbb{R}$, where $F(X) = (CX)_{\max}$, is convex.
\end{lemma}

\begin{proof}
We may write $F$ as
\begin{align*}
F(X) = \max_Y \left\{ Y \cdot CX \right\}.
\end{align*}
By a basic property of the maximum function, we obtain for all $X \neq X'$ where $X, X' \in \mathbb{X}(C)$,
\begin{align*}
F((1-\epsilon)X + \epsilon X') = \max_Y \left\{ Y \cdot C((1-\epsilon)X + \epsilon X') \right\} \leq (1-\epsilon) \max_Y \left\{ Y \cdot CX \right\} + \epsilon \max_Y \left\{ Y \cdot CX' \right\}.
\end{align*}
Thus, $F$ is convex as claimed.
\end{proof}

\begin{lemma}
\label{Bertsekas_based_lemma}
Let $C$ be a Nisan-Bomze payoff matrix corresponding to a complete graph. Then the equilibrium, say $X^*$, of $C$ is a global minimizer of $F : \mathbb{X}(C) \rightarrow \mathbb{R}$ where $F(X) = (CX)_{\max}$.
\end{lemma}

\begin{proof}
To show that $X^*$ is a global minimizer of $F$, we will show a stronger property that $X^*$ is global minimizer of $(CX)_{\max}$ over all $X$ in the hyperplane
\begin{align*}
X(0) + X(1) + \cdots + X(n) = 1.
\end{align*}
Following \citep[pp. 331-332] {Bertsekas}, a necessary condition for $X^*$ to be a local minimizer of $(CX)_{\max}$ over the previous hyperplane is that there exists a vector $\mu$ such that $\mu \geq 0$ and $\sum \mu_i = 1$ and a scalar $\lambda$ such that
\begin{align*}
C \mu + \lambda \mathbf{1} = 0.
\end{align*}
Letting $\mu = X^*$ and $\lambda = - (1-1/(2n))$ satisfies these conditions. Thus, since, following the proof of Lemma \ref{convexityy_lemma}, $(CX)_{\max}$ is a convex function the aforementioned necessary condition is also sufficient, which implies that $X^*$ is a global minimizer of $(CX)_{\max}$ over the previous hyperplane and, therefore, also of $F$. This completes the proof of the lemma.
\end{proof}

The previous lemma implies by the continuity of $F(X) = (CX)_{\max}$ that there exists a neighborhood $O$ of $X^*$ such that
\begin{align*}
\forall X \in O : \max_{i \in \mathcal{C}(X^*)} \left\{ (CX)_i \right\} > \max_{i=1}^n \left\{ (CX^*)_i \right\}
\end{align*}
which further implies that
\begin{align*}
\forall X \in O : \max_{i = 1}^n \left\{ (CX)_i \right\} > \max_{i=1}^n \left\{ (CX^*)_i \right\}.
\end{align*}
Therefore, we may only consider two possibilities: 

\begin{itemize}

\item The first possibility is that the undesirable fixed point is on the boundary of the upper or lower boundary of the effective interior. On such event, once iterates are in the probability sector of the corresponding fixed point (an event which can be readily detected using the previous definitions by ranking the elements of corresponding iterates and checking if the top iterates correspond to a fixed point), we may temporarily increase $\epsilon$ or $\mathsf{G}^*_0$ until the potential value of the current iterate exceeds the potential value of the corresponding fixed point at which point we may restore the corresponding parameter to its original value. 

\begin{figure}[tb]
\centering
\includegraphics[width=14cm]{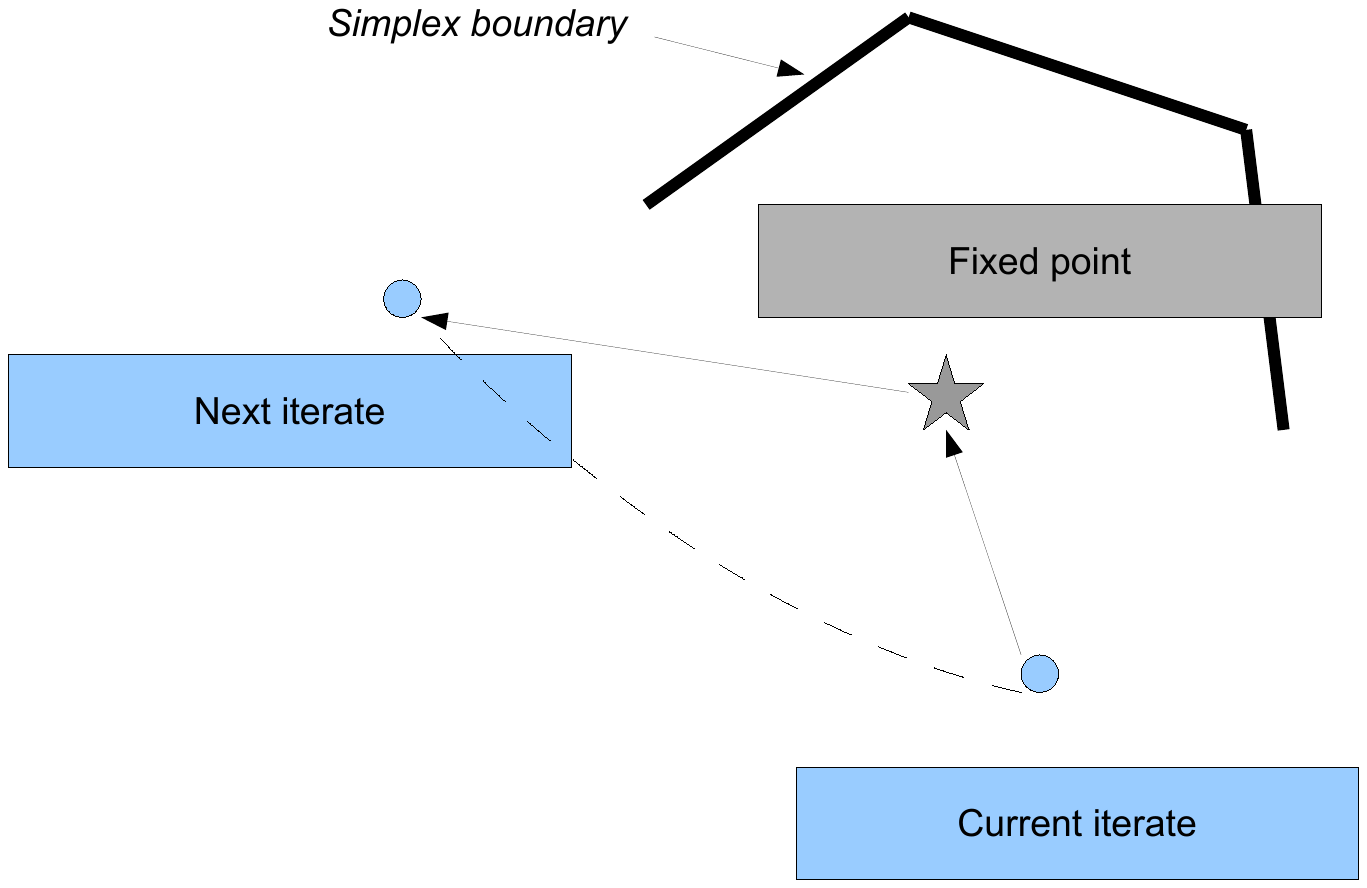}
\caption{\label{leapfrogging}
The leapfrogging mechanism in the typical case.}
\end{figure}

\item The second possibility (see Figure \ref{leapfrogging}), which is the typical case, is that the undesirable fixed point is not on the boundary of neither the upper boundary nor the lower boundary of the effective interior. On such event, we proceed as follows: If an iterate in $\mathbb{F_L}$, say $X$, enters the probability sector of a non-equilibrium fixed point (which is a uniform equalizer) in $\mathbb{F_L}$, say $X^*$, and $\mathsf{P}(X^*) > \mathsf{P}(X)$, then the next iterate, say $\hat{X}$, is selected such that $\hat{X} \cdot C\hat{X} \geq X^* \cdot CX^*$. Computing such a $\hat{X}$ is simple: Since $X^*$ is not an equilibrium, the set $\{ Z \in \mathbb{Y} | (Z-X^*) \cdot CX^* > 0 \}$ (of probability vectors $Z$ corresponding to directional derivatives of the potential function at $X^*$ in the direction from $X^*$ to $Z$) is the intersection of a half-space (cutting through the interior of $\mathbb{Y}$) with $\mathbb{Y}$. Selecting any interior $Z$ in the intersection of this set with $\mathbb{Y}$, for example, by solving the convex optimization problem (wth a self-concordant objective function)
\begin{align*}
\mbox{maximize} \quad &\sum_{i=1}^n \ln(Z(i)) + \ln \left( (Z-X^*) \cdot CX^* \right)\\
\mbox{subject to} \quad &Z \in \mathbb{Y}
\end{align*}
whose optimal solution is the {\em analytic center} of the set $\{ Z \in \mathbb{Y} | (Z-X^*) \cdot CX^* \geq 0 \}$, which can be computed in polynomial time \citep{WAC} (see also \citep{ConvexOptimization, YeBook}),
and, then, using a halving scheme starting at $Z$ and iteratively approaching $X^*$ halving the distance between $Z$ and $X^*$ until the effective interior is reached and the potential function $\mathsf{P}$ increases, will yield a desirable (interior) $\hat{X}$. 

\end{itemize}

The next lemma aims to eliminate the possibility of convergence to an undesirable equilibrium point.

\begin{lemma}
\label{iloveyoustillevenmore}
Uniform equalizers correspond to carriers whose payoff matrices are scalar multiples of doubly stochastic matrices. Uniform equalizers that are not characteristic vectors of cliques are global minima of $\mathsf{P}$ (within their carrier).
\end{lemma}

\begin{proof}
Let us first prove that the carrier (of cardinality $m$) of a uniform equalizer corresponds to a carrier whose payoff matrix is scalar multiple of a doubly stochastic matrix. Denoting the uniform equalizer by $X^*$, and the payoff matrix corresponding to the carrier of $X^*$ by $\hat{C}$, observe that
\begin{align*}
\hat{C}X^* = c \mathbf{1}
\end{align*}
where $\mathbf{1}$ is a vector of ones. The previous equation implies
\begin{align*}
\hat{C} c^* \mathbf{1} = c \mathbf{1}
\end{align*}
which further implies
\begin{align*}
\left( \frac{c^*}{c} \right) \hat{C} \mathbf{1} = \mathbf{1},
\end{align*}
which proves our claim. Since the payoff matrix is a scalar multiple of a doubly stochastic matrix, call it $S$, it implies spherical symmetry of $\mathsf{P}$ over the corresponding tangent space. Let us prove this: We would like to show that 
\begin{align*}
\min \left\{ X \cdot SX \bigg| \frac{1}{2} \sum_{i=1}^m (X(i)-X^*(i))^2 = c \quad \sum_{i=1}^m X(i) = 1 \right\} =
\end{align*}
\begin{align*}
= \max \left\{ X \cdot SX \bigg| \frac{1}{2} \sum_{i=1}^m (X(i) - X^*(i))^2 = c \quad \sum_{i=1}^m X(i) = 1 \right\}
\end{align*}
or, equivalently, that
\begin{align*}
\forall X \in \mathbb{S} \equiv \left\{ X \in \mathbb{R}^m \bigg| \frac{1}{2} \sum_{i=1}^m (X(i) - X^*(i))^2 = c \quad \sum_{i=1}^m X(i) = 1 \right\} : X \cdot SX = \mbox{ constant}.
\end{align*}
We may write KKT conditions either for the minimization or the maximization problem. These read as follows for the minimization problem
\begin{align*}
- SX + \mu (X - X^*) + \lambda \mathbf{1} = 0
\end{align*}
and for the maximization problem they are
\begin{align*}
SX + \mu' (X - X^*) + \lambda' \mathbf{1} = 0
\end{align*}
which imply
\begin{align*}
- S^{\infty} X + \mu S^{\infty} (X - X^*) + \lambda \mathbf{1} = 0
\end{align*}
and 
\begin{align*}
S^{\infty} X + \mu' S^{\infty} (X - X^*) + \lambda' \mathbf{1} = 0,
\end{align*}
where
\begin{align*}
S^{\infty} = \lim_{k \rightarrow \infty} \left\{ S^k \right\}.
\end{align*}
The limit exists since $S$ is assumed to be a doubly stochastic matrix. Taking
\begin{align*}
- \mu S^{\infty} X^* + \lambda \mathbf{1} = 0,
\end{align*}
which is always possible since $X^*$ is a scalar multiple of $\mathbf{1}$, which implies, 
\begin{align*}
- \mu c \mathbf{1} + \lambda \mathbf{1} = 0,
\end{align*}
$\mu' = - \mu$, and $\lambda' = -\lambda$, we obtain that $X$ is a KKT point (for both minimization and maximization problems). Since $X$ is an arbitrary element of the sphere, every point on the sphere is a KKT point and, therefore, spherical symmetry follows: Considering an arc $a : [0, 1] \rightarrow \mathbb{S}$ on the sphere and denoting $X_{a[0]}$ and $X_{a[1]}$ its endpoints, if $X_{a[0]} \cdot S X_{a[0]} \neq X_{a[1]} \cdot S X_{a[1]}$, then there exists a point $t$ on the arc such that $X_{a[t]}$ is not a KKT point of the restricted problem on the arc, which contradicts that $X_{a[t]}$ is a KKT point on the sphere. Therefore, the equalizer $X^*$ is either a local maximum or a local minimum, which implies by the property of $X^*$ being an equalizer, that $X^*$ is either a global maximum or a global minimum. If the underlying graph is not a clique (complete graph), then there exists a maximum clique, which is an equilibrium. If $X^*$ is a global maximum, this leads to a contradiction (as $X^*$ is a GESS, which eliminates the possibility of the presence of other equilibria). Therefore, as claimed, $X^*$ is a global minimum of $\mathsf{P}$ over the carrier of $X^*$.
\end{proof}

\subsection{The secondary dynamical system}
\label{secondary_system}

If the Nisan parameter is equal to the clique number, upon an iterate of the primary dynamical system satisfying the condition $\mathsf{P} > \mathsf{C}_\ell + \epsilon$, the secondary dynamical system is activated in lieu of the primary. The secondary system is comprised of growth transformations for the barrier functions
\begin{align*}
\mathsf{G_2}(X) = \ddfrac{X \cdot CX - \mathsf{C}}{ \left( X \cdot CX  - \mathsf{C}_\ell \right) \prod_{i=1}^n \left( \mathsf{C}_u - (CX)_i \right) } \quad X \in \mathbb{F_L}
\end{align*}
and
\begin{align*}
\mathsf{G_2}'(X) = \ddfrac{\left( X \cdot CX - \mathsf{C} \right) \left( \frac{1}{2} X \cdot X \right)}{ \left( X \cdot CX  - \mathsf{C}_\ell \right) \prod_{i=1}^n \left( \mathsf{C}_u - (CX)_i \right) } \quad X \in \mathbb{F_L}
\end{align*}
with the same parameters as above. The secondary system operates in a fashion analogous to the primary barring the change in the definition of the lower boundary such that the condition $\mathsf{P} > \mathsf{C}_\ell + \epsilon$ is maintained throughout and the upper bound $\alpha_h$ on the learning rate, which is set equal to the solution of the equation
\begin{align*}
(\exp\{\alpha\} - 1)^2 = \frac{1}{4} d^2,
\end{align*}
where $d$ is the Euclidean distance between the midpoints of any pair of adjacent edges of the corresponding probability simplex. (Larger values for $d$ may also be considered.) This facilitates convergence to a unique maximum-clique equilibrium. Analogues of Lemmas \ref{gt_lemma} and \ref{123abcde} and Theorems \ref{fundamental_theorem_on_fixed_points} and \ref{overall_fixed_points} are obtained in a straightforward fashion. Note that the benefit of switching to the secondary system is to obviate invocations of map $T_{\mathsf{G}}^{OPT}$ and, therefore, simplify operation---our dynamical system converges to a maximum-clique equilibrium even if the secondary system is not invoked.

\subsection{The process of initialization of our dynamical system}
\label{initialization}

Ariadne looks for a maximum clique starting with a large value of the Nisan parameter $k$ (possibly the largest, however, upper bounds on the clique number, e.g., \citep{Pardalos-Philips}, can reduce the search space for the appropriate value of the Nisan parameter)  iteratively subtracting one from this parameter upon failure to compute a clique of size equal to $k$. If the Nisan parameter is equal to $\omega(G)$, Ariadne is guaranteed to compute a maximum clique. Appendix \ref{pseudocode} complements our subsequent higher level discussion on the workings of Ariadne in the form of pseudocode.

Initializing our dynamical system involves configuring six parameters, namely, the equilibrium approximation error $\epsilon$, $\mathsf{C}$, $\mathsf{C}_\ell$, $\mathsf{C}_u$, and the initial condition $X^0$. The equilibrium approximation error $\epsilon$ is set equal to 
\begin{align*}
\epsilon = \frac{\epsilon_a^2}{8} \mbox{ where } \epsilon_a = \frac{1}{2} \left( 1 + 1 - \frac{1}{2k} \right) - \frac{1}{2} \left( 1 + 1 - \frac{1}{2(k-1)} \right) = \frac{1}{4} \left( \frac{1}{k-1} - \frac{1}{k} \right)
\end{align*}
$\mathsf{C}$ is set equal to $C_{00}$, $\mathsf{C}_\ell$, the lower bound on the maximum payoff, is set equal to (cf. Lemma \ref{convergence})
\begin{align*}
\mathsf{C}_\ell = \frac{1}{2} \left( \frac{1}{2} \left( 1 + 1 - \frac{1}{2(k-1)} \right) + \frac{1}{2} \left( 1 + 1 - \frac{1}{2k} \right) \right)
\end{align*}
and, $\mathsf{C}_u$, the upper bound on the maximum payoff is set equal to
\begin{align*}
\mathsf{C}_u = \frac{1}{2} \left( \frac{1}{2} \left( 1 + 1 - \frac{1}{2k} \right) + \frac{1}{2} \left( 1 + 1 - \frac{1}{2(k+1)} \right) \right).
\end{align*}
The initial condition $X^0$ is set such that it is strictly upper feasible, such that
\begin{align*}
X^0 \cdot CX^0 \geq \frac{1}{2} \left( 1 + \frac{1}{2} \right),
\end{align*} 
and such that
\begin{align*}
\min_{i=1}^n \left\{ X^0(i) \right\} \geq 2^{-n}.
\end{align*}
To that end, we select a pair of pure strategies, say, $E_i$ and $E_j$, where $i \neq j$, we set $X^0(i)$ and $X^0(j)$ equal to $(1/2) (1 - (n-2) c)$ and $X^0(k), k \neq i, j$ equal to $c$. We are looking for $E_i$, $E_j$, and $c$ that satisfy the previous conditions. We may assume without loss of generality $i=1$ and $j=2$. 

\begin{lemma}
Let $(u, v)$ be an edge of $G$ and let us renumber vertices such that $u \equiv 1$ and $v \equiv 2$. Furthermore, let $C$ be a Nisan-Bomze payoff matrix, for example,
\begin{align*}
C = \frac{1}{2} \left( \left[ \begin{array}{cccccc}
1/2 & 1 & 0 & 1 & 0 & 1 \\
1 & 1/2 & 0 & 1 & 0 & 1 \\
0 & 0 & 1/2 & 1 & 1 & 0 \\
1 & 1 & 1 & 1/2 & 0 & 0 \\
0 & 0 & 1 & 0 & 1/2 & 1 \\
1 & 1 & 0 & 0 & 1 & 1/2 \\
\end{array} \right] +  
\left[ \begin{array}{cccccc}
1 & 1 & 1 & 1 & 1 & 1 \\
1 & 1 & 1 & 1 & 1 & 1 \\
1 & 1 & 1 & 1 & 1 & 1 \\
1 & 1 & 1 & 1 & 1 & 1 \\
1 & 1 & 1 & 1 & 1 & 1 \\
1 & 1 & 1 & 1 & 1 & 1 \\
\end{array} \right] \right),
\end{align*}
and let us write $C$ in block format
\begin{align*}
C \equiv
\left[ \begin{array}{cccccc}
C_{LL} & C_{LR} \\
C_{RL} & C_{RR} \\
\end{array} \right]
\end{align*}
where $C_{LL}$ is $2 \times 2$, $C_{LR}$ is $2 \times n-2$, $C_{RL} = C_{LR}^T$, and $C_{LL}$ is $(n-2) \times (n-2)$. Moreover, let $X^0$ be such that
\begin{align*}
X^0 = \left[ \begin{array}{cccccc}
(1/2) (1 - (n-2) c) \\
(1/2) (1 - (n-2) c) \\
c \\
\vdots \\
c \\
c \\
\end{array} \right]
\equiv \left[ \begin{array}{cccccc}
X^0_u \\
X^0_l \\
\end{array} \right].
\end{align*}
Then, for all
\begin{align*}
c \leq \frac{1}{n-2} \left( 1 - \sqrt{ \ddfrac{3}{ \mathbf{1}^T \cdot C_{LL} \mathbf{1} } } \right),
\end{align*}
we have that $X^0 \cdot CX^0 > (1/2) (1 + 1/2)$.
\end{lemma}

\begin{proof}
We have
\begin{align*}
\left[ \begin{array}{cccccc}
X^0_u & X^0_l \\
\end{array} \right]
\left[ \begin{array}{cccccc}
C_{LL} & C_{LR} \\
C_{RL} & C_{RR} \\
\end{array} \right]
\left[ \begin{array}{cccccc}
X^0_u \\
X^0_l \\
\end{array} \right] = 
\end{align*}
\begin{align*}
= \left[ \begin{array}{cccccc}
X^0_u & X^0_l \\
\end{array} \right]
\left[ \begin{array}{cccccc}
C_{LL} X^0_u + C_{LR} X^0_l \\
C_{RL} X^0_u + C_{RR} X^0_l \\
\end{array} \right]
\end{align*}
\begin{align*}
= X^0_u \cdot C_{LL} X^0_u + X^0_u C_{LR} X^0_l + X^0_l \cdot C_{RL} X^0_u + X^0_l C_{RR} X^0_l
\end{align*}
\begin{align*}
= X^0_u \cdot C_{LL} X^0_u + 2 X^0_u C_{LR} X^0_l + X^0_l C_{RR} X^0_l
\end{align*}
\begin{align*}
= ((1/2) (1 - (n-2) c))^2 \mathbf{1}^T \cdot C_{LL} \mathbf{1} + 2 X^0_u C_{LR} X^0_l + X^0_l C_{RR} X^0_l
\end{align*}
\begin{align*}
> ((1/2) (1 - (n-2) c))^2 \mathbf{1}^T \cdot C_{LL} \mathbf{1} + 2 X^0_u C_{LR} X^0_l
\end{align*}
\begin{align*}
= ((1/2) (1 - (n-2) c))^2 \mathbf{1}^T \cdot C_{LL} \mathbf{1} + 2 (1/2) (1 - (n-2) c) c \mathbf{1}^T C_{LR} \mathbf{1}
\end{align*}
\begin{align*}
> \left(\frac{1}{2} (1 - (n-2) c) \right)^2 \mathbf{1}^T \cdot C_{LL} \mathbf{1}
\end{align*}
We would like to find $c$ such that
\begin{align*}
&\left(\frac{1}{2} (1 - (n-2) c) \right)^2 \mathbf{1}^T \cdot C_{LL} \mathbf{1} \geq \frac{1}{2} \left(1 + \frac{1}{2} \right)\\
\Leftrightarrow &\left(\frac{1}{2} (1 - (n-2) c) \right)^2 \geq \frac{1 + 1/2}{2 \left( \mathbf{1}^T \cdot C_{LL} \mathbf{1}\right)}\\
\Leftrightarrow &\frac{1}{2} (1 - (n-2) c) \geq \sqrt{ \frac{1 + 1/2}{2 \left( \mathbf{1}^T \cdot C_{LL} \mathbf{1} \right)} }\\
\Leftrightarrow &1 - (n-2) c \geq \sqrt{ \frac{3}{\mathbf{1}^T \cdot C_{LL} \mathbf{1}}}\\
\Leftrightarrow &c \leq \frac{1}{n-2} \left( 1 - \sqrt{ \frac{3}{ \mathbf{1}^T \cdot C_{LL} \mathbf{1} } } \right)
\end{align*}
Any $c > 0$ that satisfies the previous inequality, implies $X^0 \cdot CX^0 > (1/2)(1+1/2)$ as claimed.
\end{proof}

The previous lemma suggests a simple algorithm to find $X^0$ that meets our specification, namely, starting at $X^0$ being equal to the uniform strategy and $X^*$ the fixed point on the edge, we iteratively set
\begin{align*}
X^0 \leftarrow \frac{1}{2} (X^0 + X^*)
\end{align*}
until our specification is satisfied. That this algorithm terminates satisfying our specification is implied by the next lemma:

\begin{lemma}
\label{aaaomega}
Suppose there exists a clique of size four and consider an edge in this clique. Renumbering the vertices accordingly, let $X^0$ be such that
\begin{align*}
X^0 = \left[ \begin{array}{cccccc}
(1/2) (1 - (n-2) c) \\
(1/2) (1 - (n-2) c) \\
c \\
\vdots \\
c \\
c \\
\end{array} \right].
\end{align*}
Then, if $c = 2^{-n}$, $X^0$ is strictly upper feasible.
\end{lemma}

\begin{proof}
Let $X^*$ denote the (uniform) equalizer of our edge. We then have
\begin{align*}
C X^0 \equiv C (X^* - c N)
\end{align*}
and, therefore,
\begin{align*}
(CX^0)_{\max} = \| CX^0 \|_{\infty} = \| CX^* - c CN \|_{\infty} \geq | \|CX^*\|_{\infty} - c \| CN \|_{\infty} | = | \|CX^*\|_{\infty} - c (n-2) |.
\end{align*}
Using strategy algebra (in the multiplication of $CX^*$), we obtain $(CX^*)_{\max} = \|CX^*\|_{\infty} = 1$, which implies $(CX^0)_{\max} > C_u$ and this completes the proof.
\end{proof}

We then iterate $J$ until the sequence of iterates enters (or cuts through) the effective interior of the lower feasibility set. Let us prove that the sequence of iterates is guaranteed to do so:

\begin{lemma}
\label{fundamental_lemma}
Starting at any interior strategy of the upper feasibility set or the upper boundary of the lower feasibility set, iterating $J$ is guaranteed to either enter the effective interior of the lower feasibility set or ``cut through'' the effective interior and enter either the lower boundary of the lower feasibility set or directly the infeasibility set. 
\end{lemma}

\begin{proof}
Proposition \ref{Baum_Eagon} implies that $J$ increases the potential function $\mathsf{P}$. \cite[Convergence Theorem 2]{LosertAkin} implies that the sequence of iterates generated by $J$ converges to a fixed point. Therefore, it suffices to show that such fixed point, call $X^*$, is an equilibrium of $C$. Our argument is similar to \citep[Proposition 3]{Pelillo}. Let us assume for the sake of contradiction that $X^*$ is a non-equilibrium fixed point. Then, there exists a pure strategy $j$ such that $(CX^*)_j > (CX^*)_i$ for all $i \in \mathcal{C}(X^*)$. Therefore, by continuity,  there exists a neighborhood $O$ of $X^*$ such that, for all $X \in O$, $(CX)_j > (CX)_i$. Thus, for a sufficiently large iteration count $K \geq 0$ and $\forall k \geq K$, the probability mass of strategy $j$ increases with respect to the probability mass of all strategies $i$, which contradicts $j$ not being in the carrier of $X^*$. Therefore, $J$ converges to an equilibrium fixed point of the replicator dynamic. But all such equilibria are located in the union of the lower feasibility set and the infeasibility set. This completes the proof.
\end{proof}

If the sequence of iterates of $J$ cuts through the effective interior of the lower feasibility set, we backtrack one iteration and invoke map $J_{1/2}$, which selects an intermediate point in the secant line between the last iterate and the second-to-last iterate such that the iterate we obtain becomes a strictly lower feasible strategy (subject to the constraints previously discussed in Section \ref{barrier_lower_bound}). Subsequently, we activate the (primary) dynamical system that keeps the iterates inside the lower feasibility set (barring excursions). Since $X^0 \cdot CX^0 > (1/2) (1+1/2)$ and since $J$ monotonically increases the potential value, the potential value of the first iterate inside the lower feasibility set is $> (1/2)(1 + 1/2)$. Lemma \ref{iloveyoustillevenmore} eliminates the possibility of convergence to a non-clique equalizer.

\subsection{Asymptotic convergence to a maximum-clique equilibrium}
\label{convergence_to_fixed_point}

\begin{proposition}[\citep{LosertAkin}]
\label{LosertAkinProposition}
Suppose a discrete time dynamical system obtained by iterating a continuous map $F : \Delta \rightarrow \Delta$ admits a Lyapunov function $G : \Delta \rightarrow \mathbb{R}$, i.e., $G(F(p)) \geq F(p)$ with equality at $p$ only when $p$ is an equilibrium. The limit point set $\Omega$ of an orbit $\{ p(t) \}$ is then a compact, connected set consisting entirely of equilibria and upon which $G$ is constant.
\end{proposition}

\begin{lemma}
\label{convergence}
If the Nisan parameter is equal to the clique number $\omega(G)$, the sequence $\{ X^k \}$ of iterates Ariadne generates converges to a maximum-clique equilibrium.
\end{lemma}

\begin{proof}
\cite{LosertAkin} in their Proposition 1 show that under the assumption the Lyapunov function $G$ is continuous and strictly (monotonically) increasing every limit point of an orbit under $F$ is a fixed point of $F$ upon which $G$ is constant (even if $F$ is discontinuous as is the case for our primary dynamical system). Our Lyapunov function, $\mathsf{P}$, is strictly monotonically increasing and in virtue of Lemma \ref{iloveyoustillevenmore} and the assumption that the initial condition $X^0$ is such that $X^0 \cdot CX^0 > (1/2)(1+1/2)$ (since $(1/2)(1+1/2)$ is the value of the potential function at a pure strategy and Lemma \ref{iloveyoustillevenmore} implies that the potential value of a uniform equalizer is lower than $(1/2)(1+1/2)$), the set of maximum clique equilibria are the unique attractive fixed points of our (primary) dynamical system upon which $\mathsf{P}$ assumes the value $C_{00}$. Therefore, had Ariadne been such that the secondary system were not activated (upon $\mathsf{P} > \mathsf{C}_\ell + \epsilon$), every limit point of the sequence of iterates would have been a maximum-clique equilibrium. The goal of activating the secondary system, in lieu of the primary, is to ensure convergence to a maximum clique equilibrium. Once $\mathsf{P}$ is sufficiently close to $C_{00}$, the secondary dynamical system is activated, giving rise to a dynamical system that retains the property that every limit point of the sequence of iterates is a maximum-clique equilibrium. However, by the upper bound we impose on the learning rate, since maximum-clique equilibria are the only fixed points such that $\mathsf{P} > \mathsf{C}_\ell + \epsilon$ and they are also isolated fixed points, the Euclidean distance between any pair of maximum-clique equilibria is greater than $d$ and inequality \eqref{vgood_inequality} implies that the sequence of iterates converges to a single limit point and, therefore, has a limit, which is a maximum-clique equilibrium. (We note that such restriction can also be imposed on the primary system obviating the need to invoke the secondary.) This completes the proof.
\end{proof}

\section{Ariadne: The secondary sequence of iterates}
\label{Algorithm2}

For any given value of the Nisan parameter, Ariadne either computes an equilibrium of $C$ or detects that the equilibrium approximation bound obtained in the next section has been violated. This process does not apply to the iterates of our dynamical system directly but rather to the empirical average of a sequence of {\em approximate multipliers}, that is, mixed strategies that are obtained by transforming the iterates according to the following process: Let us denote by $X$ the current iterate and by $\hat{X}$ the next iterate. An {\em exact multiplier} $Y$ of $X$ is a strategy Y in $\mathbb{Y}$ such that
\begin{align*}
\hat{X}(i) = X(i) \ddfrac{\exp\{ \alpha' (C Y)_i  \}}{\sum_{j=1}^n X(j) \exp\{ \alpha' (C Y)_j \}} \quad i =1, \ldots, n.
\end{align*}
An {\em approximate multiplier} $Y$ of $X$ is a strategy Y in $\mathbb{Y}$ such that
\begin{align*}
\hat{X}(i) \approx X(i) \ddfrac{\exp\{ \alpha' (C Y)_i  \}}{\sum_{j=1}^n X(j) \exp\{ \alpha' (C Y)_j \}} \quad i =1, \ldots, n.
\end{align*}
An exact multiplier can be obtained in some occasions using the {\em operative matrix} matrix at $X$ (cf. Lemma \ref{gt_lemma}). Using the operative matrix (of the growth transformation for either the first or the second primary barrier function), the next iterate our dynamical system generates can be obtained as
\begin{align*}
\hat{X}(i) = X(i) \ddfrac{\exp\{ \alpha (\bar{C}_X X)_i  \}}{\sum_{j=1}^n X(j) \exp\{ \alpha (\bar{C}_X X)_j \}} \quad i =1, \ldots, n.
\end{align*}
A multiplier strategy can, for example, be obtained by either inverting $C$ and one way to ensure that $Y$ is a probability vector is to configure $\mathsf{C}$ small enough (but greater than $X \cdot CX$).

In general, an exact multiplier is not always possible to obtain. Ariadne, thus, generates a sequence of approximate multipliers. Approximate multipliers are obtained by solving a pair of convex quadratic programs (see Lemma \ref{forever} and the succeeding discussion and derivation). Ariadne carries out this process in every iteration and computes the empirical average of the sequence of multipliers as it is for this sequence that our fixed-point (and equilibrium) approximation bounds apply.  To obtain an approximate multiplier we rely on the inverse function theorem, as used in:

\begin{lemma}
\label{forever}
$\forall$ $X \in \mathbb{Y}$, $\forall \mathcal{Y} \in \mathcal{C}(X)$, and $\forall \alpha > 0$, there exists a locally unique $Y \in \mathbb{R}^n$ such that $\mathcal{Y} = T_Y(X)$, where
\begin{align}
T_Y(X)_i = X(i) \ddfrac{\exp\{ \alpha (C Y)_i  \}}{\sum_{j=1}^n X(j) \exp\{ \alpha (C Y)_j \}} \quad i =1, \ldots, n,\label{equation_system}
\end{align}
unless $(X, Y)$ is a fixed point of \eqref{equation_system}, that is, unless $X$ is a pure strategy or otherwise
\begin{align*}
\forall i, j \in \mathcal{C}(X) : (CY)_{i} = (CY)_{j}. 
\end{align*}
\end{lemma}

\begin{proof}
We are looking to solve the system of equations
\begin{align*}
\mathcal{Y}(i) &\equiv T_Y(X)_i = X(i) \ddfrac{\exp\{ \alpha (C Y)_i  \}}{\sum_{j=1}^n X(j) \exp\{ \alpha (C Y)_j \}} \quad i =1, \ldots, n,
\end{align*}
that is to find $Y$ and $\alpha$ that satisfies this system assuming $X$ and $\mathcal{Y}$ are given subject to the constraints in the statement of the lemma. We would like to apply the inverse function theorem to show that this system always has a solution. To that end, we have
\begin{align*}
\frac{\partial T_Y(X)_i}{\partial Y(j)} = X(i) \ddfrac{ \alpha C_{ij} \exp\{ \alpha (C Y)_i \} \left( \sum_{k=1}^n X(k) \exp\{ \alpha (C Y)_k \} \right)}{\left( \sum_{k=1}^n X(k) \exp\{ \alpha (C Y)_k \} \right)^2 }-
\end{align*}
\begin{align*}
- X(i) \ddfrac{ \alpha \exp\{ \alpha (C Y)_i \} \left( \sum_{k=1}^n X(k) C_{kj} \exp\{\alpha (C Y)_k \} \right)  }{\left( \sum_{k=1}^n X(k) \exp\{ \alpha (C Y)_k \} \right)^2 }
\end{align*}
which implies
\begin{align*}
\frac{\partial T_Y(X)_i}{\partial Y(j)} \sim X(i) C_{ij} \exp\{ \alpha (C Y)_i \} \left( \sum_{k=1}^n X(k) \exp\{ \alpha (C Y)_k \} \right)-
\end{align*}
\begin{align*}
- X(i) \exp\{ \alpha (C Y)_i \} \left( \sum_{k=1}^n X(k) C_{kj} \exp\{\alpha (C Y)_k \} \right)
\end{align*}
and by rearranging
\begin{align*}
= X(i) \exp\{ \alpha (C Y)_i \} \left( C_{ij} \left( \sum_{k=1}^n X(k) \exp\{ \alpha (C Y)_k \} \right) - \left( \sum_{k=1}^n X(k) C_{kj} \exp\{\alpha (C Y)_k \} \right) \right).
\end{align*}
The induced matrix is invertible if and only if the matrix
\begin{align*}
\mathcal{C}_{ij} = C_{ij} \left( \sum_{k=1}^n X(k) \exp\{ \alpha (C Y)_k \} \right) - \left( \sum_{k=1}^n X(k) C_{kj} \exp\{\alpha (C Y)_k \} \right)
\end{align*}
is invertible. We may write $\mathcal{C}$ as
\begin{align*}
\mathcal{C} = \left( \left( \sum_{k=1}^n X(k) \exp\{ \alpha (C Y)_k \} \right) C - D \right)
\end{align*}
where $D$ is a rank one matrix (its rows are identical). $C$ is invertible and $D$ can be written as the outer product of two vectors (since its rows are identical), in particular, as 
\begin{align*}
D = \mathbf{1} v^T
\end{align*}
where $\mathbf{1}$ is a vector of ones and
\begin{align*}
v^T = \left[\sum_{k=1}^n X(k) C_{k1} \exp\{\alpha (C Y)_k \} \quad \cdots \quad \sum_{k=1}^n X(k) C_{kn} \exp\{\alpha (C Y)_k \} \right]
\end{align*}
Therefore, the Sherman-Morrison formula implies that the matrix $\mathcal{C}$ is invertible provided
\begin{align*}
1 - \frac{1}{\sum_{k=1}^n X(k) \exp\{ \alpha (C Y)_k \}} v^T C^{-1} \mathbf{1} \neq 0.
\end{align*}
However,
\begin{align*}
1 - \frac{1}{\sum_{k=1}^n X(k) \exp\{ \alpha (C Y)_k \}} v^T C^{-1} \mathbf{1} > 1 - \mathbf{1}^T C^{-1} \mathbf{1}.
\end{align*}
Furthermore, by the Hartman-Stampacchia theorem, there exists $c < 1$ such that 
\begin{align*}
C^{-1} \mathbf{1} = c Z
\end{align*}
where
\begin{align*}
\sum_{i=1}^n Z(i) = 1.
\end{align*}
Therefore,
\begin{align*}
1 - \mathbf{1}^T C^{-1} \mathbf{1} > 0,
\end{align*}
which implies $\mathcal{C}$ is invertible and, thus, the Jacobian of $T_Y(X)$ is invertible. Hence the lemma.
\end{proof}

To obtain an approximate (in general) multiplier, we may first solve \eqref{equation_system} using a variant of the Levenberg-Marquardt algorithm, in particular, the variant by \cite{Fan}, which has a favorable complexity bound (squared inverse of the norm of the gradient of the merit function) to approximate a solution. However, it is feasible to replace the previous step with an exact polynomial-time algorithm as follows: The system of equations 
\begin{align*}
\hat{X}(i) = X(i) \ddfrac{\exp\{ \alpha (C Y)_i  \}}{\sum_{j=1}^n X(j) \exp\{ \alpha (C Y)_j \}} \quad i =1, \ldots, n
\end{align*}
is equivalent to
\begin{align*}
\left( \sum_{j=1}^n X(j) \exp\{ \alpha (C Y)_j \} \right) \frac{\hat{X}(i)}{X(i)} = \exp\{ \alpha (C Y)_i  \} \quad i =1, \ldots, n
\end{align*}
which is, in turn, equivalent to 
\begin{align*}
\ddfrac{\hat{X}(i)/X(i)}{\hat{X}(j)/X(j)} = \exp\{ \alpha \left( (C Y)_i - (CY)_j \right)  \} \quad i =1, \ldots, n
\end{align*}
which is, in turn, equivalent to 
\begin{align*}
\ln \left( \frac{\hat{X}(i)}{X(i)} \right) - \ln \left( \frac{\hat{X}(j)}{X(j)} \right) = \alpha \left( (C Y)_i - (C Y)_j \right) \quad i, j = 1, \ldots, n \quad i \neq j.
\end{align*}
Therefore, $Y$ and $\alpha$ can be computed as the solution of the linear feasibility program
\begin{align}
(C Y)_i - (C Y)_j = \frac{1}{\alpha} \left( \ln \left( \frac{\hat{X}(i)}{X(i)} \right) - \ln \left( \frac{\hat{X}(j)}{X(j)} \right) \right) \quad i, j = 1, \ldots, n \quad i \neq j.\label{LFP}
\end{align}
Observe that the number of constraints in this program can be reduced to $n-1$. To ensure that the solution of \eqref{LFP} is unique even at fixed points, we solve the following convex quadratic program:
\begin{align}
\mbox{minimize} \quad &\|Y - X\|\notag\\
\mbox{subject to} \quad &(C Y)_i - (C Y)_j = \frac{1}{\alpha} \left( \ln \left( \frac{\hat{X}(i)}{X(i)} \right) - \ln \left( \frac{\hat{X}(j)}{X(j)} \right) \right) \quad i, j = 1, \ldots, n \quad i \neq j.\label{LFP_CQP}
\end{align}
Note that Ariadne does not need to a priori specify a value of $\alpha$ as that can be determined in an optimal fashion by the solution of the previous mathematical program. Having obtained $Y$, an approximate multiplier, call it $\tilde{Y}$ can be obtained by solving the convex quadratic program:
\begin{align}
\min \left\{ \| \tilde{Y} - Y \| \big| \tilde{Y} \in \mathbb{Y} \right\},\label{CQP}
\end{align}
where $\| \cdot \|$ is the Euclidean norm. Such convex quadratic programs admit polynomial-time algorithms to find a solution. Ariadne computes an approximate multiplier in every iteration of the dynamical system and in this way generates a secondary sequence, call it $\{ \tilde{X}^k \}_{k=0}^{\infty}$, of iterates. Given $X^K$ and $X^{K+1}$, Ariadne computes the corresponding approximate multiplier $\tilde{Y}^K$ and generates the next iterate $\tilde{X}^{K+1}$ of the secondary sequence using the equation:
\begin{align*}
\tilde{X}^{K+1}(i) = \tilde{X}^K(i) \ddfrac{\exp\{ \alpha (C \tilde{Y}^K)_i  \}}{\sum_{j=1}^n X(j) \exp\{ \alpha (C \tilde{Y}^K)_j \}} \quad i =1, \ldots, n.
\end{align*}
Note that the learning rate used in the iteration that generates the secondary sequence of iterates does not have to be equal to the learning rate that is obtained as a solution of the aforementioned mathematical program that is used to determine the optimal exact multiplier (out of which the approximate multiplier is obtained). The combination of optimization problems that gives $\tilde{Y}^K$ is the map of a dynamical system that receives as input an iterate of the principal dynamical system, denoted as $X^{K+1}$, and generates as output $\tilde{Y}^K$. By Berge's maximum theorem and the strict convexity of the objective function, $\tilde{Y}^K$ is continuous as a function of $Y^K$, an observation we rely upon in the sequel (noting in passing that it is also continuous as a function of $X^{K+1}$). Ariadne is in fact more complicated as discussed in the sequel. Note that the learning rate Ariadne uses in generating the secondary sequence of iterates is constant (to facilitate detecting when the Nisan parameter is greater than the clique number and, therefore, that this parameter should decrease). A range of suitable values for the learning rate is derived in Appendix \ref{repelling} (on repelling fixed points).

\subsection{Convergence of multipliers implies convergence of iterates}

Considering the map
\begin{align}
T_Y(X)_i = X(i) \cdot \frac{\exp\left\{ \alpha (CY)_i \right\}}{ \sum_{j=1}^n X(j) \exp \left\{ \alpha (CY)_j \right\} } \quad i = 1, \ldots, n,\label{eqxxx}
\end{align}
we have the following lemma:

\begin{lemma}
\label{nonuniformlemma}
Suppose $X^0$ is an arbitrary interior strategy. Then
\begin{align*}
\forall i,j \in \mathcal{K}(C) : (E_i - E_j) \cdot C\bar{Y}^K = \frac{1}{A_K} \ln \left( \frac{ X^{K+1}(i) }{ X^0(i) } \right) - \frac{1}{A_K} \ln \left( \frac{X^{K+1}(j)}{X^0(j)}\right).
\end{align*}
\end{lemma}

\begin{proof}
Let $T_Y(X) \equiv \hat{X}$. Then straight algebra gives
\begin{align*}
\frac{\hat{X}(i)}{\hat{X}(j)} = \frac{X(i)}{X(j)} \exp\{\alpha ((CY)_i - (CY)_j)\}
\end{align*}
and taking logarithms on both sides we obtain
\begin{align*}
\ln \left( \frac{\hat{X}(i)}{\hat{X}(j)} \right) = \ln \left( \frac{X(i)}{X(j)} \right) + \alpha ((CY)_i - (CY)_j).
\end{align*}
We may write the previous equation as
\begin{align*}
\ln \left( \frac{X^{k+1}(i)}{X^{k+1}(j)} \right) = \ln \left( \frac{X^k(i)}{X^k(j)} \right) + \alpha_k ((CY^k)_i - (CY^k)_j)
\end{align*}
Summing over $k = 0, \ldots K$, we obtain
\begin{align*}
\ln \left( \frac{X^{K+1}(i)}{X^{K+1}(j)} \right) = \ln \left( \frac{X^0(i)}{X^0(j)} \right) + \sum_{k=0}^K \alpha_k ((CY^k)_i - (CY^k)_j)
\end{align*}
and dividing by $A_K$ and rearranging, we further obtain
\begin{align*}
\frac{1}{A_K} \ln \left( \frac{X^{K+1}(i)}{X^{K+1}(j)} \right) = \frac{1}{A_K} \ln \left( \frac{X^0(i)}{X^0(j)} \right) + (E_i - E_j) \cdot C\bar{Y}^K
\end{align*}
which implies
\begin{align*}
(E_i - E_j) \cdot C\bar{Y}^K = \frac{1}{A_K} \ln \left( \frac{ X^{K+1}(i) }{ X^0(i) } \right) - \frac{1}{A_K} \ln \left( \frac{X^{K+1}(j)}{X^0(j)}\right)
\end{align*}
as claimed. 
\end{proof}

\begin{lemma}
\label{elementary_convergence_lemma}
If the sequence of multipliers $\left\{ Y^k \right\}$ converges, then the sequence of empirical averages $\left\{ \bar{Y}^K \right\}$ also converges to the same limit.
\end{lemma}

\begin{proof}
Assume the sequence $\left\{ Y^k \right\}$ converges and let $X^*$ be its limit. Then
\begin{align*}
\lim_{k \rightarrow \infty} \| X^* - Y^k \| = 0,
\end{align*}
where $\| \cdot \|$ is the Euclidean norm. Then the Stolz-Ces\'aro theorem implies that
\begin{align}
\lim_{K \rightarrow \infty} \left\{ \frac{1}{A_K} \sum_{k=0}^K \alpha_k \| X^* - Y^k \| \right\} = 0.\label{apricot1}
\end{align}
The convexity of the Euclidean distance function gives that
\begin{align}
\| X^* - \bar{Y}^K \| \leq \frac{1}{A_K} \sum_{k=0}^K \alpha_k \| X^* - Y^k \|.\label{apricot2}
\end{align}
\eqref{apricot1} and \eqref{apricot2} together imply
\begin{align*}
\lim_{K \rightarrow \infty} \| X^* - \bar{Y}^K \| = 0.
\end{align*}
Thus $\left\{ \bar{Y}^K \right\}$ also converges to $X^*$ as claimed.
\end{proof}

\begin{lemma}
\label{elementary_liminf_lemma_nonuniform}
If the sequence of multipliers converges to a maximal-clique equilibrium, say $X^*$, and $A_K \rightarrow \infty$, then the probability mass of every pure strategy outside the carrier of $X^*$ vanishes.
\end{lemma}

\begin{proof}
Since a maximal-clique equilibrium is a regular ESS (a property that is simple to verify from the structure of the Nisan-Bomze payoff matrix), if $i$ is a pure strategy in the carrier of $X^*$ and $j$ a pure strategy outside the carrier of $X^*$, then
\begin{align*}
(CX^*)_i > (CX^*)_j.
\end{align*}
Lemma \ref{elementary_convergence_lemma} implies that the empirical average of the sequence of multipliers also converges to $X^*$. Since the sequence $\{ C\bar{Y}^{K} \}$ converges to $CX^*$, the limit
\begin{align*}
\lim_{K \rightarrow \infty} \left\{ \frac{1}{A_K} \ln \left( \frac{ X^{K+1}(i) }{ X^0(i) } \right) - \frac{1}{A_K} \ln \left( \frac{X^{K+1}(j)}{X^0(j)}\right) \right\}
\end{align*}
exists by Lemma \ref{nonuniformlemma}. The same lemma further implies that
\begin{align*}
(CX^*)_i - (CX^*)_j = \lim_{K \rightarrow \infty} \left\{ \frac{1}{A_K} \ln \left( \frac{ X^{K+1}(i) }{ X^0(i) } \right) - \frac{1}{A_K} \ln \left( \frac{X^{K+1}(j)}{X^0(j)}\right) \right\} > 0.
\end{align*}
Let $\{X^{K_{\ell}}\}$ be a convergent subsequence such that
\begin{align*}
\lim_{\ell \rightarrow \infty} \left\{ \frac{1}{A_{K_{\ell}}} \ln \left( \frac{X^{K_{\ell}+1}(j)}{X^0(j)} \right) \right\} = \limsup\limits_{K \rightarrow \infty} \left\{ \frac{1}{A_{K}} \ln \left( \frac{X^{K+1}(j)}{X^0(j)} \right) \right\}.
\end{align*}
Then,
\begin{align*}
\lim_{K \rightarrow \infty} \left\{ \frac{1}{A_{K}} \ln \left( \frac{X^{K_{\ell}+1}(i)}{X^0(i)} \right) \right\} > \limsup\limits_{\ell \rightarrow \infty} \left\{ \frac{1}{A_{K_\ell}} \ln \left( \frac{X^{K_\ell+1}(j)}{X^0(j)} \right) \right\},
\end{align*}
which implies
\begin{align*}
\limsup\limits_{\ell \rightarrow \infty} \left\{ \frac{1}{A_{K_\ell}} \ln \left( \frac{X^{K_\ell+1}(j)}{X^0(j)} \right) \right\} < 0,
\end{align*}
which further implies, by the assumption $A_K \rightarrow \infty$,
\begin{align*}
\limsup\limits_{\ell \rightarrow \infty} \left\{ X^{K_\ell+1}(j) \right\} = 0,
\end{align*}
which even further implies
\begin{align*}
\lim_{K \rightarrow \infty} \left\{ X^{K+1}(j) \right\} = 0
\end{align*}
as claimed.
\end{proof}

\begin{lemma}
\label{multipliersconvergenceimpliesiteratesconvergence}
If the sequence of approximate multipliers converges to a maximal-clique equilibrium, the sequence of iterates converges to a fixed point in the maximal-clique equilibrium's carrier.
\end{lemma}

\begin{proof}
Let us denote the maximal-clique equilibrium the sequence of multipliers converges to by $X^*$. Lemma \ref{elementary_liminf_lemma_nonuniform} implies that the probability mass of every pure strategy outside the carrier of $X^*$ vanishes. This implies from equation \eqref{eqxxx} and straight algebra that at infinity the sequence of iterates takes the value of a fixed point in the maximal-clique equilibrium's carrier. Since the sequence of iterates is a continuous function of the sequence of multipliers, the sequence of iterates converges to that fixed point and, therefore, the lemma follows.
\end{proof}

\subsection{Ensuring approximate multipliers converge to a maximum clique}

If the Nisan parameter is equal to the clique number, the sequence of iterates converges to a maximum-clique equilibrium and as an implication of the method that generates the sequences of exact and approximate multipliers, the latter sequences also converge to a maximum-clique equilibrium---furthermore, as an implication of Lemma \ref{elementary_convergence_lemma} their empirical average converges likewise. That is, at infinity, the corresponding approximate multiplier is the same equilibrium that the sequence of (exact) multipliers converges to. Since the map that generates the sequence $\{ \tilde{Y}^k \}_{k=0}^{\infty}$ of approximate multipliers is continuous in the input (in that $\tilde{Y}^k$ is a continuous function of $Y^k$) and the sequence $\{Y^k\}_{k=0}^{\infty}$ converges (to a maximum-clique equilibrium), we obtain that $\{ \tilde{Y}^k \}_{k=0}^{\infty}$ also converges to the same equilibrium. (To summarize the argument more abstractly, we have two sequences that both assume the same value at infinity, one sequence converges to that value, and the second sequence is obtained by a continuous map from the first. We then conclude that the second sequence also converges to the value that it assumes at infinity). However, the secondary sequence of iterates may converge to a pure strategy and in this event we cannot analytically guarantee a polynomial upper bound on the algorithm's execution. To ensure that the secondary sequence of iterates converges to a maximum-clique equilibrium, we make sure that the maximum payoff of the iterates of the secondary sequence of multipliers remains bounded away from the value the maximum payoff assumes at a pure strategy, which is equal to one. In this way, the iterates generated by the sequence of approximate multipliers will be sure not to converge to a pure strategy and as we will see, this implies that we analytically prove a polynomial upper bound on the algorithm's execution. The mechanism by which we prevent the maximum payoff from assuming values close to the maximum is by interleaving iterations of map $T$ in the sequence of approximate multipliers---the technique is similar to the technique we previously employed to guarantee an upper bound on the maximum payoff in the primary sequence of iterates albeit that mechanism is based on the discrete-time replicator dynamic, whereas in the secondary sequence of multipliers we use $T$: Upon detecting an iterate of the secondary sequence whose maximum payoff exceeds $C_{00} + \epsilon$, where $0 < \epsilon < 1 - C_{00}$, our algorithm interleaves rounds of $T$ until the maximum payoff drops below $C_{00} + \epsilon$ at which point iterations using the sequence of approximate multipliers resume. We call the sequence of multipliers that ensues from the interleaving of approximate multipliers and $T$ the {\em extended sequence of approximate multipliers}. In the next lemma, we show that the extended sequence of approximate multipliers drives the secondary sequence of iterates to a maximum clique.

\begin{lemma}
\label{intermed_lemma}
If
\begin{align*}
T_Y(X)_i = X(i) \cdot \frac{\exp\left\{ \alpha (CY)_i \right\}}{ \sum_{j=1}^n X(j) \exp \left\{ \alpha (CY)_j \right\} } \quad i = 1, \ldots, n,
\end{align*}
converges to a clique, then that clique is a maximal clique.
\end{lemma}

\begin{proof}
If $T_Y$ converges to a clique, there exist at least a pair of pure strategies $i$ such that
\begin{align*}
\liminf\limits_{K \rightarrow \infty} \left\{ \ln \left( \ddfrac{X^{K+1}(i)}{X^0(i)} \right) \right\} > - \infty.
\end{align*}
Lemma \ref{nonuniform_lemma} gives
\begin{align*}
\max^{1\mbox{st}}_{p \in \{1, \ldots, n\}} \left\{ (C\bar{Y}^K)_p \right\} - \max^{2\mbox{nd}}_{p \in \{1, \ldots, n\}} \left\{ (C\bar{Y}^K)_p \right\} =
\end{align*}
\begin{align}
= \left(\ln \left( \max^{1\mbox{st}}_{p \in \{1, \ldots, n\}} \left\{ \ddfrac{X^{K+1}(p)}{X^0(p)} \right\} \right) - \ln \left( \max^{2\mbox{nd}}_{p \in \{1, \ldots, n\}} \left\{ \ddfrac{X^{K+1}(p)}{X^0(p)} \right\} \right) \right) \ddfrac{1}{A_K}\label{11111}
\end{align}
Theorem \ref{equilibrium_error_nonuniform} gives, provided $K$ is large enough such that \eqref{cond} is satisfied,
\begin{align*}
\ln \left( \max^{2\mbox{nd}}_{p \in \{1, \ldots, n\}} \left\{ \ddfrac{X^{K+1}(p)}{X^0(p)} \right\} \right) > 0 \Rightarrow
\end{align*}
\begin{align}
\max^{2\mbox{nd}}_{p \in \{1, \ldots, n\}} \left\{ (C\bar{Y}^K)_p \right\} - \bar{Y}^K \cdot C \bar{Y}^K \leq \ln \left( \max^{2\mbox{nd}}_{p \in \{1, \ldots, n\}} \left\{ \ddfrac{X^{K+1}(p)}{X^0(p)} \right\} \right) \ddfrac{2c}{A_K},\label{22222}
\end{align}
or
\begin{align*}
\ln \left( \max^{2\mbox{nd}}_{p \in \{1, \ldots, n\}} \left\{ \ddfrac{X^{K+1}(p)}{X^0(p)} \right\} \right) < 0 \Rightarrow
\end{align*}
\begin{align}
\max^{2\mbox{nd}}_{p \in \{1, \ldots, n\}} \left\{ (C\bar{Y}^K)_p \right\} - \bar{Y}^K \cdot C \bar{Y}^K \leq \ln \left( \max^{2\mbox{nd}}_{p \in \{1, \ldots, n\}} \left\{ \ddfrac{X^{K+1}(p)}{X^0(p)} \right\} \right) \ddfrac{2/c}{A_K}.\label{33333}
\end{align}
Summing \eqref{11111} and \eqref{22222}, we obtain
\begin{align*}
\max^{1\mbox{st}}_{p \in \{1, \ldots, n\}} \left\{ (C\bar{Y}^K)_p \right\} - \bar{Y}^K \cdot C \bar{Y}^K \leq \left(\ln \left( \max^{1\mbox{st}}_{p \in \{1, \ldots, n\}} \left\{ \ddfrac{X^{K+1}(p)}{X^0(p)} \right\} \right) + (2c-1) \ln \left( \max^{2\mbox{nd}}_{p \in \{1, \ldots, n\}} \left\{ \ddfrac{X^{K+1}(p)}{X^0(p)} \right\} \right) \right) \ddfrac{1}{A_K}
\end{align*}
Summing \eqref{11111} and \eqref{33333}, we obtain 
\begin{align*}
\max^{1\mbox{st}}_{p \in \{1, \ldots, n\}} \left\{ (C\bar{Y}^K)_p \right\} - \bar{Y}^K \cdot C \bar{Y}^K \leq \left(\ln \left( \max^{1\mbox{st}}_{p \in \{1, \ldots, n\}} \left\{ \ddfrac{X^{K+1}(p)}{X^0(p)} \right\} \right) + \left(\frac{2}{c}-1\right) \ln \left( \max^{2\mbox{nd}}_{p \in \{1, \ldots, n\}} \left\{ \ddfrac{X^{K+1}(p)}{X^0(p)} \right\} \right) \right) \ddfrac{1}{A_K}
\end{align*}
In our case, in the fashion we have transformed the Nisan-Bomze payoff matrix so that all payoff entries are positive, we have $c=2$. Taking the limit as $K \rightarrow \infty$, we obtain that the clique $T_Y$ converges to is an equilibrium and, therefore, it is a maximal clique.
\end{proof}

\begin{lemma}
\label{kflklfklf}
If the Nisan parameter is equal to the clique number and the learning rate used in invocations of $T$ as the extended sequence of approximate multipliers is generated is small enough, then the extended sequence of approximate multipliers converges to a maximum-clique equilibrium and the secondary sequence of iterates likewise converges to the same maximum-clique equilibrium.
\end{lemma}

\begin{proof}
If the learning rate is small enough, then $T$ does not converge to a fixed point whose maximum payoff is equal to or greater than $C_{00} + \epsilon$. This is the subject of Appendix \ref{repelling} where we also compute an appropriate value for the learning rate to prevent the possibility of convergence to such a fixed point. Under the assumption that the learning rate is small enough, the extended sequence of approximate multipliers converges to a maximum-clique equilibrium: $T$ is invoked a finite number of times and the limit of the sequence is equal to the limit of the sequence of approximate multipliers, which is a maximum-clique equilibrium. Lemma \ref{multipliersconvergenceimpliesiteratesconvergence} continues to hold for the extended sequence of approximate multipliers. If in the secondary sequence of iterates, the probability mass of a pure strategy in the carrier of this maximum-clique equilibrium vanishes, then Lemma \ref{intermed_lemma} implies the existence of an equilibrium inside the carrier of the maximum-clique equilibrium, which is an impossibility. Therefore, the secondary sequence of iterates likewise converges to the same maximum-clique equilibrium as the extended sequence of approximate multipliers as claimed.
\end{proof}

\section{Computation of a maximum clique requires polynomial time}
\label{Complexity}

In this section, we complete the proof that {\bf P = NP} by discussing how to configure the approximation error of the dynamical system such that a maximum-clique equilibrium can be computed in a polynomial number of iterations and then analyzing the complexity of the system's execution.

\subsection{On the ``minimum positive gap'' of a symmetric bimatrix game}

Our goal herein is to define a concept that is able to transform the equilibrium approximation algorithm based on Hedge to a polynomial computation algorithm in the Nisan game. But let us start more generally with the setting of symmetric bimatrix games: Let $C$ be a symmetric bimatrix game and $X \in \mathbb{X}(C)$. We may give a preliminary definition of the {\em gap} $\Gamma_C(X)$ of $X \in \mathbb{X}(C)$ as
\begin{align*}
\Gamma_C(X) \equiv (CX)_{\max} - (CX)_{\min}.
\end{align*}
Our motivation for introducing this definition has as follows: Every pure or mixed strategy of a symmetric bimatrix game $C$ has a gap (except for equalizers). One way to define a ``minimum gap'' is as the minimum over all strategies of $C$. But $C$ has sub-games. The sub-games that are carriers of fixed points have minimum gap of zero. Sub-games that do not carry fixed points also have a positive minimum gap (as sub-games). It is meaningful that in the definition of the minimum gap we take the sub-games into account and here is why: Let us extend the previous preliminary definition and define the {\em extended gap} $\Gamma_{CC'}(X)$ of $X \in \mathbb{X}(C)$ as
\begin{align*}
\Gamma_{CC'}(X) \equiv (CX)_{\max} - (C'X)_{\min},
\end{align*}
where $C'$ is a subgame of $C$ (padded with zeros so that the dimensions of $C$ and $C'$ agree. Furthermore, define the the {\em minimum positive gap} of $C$, call is $\gamma_{\min}(C)$ as
\begin{align*}
\gamma_{\min}(C) \equiv \min_{C'} \left\{ \min_{X \in \mathbb{X}(C)} \left\{ \Gamma_{CC'}(X) | \Gamma_{CC'}(X) > 0 \right\} \right\}.
\end{align*}
where the first minimization is taken over all subgames of $C$. We claim that a $\gamma_{\min}/2$-well-supported equilibrium, call it $\hat{X}$, lies inside the carrier of an equilibrium (which we can readily compute knowing the carrier). Let us assume for the sake of contradiction that the carrier of $\hat{X}$ does not carry an equilibrium (which is an equalizer of the carrier). Then with Proposition \ref{my_10_cents} in mind there is a gap equal to or greater than $\gamma_{\min}$ (inside the carrier), which is an impossibility given that $\gamma_{\min}/2$-well-supported equilibrium exists. Hence the claim. In the sequel, we are concerned with the the Nisan game. In this game, a related to the above but more appropriate, in that it simplifies the analysis, definition of gap is as follows:
\begin{align}
\hat{\Gamma}(X) \equiv \max \{ X \cdot CX | X \in \mathbb{X}(C) \} - X \cdot CX\label{gdefin1}
\end{align}
where $C$ is the Nisan-Bomze payoff matrix barring strategy $0$. We may then define the minimum positive gap as
\begin{align}
\hat{\gamma}_{\min} \equiv \min_{C'} \left\{ \min_{X \in \mathbb{X}(C')} \left\{ \hat{\Gamma}(X) | \hat{\Gamma}(X) > 0 \right\} \right\}\label{gdefin2}
\end{align}
where the first minimization is taken over all subgames of $C$. As above, we claim that a $\gamma_{\min}/2$-well-supported equilibrium, call it $\hat{X}$, lies inside the carrier of an equilibrium (which we can readily compute knowing the carrier) provided $E_0$ is a GNSS of the Nisan game but not a GESS so that a maximum-clique equilibrium exists. Let us assume for the sake of contradiction that the carrier, say $\hat{C}'$ of $\hat{X}$ does not carry an equilibrium (which is an equalizer of the carrier, where the term equalizer is to be understood with the latest definition of gap). Then, keeping again Proposition \ref{my_10_cents} in mind, there is a gap equal to or greater than $\gamma_{\min}$ (inside the carrier), that is, 
\begin{align*}
\max \{ X \cdot \hat{C}'X | X \in \mathbb{X}(\hat{C}') \} - X \cdot \hat{C}'X = (\hat{C}'X)_{\max} - X \cdot \hat{C}'X \geq\gamma_{\min}
\end{align*}
which is an impossibility given that $\gamma_{\min}/2$-well-supported equilibrium exists implying
\begin{align*}
(\hat{C}'X)_{\max} - (\hat{C}'X)_{\min} \leq \gamma_{\min}/2.
\end{align*} 
Hence the claim. We have the following theorem:

\begin{theorem}
\label{iloveyoucat}
The minimum positive gap $\hat{\gamma}_{\min}$ of $C$ is at least
\begin{align*}
\frac{1}{2} \left( 1 + 1 - \frac{1}{2k} \right) - \frac{1}{2} \left( 1 + 1 - \frac{1}{2(k-1)} \right) = \frac{1}{4} \left( \frac{1}{(k-1)} - \frac{1}{k} \right)
\end{align*}
where $k$ is the size of the maximum clique.
\end{theorem}

\begin{proof}
Let $\hat{C}'$ be a subgame of $C$. Then
\begin{align*}
\min_{X \in \mathbb{X}(\hat{C}')} \left\{ \hat{\Gamma}(X) \right\} = \max \{ X \cdot \hat{C}X | X \in \mathbb{X}(\hat{C}) \} - \max \{ X \cdot \hat{C}X | X \in \mathbb{X}(\hat{C}') \}
\end{align*}
where
\begin{align*}
\max \{ X \cdot \hat{C}' X | X \in \mathbb{X}(\hat{C}') \} = \frac{1}{2} \left( 1 + 1 - \frac{1}{2k'} \right)
\end{align*}
where $k'$ is the maximum clique in the subgraph corresponding to $\hat{C}'$. The maximum possible clique smaller than the maximum clique has one vertex less. Therefore,
\begin{align*}
\hat{\gamma}_{\min} \equiv \min_{\hat{C}'} \left\{ \min_{X \in \mathbb{X}(\hat{C}')} \left\{ \hat{\Gamma}(X) | \hat{\Gamma}(X) > 0 \right\} \right\} \geq \frac{1}{2} \left( 1 + 1 - \frac{1}{2k} \right) - \frac{1}{2} \left( 1 + 1 - \frac{1}{2(k-1)} \right)
\end{align*}
as claimed.
\end{proof}

\subsection{Computational complexity of the dynamical system's execution}

\begin{lemma}
\label{theorem1}
Suppose the Nisan parameter $k$ is equal to or greater than the clique number and that the multiplier learning rate in each iteration is lower bounded by a constant.\footnote{It, for example, satisfies \begin{align*}
\alpha (C_{00} + \epsilon) - (\exp\{\alpha\}-1) C_{00} > 0.
\end{align*}
where $C_{00} + \epsilon < 1$ (see Appendix \ref{repelling}).} Suppose further that the equilibrium approximation error is set equal to
\begin{align*}
\epsilon = \frac{\epsilon_a^2}{8} \mbox{ where } \epsilon_a = \frac{1}{4} \left( \frac{1}{k-1} - \frac{1}{k} \right).
\end{align*}
(i) If the Nisan parameter $k$ is equal to the clique number $\omega(G)$, the empirical average of the extended sequence of approximate multipliers converges to a maximum-clique equilibrium and attains the aforementioned equilibrium approximation error in a polynomial number of iterations in the inverse of the approximation error.\\ 
(ii) If the (primary) sequence of iterates does not converge to a fixed point, either the empirical average of the extended sequence of approximate multipliers converges to an equilibrium and attains the equilibrium approximation error in a polynomial number of iterations in the inverse of the approximation error or the equilibrium approximation error bound is violated in a polynomial number of iterations in the inverse of the approximation error.
\end{lemma}

\begin{proof}
{\em (i)} If the Nisan parameter is equal to the clique number, that the empirical average of the extended sequence of approximate multipliers converges to a maximum-clique equilibrium is an implication of Lemma \ref{kflklfklf}. Upon attainment of an equilibrium approximation error of $\epsilon$ conditions \eqref{cond2} / \eqref{cond3} in Theorem \ref{equilibrium_error_nonuniform} are satisfied and therefore condition \eqref{cond} in the same theorem is also satisfied. That condition \eqref{cond3} is satisfied is a simple implication of the following calculation: Assuming pure strategy $p$ carries the maximum-clique the empirical average has approximated, then
\begin{align*}
\max_{i=1}^n \left\{ (C\bar{\tilde{Y}}^K)_i \right\} \leq (C\bar{\tilde{Y}}^K)_p + \frac{1}{4} \left( \left(1 + 1 - \frac{1}{2k} \right) - \left(1 + 1 - \frac{1}{2(k-1)} \right) \right)
\end{align*}
which implies
\begin{align*}
(C\bar{\tilde{Y}}^K)_p \geq \frac{1}{4} \max_{i=1}^n \left\{ (C\bar{\tilde{Y}}^K)_i \right\} + \frac{1}{4} \left(1 + 1 - \frac{1}{2(k-1)} \right) 
\end{align*}
and, assuming the underlying graph $G$ is not complete, further implies
\begin{align*}
\frac{1}{2} \left( 1 + 1 - \frac{1}{2(k-1)} \right) \geq \min_{i=1}^n \left\{ (C\bar{\tilde{Y}}^K)_i \right\},
\end{align*}
which even further implies
\begin{align*}
(C\bar{\tilde{Y}}^K)_p \geq \frac{1}{2} \left( \max_{i=1}^n \left\{ (C\bar{\tilde{Y}}^K)_i \right\} + \min_{i=1}^n \left\{ (C\bar{\tilde{Y}}^K)_i \right\} \right).
\end{align*}
Lemma \ref{kflklfklf} implies that all pure strategies $p$ in the carrier of the limit of $\bar{\tilde{Y}}^K$ satisfy 
\begin{align*}
\liminf\limits_{K \rightarrow \infty} \left\{ \ln \left( \ddfrac{\tilde{X}^{K+1}(p)}{\tilde{X}^0(p)} \right) \right\} > - \infty.
\end{align*}
There are now two possibilities, namely, either
\begin{align*}
\ln \left( \max^{2\mbox{nd}}_{p \in \{1, \ldots, n\}} \left\{ \ddfrac{X^{K+1}(p)}{X^0(p)} \right\} \right) > 0 \Rightarrow
\end{align*}
\begin{align}
\max^{2\mbox{nd}}_{p \in \{1, \ldots, n\}} \left\{ (C\bar{Y}^K)_p \right\} - \bar{Y}^K \cdot C \bar{Y}^K \leq \ln \left( \max^{2\mbox{nd}}_{p \in \{1, \ldots, n\}} \left\{ \ddfrac{X^{K+1}(p)}{X^0(p)} \right\} \right) \ddfrac{2c}{A_K},\label{22222a}
\end{align}
or
\begin{align*}
\ln \left( \max^{2\mbox{nd}}_{p \in \{1, \ldots, n\}} \left\{ \ddfrac{X^{K+1}(p)}{X^0(p)} \right\} \right) < 0 \Rightarrow
\end{align*}
\begin{align}
\max^{2\mbox{nd}}_{p \in \{1, \ldots, n\}} \left\{ (C\bar{Y}^K)_p \right\} - \bar{Y}^K \cdot C \bar{Y}^K \leq \ln \left( \max^{2\mbox{nd}}_{p \in \{1, \ldots, n\}} \left\{ \ddfrac{X^{K+1}(p)}{X^0(p)} \right\} \right) \ddfrac{2/c}{A_K}\label{33333a}
\end{align}
where
\begin{align*}
c = \ddfrac{\max_{ij} C_{ij}}{\min_{ij} C_{ij}} = 2.
\end{align*}
Furthermore, Lemma \ref{nonuniform_lemma} gives
\begin{align*}
\max^{1\mbox{st}}_{p \in \{1, \ldots, n\}} \left\{ (C\bar{\tilde{Y}}^K)_p \right\} - \max^{2\mbox{nd}}_{p \in \{1, \ldots, n\}} \left\{ (C\bar{\tilde{Y}}^K)_p \right\} =
\end{align*}
\begin{align*}
= \left(\ln \left( \max^{1\mbox{st}}_{p \in \{1, \ldots, n\}} \left\{ \ddfrac{\tilde{X}^{K+1}(p)}{X^0(p)} \right\} \right) - \ln \left( \max^{2\mbox{nd}}_{p \in \{1, \ldots, n\}} \left\{ \ddfrac{\tilde{X}^{K+1}(p)}{X^0(p)} \right\} \right) \right) \ddfrac{1}{A_K}.
\end{align*}
Summing the previous inequalities, we obtain in the first case that
\begin{align*}
\max^{1\mbox{st}}_{p \in \{1, \ldots, n\}} \left\{ (C\bar{\tilde{Y}}^K)_p \right\} - \bar{\tilde{Y}}^K \cdot C \bar{\tilde{Y}}^K \leq \left(\ln \left( \max^{1\mbox{st}}_{p \in \{1, \ldots, n\}} \left\{ \ddfrac{\tilde{X}^{K+1}(p)}{X^0(p)} \right\} \right) + (2c-1) \ln \left( \max^{2\mbox{nd}}_{p \in \{1, \ldots, n\}} \left\{ \ddfrac{\tilde{X}^{K+1}(p)}{X^0(p)} \right\} \right) \right) \ddfrac{1}{A_K},
\end{align*}
which implies
\begin{align}
\max^{1\mbox{st}}_{p \in \{1, \ldots, n\}} \left\{ (C\bar{\tilde{Y}}^K)_p \right\} - \bar{\tilde{Y}}^K \cdot C \bar{\tilde{Y}}^K \leq \left(\ln \left( \max^{1\mbox{st}}_{p \in \{1, \ldots, n\}} \left\{ \ddfrac{\tilde{X}^{K+1}(p)}{X^0(p)} \right\} \right) \right) \ddfrac{2c}{A_K}\label{ccc}
\end{align}
and in the second case that
\begin{align*}
\max^{1\mbox{st}}_{p \in \{1, \ldots, n\}} \left\{ (C\bar{\tilde{Y}}^K)_p \right\} - \bar{\tilde{Y}}^K \cdot C \bar{\tilde{Y}}^K \leq \left(\ln \left( \max^{1\mbox{st}}_{p \in \{1, \ldots, n\}} \left\{ \ddfrac{\tilde{X}^{K+1}(p)}{X^0(p)} \right\} \right) + (2/c-1) \ln \left( \max^{2\mbox{nd}}_{p \in \{1, \ldots, n\}} \left\{ \ddfrac{\tilde{X}^{K+1}(p)}{X^0(p)} \right\} \right) \right) \ddfrac{1}{A_K},
\end{align*}
which implies (since as mentioned earlier $c=2$)
\begin{align}
\max^{1\mbox{st}}_{p \in \{1, \ldots, n\}} \left\{ (C\bar{\tilde{Y}}^K)_p \right\} - \bar{\tilde{Y}}^K \cdot C \bar{\tilde{Y}}^K \leq \left(\ln \left( \max^{1\mbox{st}}_{p \in \{1, \ldots, n\}} \left\{ \ddfrac{\tilde{X}^{K+1}(p)}{X^0(p)} \right\} \right) \right) \ddfrac{1}{A_K}.\label{ccc2}
\end{align}
If the learning rate is lower bounded by $\alpha$ we obtain that $A_K \geq \alpha (K+1)$ and the lemma follows.\\

{\em (ii)} If the sequence of iterates does not converge to a fixed point either the equilibrium approximation error $\epsilon$ is attained, which implies either
\begin{align*}
\max^{2\mbox{nd}}_{p \in \{1, \ldots, n\}} \left\{ (C\bar{\tilde{Y}}^K)_p \right\} - \bar{\tilde{Y}}^K \cdot C \bar{\tilde{Y}}^K \leq \ln \left( \max^{2\mbox{nd}}_{p \in \{1, \ldots, n\}} \left\{ \ddfrac{\tilde{X}^{K+1}(p)}{X^0(p)} \right\} \right) \ddfrac{2c}{A_K}
\end{align*}
or
\begin{align*}
\max^{2\mbox{nd}}_{p \in \{1, \ldots, n\}} \left\{ (C\bar{\tilde{Y}}^K)_p \right\} - \bar{\tilde{Y}}^K \cdot C \bar{\tilde{Y}}^K \leq \ln \left( \max^{2\mbox{nd}}_{p \in \{1, \ldots, n\}} \left\{ \ddfrac{\tilde{X}^{K+1}(p)}{X^0(p)} \right\} \right) \ddfrac{1}{A_K}
\end{align*}
which further implies by the previous derivation that
\begin{align*}
\max^{1\mbox{st}}_{p \in \{1, \ldots, n\}} \left\{ (C\bar{\tilde{Y}}^K)_p \right\} - \bar{\tilde{Y}}^K \cdot C \bar{\tilde{Y}}^K \leq \left(\ln \left( \max^{1\mbox{st}}_{p \in \{1, \ldots, n\}} \left\{ \ddfrac{\tilde{X}^{K+1}(p)}{X^0(p)} \right\} \right) \right) \ddfrac{2c}{A_K}
\end{align*}
or that
\begin{align*}
\max^{1\mbox{st}}_{p \in \{1, \ldots, n\}} \left\{ (C\bar{\tilde{Y}}^K)_p \right\} - \bar{\tilde{Y}}^K \cdot C \bar{\tilde{Y}}^K \leq \left(\ln \left( \max^{1\mbox{st}}_{p \in \{1, \ldots, n\}} \left\{ \ddfrac{\tilde{X}^{K+1}(p)}{X^0(p)} \right\} \right) \right) \ddfrac{1}{A_K}
\end{align*}
and, therefore, that the empirical average of the sequence of multipliers converges to an equilibrium in a polynomial number of iterations in the inverse of the approximation error or otherwise this bound is necessarily violated once the (polynomial) upper bound on the number of iterations to attain an equilibrium approximation of $\epsilon$ is attained. Noting that such an upper bound on the number of iterations can be readily computed once the learning rate used to generate the secondary sequence of iterates is, for example, fixed (cf. Section \ref{repelling}), this completes the proof of the lemma.
\end{proof}

\if(0)

\begin{lemma}
\label{elementary_convergence_lemma}
If the sequence of multipliers $\left\{ Y^k \right\}$ converges, then the sequence of empirical averages $\left\{ \bar{Y}^K \right\}$ also converges to the same limit.
\end{lemma}

\begin{proof}
Assume the sequence $\left\{ Y^k \right\}$ converges and let $X^*$ be its limit. Then
\begin{align*}
\lim_{k \rightarrow \infty} \| X^* - Y^k \| = 0,
\end{align*}
where $\| \cdot \|$ is the Euclidean norm. Then the Stolz-Ces\'aro theorem implies that
\begin{align}
\lim_{K \rightarrow \infty} \left\{ \frac{1}{A_K} \sum_{k=0}^K \alpha_k \| X^* - Y^k \| \right\} = 0.\label{apricot1}
\end{align}
The convexity of the Euclidean distance function gives that
\begin{align}
\| X^* - \bar{Y}^K \| \leq \frac{1}{A_K} \sum_{k=0}^K \alpha_k \| X^* - Y^k \|.\label{apricot2}
\end{align}
\eqref{apricot1} and \eqref{apricot2} together imply
\begin{align*}
\lim_{K \rightarrow \infty} \| X^* - \bar{Y}^K \| = 0.
\end{align*}
Thus $\left\{ \bar{Y}^K \right\}$ also converges to $X^*$ as claimed.
\end{proof}

\begin{theorem}
\label{theorem2}
Suppose the Nisan parameter is equal to the clique number and that the multiplier learning rate in each iteration is lower bounded by a constant. Then, the sequence of iterates converges to a maximum-clique equilibrium and the empirical average of the sequence of multipliers also converges to the same maximum-clique equilibrium in a polynomial number of iterations in the inverse of the approximation error.
\end{theorem}

\begin{proof}
If the Nisan parameter is equal to $\omega(G)$, since by Lemma \ref{convergence} the sequence of iterates converges, the sequence of multipliers (as well as their average) also converges to the same limit. This is an implication of method that generates the sequence of multipliers from the sequence of iterates (as discussed earlier---cf. Section \ref{sequence_of_multipliers}). To show the convergence of the (weighted) empirical average of the multipliers it we may either invoke Lemma \ref{elementary_convergence_lemma} or Lemma \ref{theorem1}. Lemma \ref{theorem1} further implies the polynomial bound on the number of iterations in the inverse of the approximation error.
\end{proof}

\fi

If the fixed-point approximation error is set such that
\begin{align*}
\epsilon = \frac{\epsilon_a^2}{8} \mbox{ where } \epsilon_a = \frac{1}{4} \left( \frac{1}{\omega(G)-1} - \frac{1}{\omega(G)} \right),
\end{align*}
then, provided the Nisan parameter is equal to the clique number, on attainment of this approximation error, the corresponding approximate well-supported equilibrium (cf. Section \ref{approx_equilibria}) is in the carrier of a maximum-clique equilibrium, which implies a maximum-clique can be readily computed. Therefore, starting at any value of the Nisan parameter greater than the clique number and following the algorithmic steps of Ariadne, a maximum clique is guaranteed to be computed:
\begin{itemize}

\item If the Nisan parameter is equal to $\omega(G)$ our dynamical system is ensured to generate a sequence of iterates that are able to compute a maximum clique in polynomial time. 

\item If the Nisan parameter is greater than $\omega(G)$, the number of iterations required to detect so so as to decrease the Nisan parameter by one is also polynomial, implying that the Nisan parameter eventually becomes equal to the clique number and a maximum clique is computed.

\end{itemize}

\begin{theorem}
If in every iteration our algorithm solves the convex quadratic programs \eqref{LFP_CQP} and \eqref{CQP} to compute an approximate multiplier, in the worst case, the number of iterations that our algorithm requires to compute a maximum-clique equilibrium is $O(n^6)$.
\end{theorem}

\begin{proof}
Referring back to $\eqref{ccc}$ since, by Lemma \ref{aaaomega},
\begin{align*}
\ln \left( \max^{1\mbox{st}}_{p \in \{1, \ldots, n\}} \left\{ \ddfrac{\tilde{X}^{K+1}(p)}{X^0(p)} \right\} \right) = O\left( \frac{1}{n} \right),
\end{align*}
and the desired fixed-point approximation error is
\begin{align*}
\epsilon = O \left( \frac{1}{n^4} \right)
\end{align*}
and, furthermore, since the Nisan parameter is decreased up to $O(n)$ times, the bound follows.
\end{proof}

* Due to Lemma \ref{aaaomega}, the previous bounds require that $G$ is equipped with a clique of size four.

\section{Closing remarks and future work}
\label{conclusion}

In the course of numerically testing Ariadne, we observed a phenomenon whereby the discrete-time replicator dynamic returned an iterate wherein the probability masses of various pure strategies where close to zero and that to recover from such iterate and continue the execution of our dynamical system we found it beneficial to use the leapfrogging mechanism. Further work is required to understand this phenomenon. 
In the course of testing a related algorithm, we observed a phenomenon whereby Hedge did not increase the potential function. The iterate was close to a fixed point and the value of the learning rate was large. We believe this phenomenon can be attributed to a numerical roundoff error (possibly to the implementation of the exponential function). 

Further work is required to understand how Ariadne interacts with commodity software and hardware systems and to devise designs that are backward compatible with these systems.

In closing, we would like to point out that a feature that is not unique to Ariadne (for example, see \citep{Ciaran}) is that it admits a parallel implementation. We leave the details of such an implementation as future work. We would also like to raise the possibility that the research presented in this paper would benefit from the efficient computation of {\em centroids or barycenters} \citep{SW, Shvartsman}. Although there exist efficient randomized algorithms for computing barycenters of convex bodies (for example, see \citep{Bertimas}) it is an interesting question whether such algorithms can be derandomized (cf. \citep{Rademacher}). We leave this as a question for future work. We would finally like to point out that our result can shed further light on the exact relationship among complexity classes using {\em inapproximability results} that have been obtained from the maximum clique problem (see \citep{QWu} for a summary of these results). This is also an exciting question for future work.

\section*{Acknowledgments}

This paper has benefitted from my interaction with my YouTube account and I thank those that are responsible for the configuration of the content in that account.

\bibliographystyle{abbrvnat}
\bibliography{real}

\appendix

\section{Hedge as a growth transformation}

Our main result in this section is that Hedge is a growth transformation for all $\alpha > 0$. In our proof of this result, we follow \cite{Baum-Eagon}. We will need two auxiliary results:

\begin{proposition}[H\"older's inequality]
For all $(x_1, \ldots, x_n), (y_1, \ldots, y_n) \in \mathbb{R}^n$, we have
\begin{align*}
\sum_{k=1}^n |x_k| |y_k| \leq \left( \sum_{k=1}^n |x_k|^p \right)^{\frac{1}{p}} \left( \sum_{k=1}^n |y_k|^q \right)^{\frac{1}{q}}
\end{align*}
where
\begin{align*}
\frac{1}{p} + \frac{1}{q} = 1, \quad p, q > 1.
\end{align*}
\end{proposition}

\begin{proposition}[Weighted AM-GM inequality]
\label{AMGM_inequality}
Given nonnegative numbers $x_1, \ldots, x_n$ and nonnegative weights $w_1, \ldots, w_n$, set $w = w_1 + \cdots w_n$. If $w > 0$, then the inequality
\begin{align*}
\frac{w_1 x_1 + \cdots + w_n x_n}{w} \geq \left( x_1^{w_1} \cdots x_n^{w_n} \right)^{\frac{1}{w}}
\end{align*}
holds with equality if and only if all the $x_k$ with $w_k > 0$ are equal.
\end{proposition}

\cite{Baum-Eagon} show that:

\begin{proposition}
\label{Baum_Eagon}
Let $P(X) = P(\{x_{ij}\})$ be a polynomial with nonnegative coefficients homogeneous in its variables $\{x_{ij}\}$. Let $x = \{x_{ij}\}$ be any point in the domain $D$, where
\begin{align*}
D = \left\{ x \bigg| x_{ij} \geq 0, \sum_{j = 1}^{q_i} x_{ij} = 1, i = 1, \ldots, p, j = 1, \ldots, q_i \right\}.
\end{align*}
For $x = \{ x_{ij} \} \in D$, let $J(X) = J(\{x_{ij} \})$ denote the point of $D$ whose $i, j$ coordinate is
\begin{align*}
J(x)_{ij} = x_{ij} \frac{\left.\frac{\partial P}{\partial x_{ij}} \right|_{(x)}}{\sum_{j=1}^{q_i} x_{ij} \left. \frac{\partial P}{\partial x_{ij}} \right|_{(x)}}.
\end{align*}
Then $J$ is a growth transformation for $P$, that is, $P(J(x)) > P(x)$ unless $J(x) = x$.
\end{proposition}

Let us summarize their notation (which we will follow in our own results): $\mu$ denotes a doubly indexed array of nonnegative integers: $\mu = \{ \mu_{ij} \}$. $x^\mu$ is an abbreviation for
\begin{align*}
x^{\mu} \equiv \prod_{i=1}^p \prod_{j=1}^{q_i} x_{ij}^{\mu_{ij}}.
\end{align*}
$c_{\mu}$ is an abbreviation for $c_{\mu_{ij}}$. Using the previous conventions,
\begin{align*}
P(X) \equiv \sum_{\mu} c_{\mu} x^{\mu}
\end{align*}
and
\begin{align*}
J(x)_{ij} =  \left( \sum_{\mu} c_{\mu} \mu_{ij} x^{\mu} \right) \bigg/ \left( \sum_{j=1}^{q_i} \sum_{\mu} c_{\mu} \mu_{ij} x^{\mu} \right).
\end{align*}
The goal is to prove that
\begin{align*}
P(x) = \sum_{\mu} c_{\mu} x^{\mu} \leq \sum_{\mu} c_{\mu} \prod_{i=1}^p \prod_{j=1}^{q_i} J(x)_{ij}^{\mu_{ij}}.
\end{align*}

\begin{lemma}
\label{Baum_Eagon_extension_powers_4}
Let $P(X) = P(\{x_{ij}\})$ be a polynomial homogeneous in its variables $\{x_{ij}\}$ of degree $d$. Let $x = \{x_{ij}\}$ be any point in the domain $D$, where
\begin{align*}
D = \left\{ x \bigg| x_{ij} \geq 0, \sum_{j = 1}^{q_i} x_{ij} = 1, i = 1, \ldots, p, j = 1, \ldots, q_i \right\}.
\end{align*}
For $x = \{ x_{ij} \} \in D$, let $J^k(X) = J^k(\{x_{ij} \})$ denote the point of $D$ whose $i, j$ coordinate is
\begin{align*}
J^k(x)_{ij} = x_{ij} \frac{\left( 1 + \frac{1}{k} \alpha \left.\frac{\partial P}{\partial x_{ij}} \right|_{(x)} \right)^{k}}{\sum_{j=1}^{q_i} x_{ij} \left( 1 + \frac{1}{k} \alpha \left.\frac{\partial P}{\partial x_{ij}} \right|_{(x)} \right)^{k}}
\end{align*}
where $k$ is a positive integer. Then provided that
\begin{align}
\forall i = 1, \ldots, p \mbox{ } \forall j = 1, \ldots, q_i : \left. \frac{\partial P}{\partial x_{ij}} \right|_{(x)} \leq 1,\label{basic_assumption}
\end{align}
$J^k$ is a growth transformation for $P$ for all $\alpha > 0$. Furthermore, for all positive integers $k$, the fixed points of $J^k$ coincide with the fixed points of $J$.
\end{lemma}

\begin{proof}
We have
\begin{align*}
P(x) = \sum_{\mu} c_{\mu} x^{\mu}
\end{align*}
which we may equivalently write as
\begin{align*}
= \sum_{\mu} \left( c_{\mu} \right)^{\frac{1}{dk+1}} \left( c_{\mu} \right)^{\frac{dk}{dk+1}} x^{\mu}
\end{align*}
which we may equivalently write as
\begin{align*}
= \sum_{\mu} \left( c_{\mu} \right)^{\frac{1}{dk+1}} \left( c_{\mu} \right)^{\frac{dk}{dk+1}} x^{\mu} \left( \prod_{i=1}^p \prod_{j=1}^{q_i} J^k(x)_{ij}^{\mu_{ij}} \right)^{\frac{1}{dk+1}} \left( \prod_{i=1}^p \prod_{j=1}^{q_i} \left( \frac{1}{J^k(x)_{ij}} \right)^{\mu_{ij}} \right)^{\frac{1}{dk+1}}
\end{align*}
and, rearranging terms, we obtain
\begin{align*}
= \sum_{\mu} \left( c_{\mu} \prod_{i=1}^p \prod_{j=1}^{q_i} J^k(x)_{ij}^{\mu_{ij}} \right)^{\frac{1}{dk+1}} \times \left\{ \left( c_{\mu} \right)^{\frac{dk}{dk+1}} x^{\mu} \left( \prod_{i=1}^p \prod_{j=1}^{q_i} \frac{1}{J^k(x)_{ij}} \right)^{\frac{\mu_{ij}}{dk+1}} \right\}.
\end{align*}
We next apply H\"older's inequality with parameters $p = dk+1$ and $q = \frac{dk+1}{dk}$ to obtain 
\begin{align*}
P(x) \leq \left( \sum_{\mu} \left( c_{\mu} \prod_{i=1}^p \prod_{j=1}^{q_i} J^k(x)_{ij}^{\mu_{ij}} \right) \right)^{\frac{1}{dk+1}} \left( \sum_{\mu} c_{\mu} x^{\mu \frac{dk+1}{dk}} \left( \prod_{i=1}^p \prod_{j=1}^{q_i} \frac{1}{J^k(x)_{ij}} \right)^{\frac{\mu_{ij}}{dk}} \right)^{\frac{dk}{dk+1}}
\end{align*}
and using that
\begin{align*}
x^{\mu \frac{dk+1}{dk}} = x^{\mu \left( 1 + \frac{1}{dk} \right)} = x^{\mu} x^{\frac{\mu}{dk}} = x^\mu \prod_{i=1}^p \prod_{j=1}^{q_i} x_{ij}^{\frac{\mu_{ij}}{dk}}
\end{align*}
we further obtain
\begin{align*}
\leq \left( \sum_{\mu} \left( c_{\mu} \prod_{i=1}^p \prod_{j=1}^{q_i} J^k(x)_{ij}^{\mu_{ij}} \right) \right)^{\frac{1}{dk+1}} \left( \sum_{\mu} c_{\mu} x^{\mu} \prod_{i=1}^p \prod_{j=1}^{q_i} \left( \frac{x_{ij}}{J^k(x)_{ij}} \right)^{\frac{\mu_{ij}}{dk}} \right)^{\frac{dk}{dk+1}}
\end{align*}
which, using the weighted AM-GM inequality, yields
\begin{align*}
\leq \left( \sum_{\mu} \left( c_{\mu} \prod_{i=1}^p \prod_{j=1}^{q_i} J^k(x)_{ij}^{\mu_{ij}} \right) \right)^{\frac{1}{dk+1}} \left( \sum_{\mu} c_{\mu} x^{\mu} \left( \frac{1}{\sum_{i=1}^p \sum_{j=1}^{q_i} \frac{\mu_{ij}}{d}} \right) \sum_{i=1}^p \sum_{j=1}^{q_i} \frac{\mu_{ij}}{d} \left( \frac{x_{ij}}{J^k(x)_{ij}} \right)^{\frac{1}{k}} \right)^{\frac{dk}{dk+1}}
\end{align*}
which, by the homogeneity of $P$, implies
\begin{align}
= \left( \sum_{\mu} \left( c_{\mu} \prod_{i=1}^p \prod_{j=1}^{q_i} J^k(x)_{ij}^{\mu_{ij}} \right) \right)^{\frac{1}{dk+1}} \left( \sum_{\mu} c_{\mu} x^{\mu} \sum_{i=1}^p \sum_{j=1}^{q_i} \frac{\mu_{ij}}{d} \left( \frac{x_{ij}}{J^k(x)_{ij}} \right)^{\frac{1}{k}} \right)^{\frac{dk}{dk+1}}\label{askdjfasjdfhdskjfhd}
\end{align}
Let us work with the expression inside the parenthesis in the second product term. We, thus, substituting the expression for $J^k$, have
\begin{align*}
\sum_{\mu} c_{\mu} x^{\mu} \sum_{i=1}^p \sum_{j=1}^{q_i} \frac{\mu_{ij}}{d} \left( \frac{x_{ij}}{J^k(x)_{ij}} \right)^{\frac{1}{k}} =
\end{align*}
\begin{align}
= \sum_{\mu} c_{\mu} x^{\mu} \sum_{i=1}^p \sum_{j=1}^{q_i} \frac{\mu_{ij}}{d} \left( \frac{\sum_{\ell=1}^{q_i} x_{i \ell} \left( 1 + \frac{1}{k} \alpha \left.\frac{\partial P}{\partial x_{i \ell}} \right|_{(x)} \right)^{k}}{\left( 1 + \frac{1}{k} \alpha \left.\frac{\partial P}{\partial x_{i j}} \right|_{(x)} \right)^{k}} \right)^{\frac{1}{k}}\label{aksvbkjsdfshjfjd}
\end{align}
which, cancelling powers in the denominator, implies
\begin{align*}
= \frac{1}{d} \sum_{\mu} c_{\mu} x^{\mu} \sum_{i=1}^p \sum_{j=1}^{q_i} \mu_{ij} \frac{\left( \sum_{\ell=1}^{q_i} x_{i \ell} \left( 1 + \frac{1}{k} \alpha \left.\frac{\partial P}{\partial x_{i \ell}} \right|_{(x)} \right)^{k} \right)^{\frac{1}{k}}}{1 + \frac{1}{k} \alpha \left.\frac{\partial P}{\partial x_{i j}} \right|_{(x)}}
\end{align*}
which, since
\begin{align*}
\sum_{\ell=1}^{q_i} x_{i \ell} \left( 1 + \frac{1}{k} \alpha \left.\frac{\partial P}{\partial x_{i \ell}} \right|_{(x)} \right)^{k} > 1,
\end{align*}
implies that
\begin{align*}
\leq \frac{1}{d} \sum_{\mu} c_{\mu} x^{\mu} \sum_{i=1}^p \sum_{j=1}^{q_i} \mu_{ij} \frac{\sum_{\ell=1}^{q_i} x_{i \ell} \left( 1 + \frac{1}{k} \alpha \left.\frac{\partial P}{\partial x_{i \ell}} \right|_{(x)} \right)^{k}}{1 + \frac{1}{k} \alpha \left.\frac{\partial P}{\partial x_{i j}} \right|_{(x)}}
\end{align*}
which further implies by assumption \eqref{basic_assumption} that
\begin{align*}
\leq \frac{1}{d} \sum_{\mu} c_{\mu} x^{\mu} \sum_{i=1}^p \sum_{j=1}^{q_i} \mu_{ij} \frac{\sum_{\ell=1}^{q_i} x_{i \ell} \left( 1 + \frac{1}{k} \alpha \left.\frac{\partial P}{\partial x_{i \ell}} \right|_{(x)} \right)^{k}}{\left( 1 + \frac{1}{k} \alpha \right) \left.\frac{\partial P}{\partial x_{i j}} \right|_{(x)}}
\end{align*}
which, by rearranging terms, further implies
\begin{align*}
= \left( \frac{1}{1 + \frac{1}{k} \alpha} \right) \frac{1}{d} \sum_{\mu} c_{\mu} x^{\mu} \sum_{i=1}^p \sum_{j=1}^{q_i} \mu_{ij} \frac{\sum_{\ell=1}^{q_i} x_{i \ell} \left( 1 + \frac{1}{k} \alpha \left.\frac{\partial P}{\partial x_{i \ell}} \right|_{(x)} \right)^{k}}{\left.\frac{\partial P}{\partial x_{i j}} \right|_{(x)}}
\end{align*}
which even further implies
\begin{align*}
= \left( \frac{1}{1 + \frac{1}{k} \alpha} \right) \frac{1}{d} \sum_{\mu} c_{\mu} x^{\mu} \sum_{i=1}^p \sum_{j=1}^{q_i} \mu_{ij} x_{ij} \frac{\sum_{\ell=1}^{q_i} x_{i \ell} \left( 1 + \frac{1}{k} \alpha \left.\frac{\partial P}{\partial x_{i \ell}} \right|_{(x)} \right)^{k}}{x_{ij} \left.\frac{\partial P}{\partial x_{i j}} \right|_{(x)}}
\end{align*}
which even further implies
\begin{align*}
= \left( \frac{1}{1 + \frac{1}{k} \alpha} \right) \frac{1}{d} \sum_{\mu} c_{\mu} x^{\mu} \sum_{i=1}^p \sum_{j=1}^{q_i} \mu_{ij} x_{ij} \frac{\sum_{\ell=1}^{q_i} x_{i \ell} \left( 1 + \frac{1}{k} \alpha \left.\frac{\partial P}{\partial x_{i \ell}} \right|_{(x)} \right)^{k}}{\sum_{\nu} c_{\nu} x^{\nu} \nu_{ij}}
\end{align*}
which by rearranging the order of summation even further implies
\begin{align*}
= \left( \frac{1}{1 + \frac{1}{k} \alpha} \right) \frac{1}{d} \sum_{i=1}^p \sum_{j=1}^{q_i} x_{ij} \sum_{\mu} c_{\mu} x^{\mu} \mu_{ij}  \frac{\sum_{\ell=1}^{q_i} x_{i \ell} \left( 1 + \frac{1}{k} \alpha \left.\frac{\partial P}{\partial x_{i \ell}} \right|_{(x)} \right)^{k}}{\sum_{\nu} c_{\nu} x^{\nu} \nu_{ij}}
\end{align*}
which, cancelling terms and summing probability masses, even further implies
\begin{align*}
= \left( \frac{1}{1 + \frac{1}{k} \alpha} \right) \frac{1}{d} \sum_{i=1}^p \sum_{\ell=1}^{q_i} x_{i \ell} \left( 1 + \frac{1}{k} \alpha \left.\frac{\partial P}{\partial x_{i \ell}} \right|_{(x)} \right)^{k}
\end{align*}
which, using the binomial theorem, implies
\begin{align*}
= \left( \frac{1}{1 + \frac{1}{k} \alpha} \right) \frac{1}{d} \sum_{i=1}^p \sum_{\ell=1}^{q_i} x_{i \ell}  \sum_{m=0}^k {k \choose m} \left( \frac{1}{k} \alpha \left.\frac{\partial P}{\partial x_{i \ell}} \right|_{(x)} \right)^m
\end{align*}
which, using assumption \eqref{basic_assumption}, implies
\begin{align*}
\leq \left( \frac{1}{1 + \frac{1}{k} \alpha} \right) \frac{1}{d} \sum_{i=1}^p \sum_{\ell=1}^{q_i} x_{i \ell} \left.\frac{\partial P}{\partial x_{i \ell}} \right|_{(x)}  \sum_{m=0}^k {k \choose m} \left( \frac{1}{k} \alpha  \right)^m
\end{align*}
and using the binomial theorem for a second time, we further obtain
\begin{align*}
\leq \frac{1}{d} \sum_{i=1}^p \sum_{\ell=1}^{q_i} x_{i \ell} \left.\frac{\partial P}{\partial x_{i \ell}} \right|_{(x)} = P(x)
\end{align*}
where the equality follows by the Euler theorem for homogeneous functions. Substituting in \eqref{askdjfasjdfhdskjfhd}, we obtain
\begin{align*}
P(x) \leq \left( \sum_{\mu} \left( c_{\mu} \prod_{i=1}^p \prod_{j=1}^{q_i} J^k(x)_{ij}^{\mu_{ij}} \right) \right)^{\frac{1}{dk+1}} \left( P(x) \right)^{\frac{dk}{dk+1}}
\end{align*}
which implies
\begin{align*}
(P(x))^{1 - \frac{dk}{dk+1}} \leq \sum_{\mu} \left( c_{\mu} \prod_{i=1}^p \prod_{j=1}^{q_i} J^k(x)_{ij}^{\mu_{ij}} \right)^{\frac{1}{dk+1}}
\end{align*}
which finally implies
\begin{align*}
P(x) \leq P(J(x))
\end{align*}
The strictness of the inequality if $J^k(x) \neq x$ follows from \eqref{askdjfasjdfhdskjfhd} and the strictness of the weighted arithmetic-geometric inequality in Proposition \ref{AMGM_inequality} if all summands are not equal.\\

Let us proceed with the proof of the second part of the lemma: We will show that $x$ is a fixed point if and only if for all $i = 1, \ldots, p$ and for all $j, \ell$ such that $x_{ij}, x_{i\ell} > 0$, we have that
\begin{align*}
\left.\frac{\partial P}{\partial x_{ij}} \right|_{(x)} = \left.\frac{\partial P}{\partial x_{i\ell}} \right|_{(x)}.
\end{align*}
First we show sufficiency: Some of the coordinates of $x$
are zero and some are positive. Clearly, 
the zero coordinates will not become positive after applying the map. 
Now, notice that, given $i$, for all $j$ such that $x_{ij} > 0$, 
\begin{align*}
\left( 1 + \frac{1}{k} \alpha \left.\frac{\partial P}{\partial x_{ij}} \right|_{(x)} \right)^{k} = \sum_{j=1}^{q_i} x_{ij} \left( 1 + \frac{1}{k} \alpha \left.\frac{\partial P}{\partial x_{ij}} \right|_{(x)} \right)^{k}
\end{align*}
and, therefore, $J^k(x)_{ij} = x_{ij}$, and this is true for all $i$. Now we show necessity: 
If $x$ is a fixed point, then for all $i = 1, \ldots, p$ and for all $j, \ell$ such that $x_{ij}, x_{i\ell} > 0$, we have that
\begin{align*}
\left.\frac{\partial P}{\partial x_{ij}} \right|_{(x)} = \left.\frac{\partial P}{\partial x_{i\ell}} \right|_{(x)}.
\end{align*}
Because $x$ is a fixed point, $J^k(x)_{ij} = x_{ij}$. Therefore,
{\allowdisplaybreaks
\begin{align}
J^k(x)_{ij} &= x_{ij}\notag\\
x_{ij} \frac{\left( 1 + \frac{1}{k} \alpha \left.\frac{\partial P}{\partial x_{ij}} \right|_{(x)} \right)^{k}}{\sum_{j=1}^{q_i} x_{ij} \left( 1 + \frac{1}{k} \alpha \left.\frac{\partial P}{\partial x_{ij}} \right|_{(x)} \right)^{k}} &= x_{ij}\notag\\
\left( 1 + \frac{1}{k} \alpha \left.\frac{\partial P}{\partial x_{ij}} \right|_{(x)} \right)^{k} &= \sum_{j=1}^{q_i} x_{ij} \left( 1 + \frac{1}{k} \alpha \left.\frac{\partial P}{\partial x_{ij}} \right|_{(x)} \right)^{k}.\label{eqcondition-finalstep_1}
\end{align}
}
Equation \eqref{eqcondition-finalstep_1} implies that, for all $i=1,\ldots, p$ and for all $j$ such that $x_{ij} > 0$,
\begin{align*}
\left( 1 + \frac{1}{k} \alpha \left.\frac{\partial P}{\partial x_{ij}} \right|_{(x)} \right)^{k} = c
\end{align*}
where $c$ is a constant. Cancelling the power, the constant $1$, and the factor $(1/k) \alpha$ yields the claim. Notice that the fixed points are independent of $k$. That these are also the fixed points of the replicator dynamic follows the same pattern (see also \citep{LosertAkin}).
\end{proof}

\begin{lemma}
\label{Baum_Eagon_extension_exponential_function}
Let $P(X) = P(\{x_{ij}\})$ be a polynomial homogeneous in its variables $\{x_{ij}\}$ of degree $d$. Let $x = \{x_{ij}\}$ be any point in the domain $D$, where
\begin{align*}
D = \left\{ x \bigg| x_{ij} \geq 0, \sum_{j = 1}^{q_i} x_{ij} = 1, i = 1, \ldots, p, j = 1, \ldots, q_i \right\}.
\end{align*}
For $x = \{ x_{ij} \} \in D$, let $J^{\infty} (X) = J^{\infty}(\{x_{ij} \})$ denote the point of $D$ whose $i, j$ coordinate is
\begin{align*}
J^{\infty} (x)_{ij} = x_{ij} \frac{ \exp \left\{ \alpha \left.\frac{\partial P}{\partial x_{ij}} \right|_{(x)} \right\}}{\sum_{j=1}^{q_i} x_{ij} \exp \left\{ \alpha \left.\frac{\partial P}{\partial x_{ij}} \right|_{(x)} \right\}},
\end{align*}
where $\alpha > 0$. Then, provided that \eqref{basic_assumption} holds, $J^{\infty}$ is a growth transformation for $P$.
\end{lemma}

\begin{proof}
We may equivalently write $J^{\infty}$ as
\begin{align*}
J^{\infty}(x)_{ij} = \frac{ x_{ij} \lim_{k \rightarrow \infty} \left\{ \left( 1 + \frac{1}{k} \alpha \left.\frac{\partial P}{\partial x_{ij}} \right|_{(x)} \right)^{k} \right\}}{\sum_{j=1}^{q_i} x_{ij} \left( \lim_{k \rightarrow \infty} \left\{ \left( 1 + \frac{1}{k} \alpha \left.\frac{\partial P}{\partial x_{ij}} \right|_{(x)} \right)^{k} \right\} \right)} = \lim_{k \rightarrow \infty} \left\{ \frac{ x_{ij} \left( 1 + \frac{1}{k} \alpha \left.\frac{\partial P}{\partial x_{ij}} \right|_{(x)} \right)^{k}}{\sum_{j=1}^{q_i} x_{ij} \left( 1 + \frac{1}{k} \alpha \left.\frac{\partial P}{\partial x_{ij}} \right|_{(x)} \right)^{k}}  \right\},
\end{align*}
which implies that the sequence of maps
\begin{align*}
\left\{ J^k(x)_{ij} = \frac{ x_{ij} \left( 1 + \frac{1}{k} \alpha \left.\frac{\partial P}{\partial x_{ij}} \right|_{(x)} \right)^{k}}{\sum_{j=1}^{q_i} x_{ij} \left( 1 + \frac{1}{k} \alpha \left.\frac{\partial P}{\partial x_{ij}} \right|_{(x)} \right)^{k}} \right\}_{k = 1}^{\infty}
\end{align*}
converges pointwise to $J^{\infty}$. Therefore, by the definition of pointwise convergence, for each $x$ and an arbitrarily small $\epsilon > 0$, we can find $\hat{k} > 0$ such that for all $k \geq \hat{k}$, the maps $J^k$ map $x$ to within an $\epsilon$-ball of $J^{\infty}(x)$. It is easy to show that the fixed points of $J^{\infty}$ (cf. Lemma \ref{fixed_points_Hedge}) coincide with the fixed points of $J$. Lemma \ref{Baum_Eagon_extension_powers_4} then implies that, for all $x$ and $\alpha > 0$, unless $x$ is a fixed point, $J^{\infty}$ strictly increases the value of the polynomial $P$. This completes the proof.
\end{proof}

\section{An inequality on the approximation error of multiplicative weights}
\label{Fixed_point_bound}

In this section, we analyze the fixed-point approximation error of the general map
\begin{align*}
T_i(X) = X(i) \cdot \frac{\exp\left\{ \alpha (CY)_i \right\}}{ \sum_{j=1}^n X(j) \exp \left\{ \alpha (CY)_j \right\} } \quad i = 1, \ldots, n,
\end{align*}
in particular, the fixed point approximation error of the sequence $\left\{ \bar{Y}^K \right\}_{K = 0}^{\infty}$ of empirical averages of $\left\{ Y^k \right\}$. The empirical average $\bar{Y}^K$ at iteration $K = 0, 1, 2, \ldots$  is a weighted arithmetic mean
\begin{align*}
\bar{Y}^K = \frac{1}{A_K} \sum_{k=0}^K \alpha_k Y^k, \mbox{ where } A_K = \sum_{k = 0}^K \alpha_k
\end{align*}
and $\alpha_k > 0$ is the learning rate parameter used in step $k$. In the case when the learning rate is held constant from round to round, the weighted arithmetic means reduces to a simple arithmetic mean
\begin{align*}
\bar{Y}^K = \frac{1}{K+1} \sum_{k=0}^K Y^k.
\end{align*}

\begin{lemma}
\label{nonuniform_lemma}
Suppose $X^0$ is an arbitrary interior strategy. Then
\begin{align*}
\forall i,j \in \mathcal{K}(C) : (E_i - E_j) \cdot C\bar{Y}^K = \frac{1}{A_K} \ln \left( \frac{ X^{K+1}(i) }{ X^0(i) } \right) - \frac{1}{A_K} \ln \left( \frac{X^{K+1}(j)}{X^0(j)}\right).
\end{align*}
Let $Y \in \mathbb{X}(C)$ be arbitrary. Then, for all $p \in \mathcal{C}(Y)$ the approximation error of the weighted empirical average is
\begin{align*}
(C\bar{Y}^K)_p - \bar{Y}^K \cdot C\bar{Y}^K = \frac{1}{A_K} \ln \left( \frac{ X^{K+1}(p) }{ X^0(p) } \right) - \frac{1}{A_K} \sum_{j=1}^n \bar{Y}^K(j) \ln \left( \frac{X^{K+1}(j)}{X^0(j)}\right)
\end{align*}
an expression which we may equivalently write as follows:
\begin{align*}
(C\bar{Y}^K)_p - \bar{Y}^K \cdot C\bar{Y}^K = \frac{1}{A_K} \sum_{j=1}^n \bar{Y}^K(j) \ln \left( \ddfrac{\frac{ X^{K+1}(p) }{ X^0(p) }}{\frac{X^{K+1}(j)}{X^0(j)}} \right).
\end{align*}
\end{lemma}

\begin{proof}
Let $T(X) \equiv \hat{X}$. Then straight algebra gives
\begin{align*}
\frac{\hat{X}(i)}{\hat{X}(j)} = \frac{X(i)}{X(j)} \exp\{\alpha ((CY)_i - (CY)_j)\}
\end{align*}
and taking logarithms on both sides we obtain
\begin{align*}
\ln \left( \frac{\hat{X}(i)}{\hat{X}(j)} \right) = \ln \left( \frac{X(i)}{X(j)} \right) + \alpha ((CY)_i - (CY)_j).
\end{align*}
We may write the previous equation as
\begin{align*}
\ln \left( \frac{X^{k+1}(i)}{X^{k+1}(j)} \right) = \ln \left( \frac{X^k(i)}{X^k(j)} \right) + \alpha_k ((CY^k)_i - (CY^k)_j)
\end{align*}
Summing over $k = 0, \ldots K$, we obtain
\begin{align*}
\ln \left( \frac{X^{K+1}(i)}{X^{K+1}(j)} \right) = \ln \left( \frac{X^0(i)}{X^0(j)} \right) + \sum_{k=0}^K \alpha_k ((CY^k)_i - (CY^k)_j)
\end{align*}
and dividing by $A_K$ and rearranging, we further obtain
\begin{align*}
\frac{1}{A_K} \ln \left( \frac{X^{K+1}(i)}{X^{K+1}(j)} \right) = \frac{1}{A_K} \ln \left( \frac{X^0(i)}{X^0(j)} \right) + (E_i - E_j) \cdot C\bar{Y}^K
\end{align*}
which implies
\begin{align*}
(E_i - E_j) \cdot C\bar{Y}^K = \frac{1}{A_K} \ln \left( \frac{ X^{K+1}(i) }{ X^0(i) } \right) - \frac{1}{A_K} \ln \left( \frac{X^{K+1}(j)}{X^0(j)}\right)
\end{align*}
as claimed in the first equation of the lemma. The previous equation further implies
\begin{align*}
(C\bar{Y}^K)_p - E_j \cdot C\bar{Y}^K = \frac{1}{A_K} \ln \left( \frac{ X^{K+1}(p) }{ X^0(p) } \right) - \frac{1}{A_K} \ln \left( \frac{X^{K+1}(j)}{X^0(j)}\right)
\end{align*}
which even further implies
\begin{align*}
(C\bar{Y}^K)_p - \bar{Y}^K \cdot C\bar{Y}^K = \frac{1}{A_K} \ln \left( \frac{ X^{K+1}(p) }{ X^0(p) } \right) - \frac{1}{A_K} \sum_{j=1}^n \bar{Y}^K(j) \ln \left( \frac{X^{K+1}(j)}{X^0(j)}\right)
\end{align*}
as claimed in the second and third equations of the lemma.
\end{proof}

\begin{lemma}
\label{ranking_lemma_nonuniform}
Suppose $X^0$ is an arbitrary interior strategy. Then, for all $K \geq 0$, the vectors
\begin{align*}
\left[ \begin{array}{c}
(C \bar{Y}^K)_1 \\
\vdots \\
(C \bar{Y}^K)_n \\
\end{array} \right] 
\quad \mbox{and} \quad
\left[ \begin{array}{c}
X^{K+1}(1)/X^0(1) \\
\vdots \\
X^{K+1}(n)/X^0(n) \\
\end{array} \right] 
\end{align*}
have the same ranking in the following sense: For all $K \geq 0$, if 
\begin{align*}
\sigma_K(1), \sigma_K(2), \ldots, \sigma_K(n)
\end{align*} 
is a permutation of the set of pure strategies such that
\begin{align*}
(C \bar{Y}^K)_{\sigma_K(1)} \geq_1 \cdots \geq_{n-1} (C\bar{Y}^K)_{\sigma_K(n)},
\end{align*}
then
\begin{align*}
\frac{X^{K+1}(\sigma_K(1))}{X^{0}(\sigma_K(1))} \geqslant_1 \cdots \geqslant_{n-1} \frac{X^{K+1}(\sigma_K(n))}{X^{0}(\sigma_K(n))}
\end{align*}
and, for all $i = 1, \ldots, n-1$, we have that $\geq_i$ is an equality if and only if $\geqslant_i$ is an equality.
\end{lemma}

\begin{proof}
Straightforward implication of Lemma \ref{nonuniform_lemma}.
\end{proof}

\begin{theorem}
\label{equilibrium_error_nonuniform}
Let $C$ be a positive payoff matrix such that
\begin{align*}
\ddfrac{\max_{ij} C_{ij}}{\min_{ij} C_{ij}} = c
\end{align*}
and denote $X^k \equiv T^k(X^0)$, where $X^0$ is an interior strategy.\\

(i) Consider a sequence of positive learning rates such that $p \in \{1, \ldots, n\}$ satisfies
\begin{align*}
\liminf\limits_{K \rightarrow \infty} \left\{ \ln \left( \ddfrac{X^{K+1}(p)}{X^0(p)} \right) \right\} > - \infty.
\end{align*}
Then for all $K \geq 0$ such that no probability mass of the sequence of iterates has vanished due to roundoff in finite precision arithmetic, such that
\begin{align}
\max_{\ell=1}^n \left\{  \left( \ln \left( \ddfrac{\ddfrac{X^{K+1}(p)}{X^0(p)}}{\ddfrac{X^{K+1}(\ell)}{X^0(\ell)}} \right) \right)^2 \right\} =  \left( \max_{\ell=1}^n \left\{ \ln \left( \ddfrac{\ddfrac{X^{K+1}(p)}{X^0(p)}}{\ddfrac{X^{K+1}(\ell)}{X^0(\ell)}} \right) \right\} \right)^2,\label{cond}
\end{align}
such that
\begin{align*}
(C\bar{Y}^K)_p < \max_{i=1}^n \left\{ (C\bar{Y}^K)_i \right\},
\end{align*}
and such that
\begin{align*}
\ln \left( \ddfrac{X^{K+1}(p)}{X^0(p)} \right) > 0,
\end{align*}
the weighted empirical average
\begin{align*}
\bar{Y}^K = \frac{1}{A_K} \sum_{k=0}^K \alpha_k Y^k
\end{align*} 
satisfies the following inequality
\begin{align*}
(C\bar{Y}^K)_p - \bar{Y}^K \cdot C \bar{Y}^K \leq \ln \left( \ddfrac{X^{K+1}(p)}{X^0(p)} \right) \ddfrac{2c}{A_K}.
\end{align*}
If 
\begin{align*}
\ln \left( \ddfrac{X^{K+1}(p)}{X^0(p)} \right) < 0,
\end{align*}
we obtain 
\begin{align*}
(C\bar{Y}^K)_p - \bar{Y}^K \cdot C \bar{Y}^K \leq \ln \left( \ddfrac{X^{K+1}(p)}{X^0(p)} \right) \ddfrac{2/c}{A_K}.
\end{align*}

(ii) The previous bound continues to hold even if 
\begin{align*}
(C\bar{Y}^K)_p = \max_{i=1}^n \left\{ (C\bar{Y}^K)_i \right\},
\end{align*}
and the maximum is attained by two or more pure strategies.\\

(iii) The condition
\begin{align}
\max_{\ell=1}^n \left\{ \ln \left( \ddfrac{\ddfrac{X^{K+1}(\ell)}{X^0(\ell)}}{\ddfrac{X^{K+1}(p)}{X^0(p)}} \right) \right\} \leq \max_{\ell=1}^n \left\{ \ln \left( \ddfrac{\ddfrac{X^{K+1}(p)}{X^0(p)}}{\ddfrac{X^{K+1}(\ell)}{X^0(\ell)}} \right) \right\}\label{cond2}
\end{align}
implies \eqref{cond}. Condition \eqref{cond2} is equivalent to
\begin{align}
(C\bar{Y})_p \geq \frac{1}{2} \left( \max_{i=1}^n \left\{ (C\bar{Y})_i \right\} + \min_{i=1}^n \left\{ (C\bar{Y})_i \right\} \right)\label{cond3}
\end{align} 
provided no probability mass of the sequence of iterates has vanished due to roundoff in finite precision arithmetic.
\if(0) If there exists $q \in \{1, \ldots, n\}$ such that
\begin{align*}
\lim_{K \rightarrow \infty} \left\{ \ln \left( \ddfrac{X^{K+1}(q)}{X^0(q)} \right) \right\} = - \infty.
\end{align*}
then \eqref{cond2} / \eqref{cond3} are eventually satisfied. \fi
\end{theorem}

\begin{proof}
{\em (i)} Lemma \ref{nonuniform_lemma} gives that, $\forall p \in \{1, \ldots, n\}$,
\begin{align}
(C\bar{Y}^K)_p - \bar{Y}^K \cdot C\bar{Y}^K = \frac{1}{A_K} \sum_{j=1}^n \bar{Y}^K(j) \ln \left( \ddfrac{\frac{ X^{K+1}(p) }{ X^0(p) }}{\frac{X^{K+1}(j)}{X^0(j)}} \right).\label{enaaa}
\end{align}
Furthermore, from Chebyshev's order inequality and Lemma \ref{ranking_lemma_nonuniform}, we have
\begin{align*}
\sum_{i=1}^n (C\bar{Y}^K)_i X^{K+1}(i) \left(\sigma_K + \rho_K \ln \left( \ddfrac{\ddfrac{X^{K+1}(p)}{X^0(p)}}{\ddfrac{X^{K+1}(i)}{X^0(i)}} \right) \right) \leq
\end{align*}
\begin{align*}
\leq \left( \sum_{i=1}^n (C\bar{Y}^K)_i X^{K+1}(i) \right)  \left( \sum_{i=1}^n X^{K+1}(i) \left(\sigma_K + \rho_K \ln \left( \ddfrac{\ddfrac{X^{K+1}(p)}{X^0(p)}}{\ddfrac{X^{K+1}(i)}{X^0(i)}} \right) \right) \right),
\end{align*}
where $\sigma_K > 0$ and $\rho_K > 0$. We may rewrite the previous inequality as the following more concise expression:
\begin{align*}
\sum_{i=1}^n (C\bar{Y}^K)_i X^{K+1}(i) \left(\sigma_K + \rho_K \ln \left( \ddfrac{\ddfrac{X^{K+1}(p)}{X^0(p)}}{\ddfrac{X^{K+1}(i)}{X^0(i)}} \right) \right) \leq
\end{align*}
\begin{align}
\leq \left( X^{K+1} \cdot C \bar{Y}^K \right)  \left( \sum_{i=1}^n X^{K+1}(i) \left(\sigma_K + \rho_K \ln \left( \ddfrac{\ddfrac{X^{K+1}(p)}{X^0(p)}}{\ddfrac{X^{K+1}(i)}{X^0(i)}} \right) \right) \right).\label{askdlsfdbvmzncvbbba}
\end{align}
Let us work first with the expression on the left-hand-side of the previous inequality. To that end, we have
\begin{align*}
\sum_{i=1}^n (C\bar{Y}^K)_i X^{K+1}(i) \left( \sigma_K + \rho_K \ln \left( \ddfrac{\ddfrac{X^{K+1}(p)}{X^0(p)}}{\ddfrac{X^{K+1}(i)}{X^0(i)}} \right) \right) =
\end{align*}
\begin{align*}
= \sum_{i=1}^n \sum_{j=1}^n C_{ij} \bar{Y}^K(j) X^{K+1}(i) \left(\sigma_K + \rho_K \ln \left( \ddfrac{\ddfrac{X^{K+1}(p)}{X^0(p)}}{\ddfrac{X^{K+1}(i)}{X^0(i)}} \right) \right).
\end{align*}
Assuming $p$ is not a best response to $X^{K+1}$ and letting
\begin{align*}
\ell^* \in \arg\min_{\ell =1}^n \left\{ \left( \sigma_K + \rho_K \ln \left( \ddfrac{\ddfrac{X^{K+1}(p)}{X^0(p)}}{\ddfrac{X^{K+1}(\ell)}{X^0(\ell)}} \right) \right) \right\},
\end{align*}
and further assuming that $\sigma_K > 0, \rho_K > 0$ are such that
\begin{align*}
\sigma_K + \rho_K \ln \left( \ddfrac{\ddfrac{X^{K+1}(p)}{X^0(p)}}{\ddfrac{X^{K+1}(\ell^*)}{X^0(\ell^*)}} \right) \geq 0,
\end{align*}
such that
\begin{align*}
\min_{\ell = 1}^n \left\{ \sigma_K - \rho_K \ln \left( \ddfrac{\ddfrac{X^{K+1}(p)}{X^0(p)}}{\ddfrac{X^{K+1}(\ell)}{X^0(\ell)}} \right) \right\} = 1,
\end{align*}
and such that
\begin{align*}
\sigma_K - \rho_K \ln \left( \ddfrac{\ddfrac{X^{K+1}(p)}{X^0(p)}}{\ddfrac{X^{K+1}(\ell^*)}{X^0(\ell^*)}} \right) = \max_{\ell = 1}^n \left\{ \sigma_K - \rho_K \ln \left( \ddfrac{\ddfrac{X^{K+1}(p)}{X^0(p)}}{\ddfrac{X^{K+1}(\ell)}{X^0(\ell)}} \right) \right\} \leq 2,
\end{align*}
and continuing from above, we obtain
\begin{align*}
\geq \frac{1}{2} \sum_{i=1}^n \sum_{j=1}^n C_{ij} \bar{Y}^K(j) X^{K+1}(i) \left( \sigma_K + \rho_K \ln \left( \ddfrac{\ddfrac{X^{K+1}(p)}{X^0(p)}}{\ddfrac{X^{K+1}(\ell^*)}{X^0(\ell^*)}} \right) \right) \left( \sigma_K - \rho_K \ln \left( \ddfrac{\ddfrac{X^{K+1}(p)}{X^0(p)}}{\ddfrac{X^{K+1}(\ell^*)}{X^0(\ell^*)}} \right) \right)
\end{align*}
which implies
\begin{align*}
= \frac{1}{2} \sum_{i=1}^n \sum_{j=1}^n C_{ij} \bar{Y}^K(j) X^{K+1}(i) \min_{\ell=1}^n \left\{ \left( \sigma_K + \rho_K \ln \left( \ddfrac{\ddfrac{X^{K+1}(p)}{X^0(p)}}{\ddfrac{X^{K+1}(\ell)}{X^0(\ell)}} \right) \right) \left( \sigma_K - \rho_K \ln \left( \ddfrac{\ddfrac{X^{K+1}(p)}{X^0(p)}}{\ddfrac{X^{K+1}(\ell^*)}{X^0(\ell^*)}} \right) \right) \right\}
\end{align*}
which, since 
\begin{align*}
\sigma_K - \rho_K \ln \left( \ddfrac{\ddfrac{X^{K+1}(p)}{X^0(p)}}{\ddfrac{X^{K+1}(\ell^*)}{X^0(\ell^*)}} \right) \geq \sigma_K - \rho_K \ln \left( \ddfrac{\ddfrac{X^{K+1}(p)}{X^0(p)}}{\ddfrac{X^{K+1}(\ell)}{X^0(\ell)}} \right), 
\end{align*}
further implies
\begin{align*}
\geq \frac{1}{2} \sum_{i=1}^n \sum_{j=1}^n C_{ij} \bar{Y}^K(j) X^{K+1}(i) \min_{\ell=1}^n \left\{ \left( \sigma_K + \rho_K \ln \left( \ddfrac{\ddfrac{X^{K+1}(p)}{X^0(p)}}{\ddfrac{X^{K+1}(\ell)}{X^0(\ell)}} \right) \right) \left( \sigma_K - \rho_K \ln \left( \ddfrac{\ddfrac{X^{K+1}(p)}{X^0(p)}}{\ddfrac{X^{K+1}(\ell)}{X^0(\ell)}} \right) \right) \right\}
\end{align*}
which even further implies that
\begin{align*}
= \frac{1}{2} \sum_{i=1}^n \sum_{j=1}^n C_{ij} \bar{Y}^K(j) X^{K+1}(i) \min_{\ell=1}^n \left\{ \left( \sigma^2_K - \rho_K^2 \left( \ln \left( \ddfrac{\ddfrac{X^{K+1}(p)}{X^0(p)}}{\ddfrac{X^{K+1}(\ell)}{X^0(\ell)}} \right) \right)^2 \right) \right\}
\end{align*}
which even further implies
\begin{align*}
= \frac{1}{2} \sum_{i=1}^n \sum_{j=1}^n C_{ij} \bar{Y}^K(j) X^{K+1}(i) \left( \min_{\ell=1}^n \left\{ \sigma^2_K - \rho_K^2 \left( \ln \left( \ddfrac{\ddfrac{X^{K+1}(p)}{X^0(p)}}{\ddfrac{X^{K+1}(\ell)}{X^0(\ell)}} \right) \right)^2 \right\} \right) 
\end{align*}
which even further implies
\begin{align*}
= \frac{1}{2} \sum_{i=1}^n \sum_{j=1}^n C_{ij} \bar{Y}^K(j) X^{K+1}(i) \left( \sigma^2_K - \rho_K^2 \max_{\ell=1}^n \left\{  \left( \ln \left( \ddfrac{\ddfrac{X^{K+1}(p)}{X^0(p)}}{\ddfrac{X^{K+1}(\ell)}{X^0(\ell)}} \right) \right)^2 \right\} \right) 
\end{align*}
which even further implies by \eqref{cond}
\begin{align*}
= \frac{1}{2} \sum_{i=1}^n \sum_{j=1}^n C_{ij} \bar{Y}^K(j) X^{K+1}(i) \left( \sigma^2_K - \rho_K^2  \left( \max_{\ell=1}^n \left\{ \ln \left( \ddfrac{\ddfrac{X^{K+1}(p)}{X^0(p)}}{\ddfrac{X^{K+1}(\ell)}{X^0(\ell)}} \right) \right\} \right)^2  \right) 
\end{align*}
which even further implies 
\begin{align*}
= \frac{1}{2} \sum_{i=1}^n \sum_{j=1}^n C_{ij} \bar{Y}^K(j) X^{K+1}(i) \left( \sigma^2_K -  \left( \rho_K \max_{\ell=1}^n \left\{ \ln \left( \ddfrac{\ddfrac{X^{K+1}(p)}{X^0(p)}}{\ddfrac{X^{K+1}(\ell)}{X^0(\ell)}} \right) \right\} \right)^2  \right) 
\end{align*}
which even further implies
\begin{align*}
= \frac{1}{2} \sum_{i=1}^n \sum_{j=1}^n C_{ij} \bar{Y}^K(j) X^{K+1}(i) \left( \left( \sigma_K - \rho_K \max_{\ell=1}^n \left\{ \ln \left( \ddfrac{\ddfrac{X^{K+1}(p)}{X^0(p)}}{\ddfrac{X^{K+1}(\ell)}{X^0(\ell)}} \right) \right\} \right) \left( \sigma_K + \rho_K \max_{\ell=1}^n \left\{ \ln \left( \ddfrac{\ddfrac{X^{K+1}(p)}{X^0(p)}}{\ddfrac{X^{K+1}(\ell)}{X^0(\ell)}} \right) \right\} \right)  \right) 
\end{align*}
which even further implies
\begin{align*}
= \frac{1}{2} \sum_{i=1}^n \sum_{j=1}^n C_{ij} \bar{Y}^K(j) X^{K+1}(i) \left( \min_{\ell = 1}^n \left\{ \sigma_K - \rho_K \ln \left( \ddfrac{\ddfrac{X^{K+1}(p)}{X^0(p)}}{\ddfrac{X^{K+1}(\ell)}{X^0(\ell)}} \right) \right\} \left( \sigma_K + \rho_K \max_{\ell=1}^n \left\{ \ln \left( \ddfrac{\ddfrac{X^{K+1}(p)}{X^0(p)}}{\ddfrac{X^{K+1}(\ell)}{X^0(\ell)}} \right) \right\} \right)  \right) 
\end{align*}
which, under the previous assumption that
\begin{align*}
\min_{\ell = 1}^n \left\{ \sigma_K - \rho_K \ln \left( \ddfrac{\ddfrac{X^{K+1}(p)}{X^0(p)}}{\ddfrac{X^{K+1}(\ell)}{X^0(\ell)}} \right) \right\} = 1,
\end{align*}
even further implies that
\begin{align*}
= \frac{1}{2} \sum_{i=1}^n \sum_{j=1}^n C_{ij} \bar{Y}^K(j) X^{K+1}(i) \left( \sigma_K + \rho_K \max_{\ell=1}^n \left\{ \ln \left( \ddfrac{\ddfrac{X^{K+1}(p)}{X^0(p)}}{\ddfrac{X^{K+1}(\ell)}{X^0(\ell)}} \right) \right\} \right) 
\end{align*}
which even further implies
\begin{align*}
\geq \frac{1}{2} \sum_{i=1}^n \sum_{j=1}^n C_{ij} \bar{Y}^K(j) X^{K+1}(i) \left( \sigma_K + \rho_K \ln \left( \ddfrac{\ddfrac{X^{K+1}(p)}{X^0(p)}}{\ddfrac{X^{K+1}(j)}{X^0(j)}} \right) \right).
\end{align*}
Combining the previous inequality with \eqref{askdlsfdbvmzncvbbba}, we obtain
\begin{align*}
\sum_{i =1}^n \sum_{j=1}^n C_{ij} \bar{Y}^K(j) X^{K+1}(i) \left( \sigma_K + \rho_K \ln \left( \ddfrac{\ddfrac{X^{K+1}(p)}{X^0(p)}}{\ddfrac{X^{K+1}(j)}{X^0(j)}} \right) \right)  \leq
\end{align*}
\begin{align*}
\leq 2 \left( X^{K+1} \cdot C \bar{Y}^K \right)  \left( \sum_{i=1}^n X^{K+1}(i) \left(\sigma_K + \rho_K \ln \left( \ddfrac{\ddfrac{X^{K+1}(p)}{X^0(p)}}{\ddfrac{X^{K+1}(i)}{X^0(i)}} \right) \right)  \right)
\end{align*}
which further implies
\begin{align*}
(CX^{K+1})_{\min} \sum_{j=1}^n \bar{Y}^K(j) \left( \sigma_K + \rho_K \ln \left( \ddfrac{\ddfrac{X^{K+1}(p)}{X^0(p)}}{\ddfrac{X^{K+1}(j)}{X^0(j)}} \right) \right)  \leq
\end{align*}
\begin{align*}
\leq 2 \left( X^{K+1} \cdot C \bar{Y}^K \right)  \left( \sum_{i=1}^n X^{K+1}(i) \left(\sigma_K + \rho_K \ln \left( \ddfrac{\ddfrac{X^{K+1}(p)}{X^0(p)}}{\ddfrac{X^{K+1}(i)}{X^0(i)}} \right) \right)  \right)
\end{align*}
which even further implies
\begin{align*}
\sum_{j=1}^n \bar{Y}^K(j) \left( \sigma_K + \rho_K \ln \left( \ddfrac{\ddfrac{X^{K+1}(p)}{X^0(p)}}{\ddfrac{X^{K+1}(j)}{X^0(j)}} \right) \right)  \leq
\end{align*}
\begin{align*}
\leq 2 \left( \frac{X^{K+1} \cdot C \bar{Y}^K}{(CX^{K+1})_{\min}} \right)  \left( \sum_{i=1}^n X^{K+1}(i) \left(\sigma_K + \rho_K \ln \left( \ddfrac{\ddfrac{X^{K+1}(p)}{X^0(p)}}{\ddfrac{X^{K+1}(i)}{X^0(i)}} \right) \right)  \right)
\end{align*}
which even further implies
\begin{align*}
\sum_{j=1}^n \bar{Y}^K(j) \ln \left( \ddfrac{\ddfrac{X^{K+1}(p)}{X^0(p)}}{\ddfrac{X^{K+1}(j)}{X^0(j)}} \right) \leq
\end{align*}
\begin{align*}
\leq \frac{2}{\rho_K} \left( \left( \frac{X^{K+1} \cdot C \bar{Y}^K}{(CX^{K+1})_{\min}} \right)  \left( \sum_{i=1}^n X^{K+1}(i) \left(\sigma_K + \rho_K \ln \left( \ddfrac{\ddfrac{X^{K+1}(p)}{X^0(p)}}{\ddfrac{X^{K+1}(i)}{X^0(i)}} \right) \right) \right) - \sigma_K \right)
\end{align*}
which even further implies
\begin{align*}
\sum_{j=1}^n \bar{Y}^K(j) \ln \left( \ddfrac{\ddfrac{X^{K+1}(p)}{X^0(p)}}{\ddfrac{X^{K+1}(j)}{X^0(j)}} \right) \leq 2 \left( \frac{X^{K+1} \cdot C \bar{Y}^K}{(CX^{K+1})_{\min}} \right) \ln \left( \ddfrac{X^{K+1}(p)}{X^0(p)} \right).
\end{align*}
Combining the previous inequality with \eqref{enaaa} we obtain
\begin{align*}
(C\bar{Y}^K)_p - \bar{Y}^K \cdot C \bar{Y}^K \leq \left( \frac{X^{K+1} \cdot C \bar{Y}^K}{(CX^{K+1})_{\min}} \right) \ln \left( \ddfrac{X^{K+1}(p)}{X^0(p)} \right) \ddfrac{2}{A_K}
\end{align*}
which implies by the assumption $C$ satisfies
\begin{align*}
\ddfrac{\max_{ij} C_{ij}}{\min_{ij} C_{ij}} = c
\end{align*}
that
\begin{align*}
\ln \left( \ddfrac{X^{K+1}(p)}{X^0(p)} \right) > 0 \Rightarrow (C\bar{Y}^K)_p - \bar{Y}^K \cdot C \bar{Y}^K \leq \ln \left( \ddfrac{X^{K+1}(p)}{X^0(p)} \right) \ddfrac{2c}{A_K}
\end{align*}
\begin{align*}
\ln \left( \ddfrac{X^{K+1}(p)}{X^0(p)} \right) < 0 \Rightarrow (C\bar{Y}^K)_p - \bar{Y}^K \cdot C \bar{Y}^K \leq \ln \left( \ddfrac{X^{K+1}(p)}{X^0(p)} \right) \ddfrac{2/c}{A_K}
\end{align*}
as claimed in the statement of the theorem. Let us now verify that there exist $\sigma_K > 0$ and $\rho_K > 0$ such that
\begin{align*}
\min_{\ell = 1}^n \left\{ \sigma_K + \rho_K \ln \left( \ddfrac{\ddfrac{X^{K+1}(p)}{X^0(p)}}{\ddfrac{X^{K+1}(\ell)}{X^0(\ell)}} \right) \right\} \geq 0,
\end{align*}
\begin{align*}
\min_{\ell = 1}^n \left\{ \sigma_K - \rho_K \ln \left( \ddfrac{\ddfrac{X^{K+1}(p)}{X^0(p)}}{\ddfrac{X^{K+1}(\ell)}{X^0(\ell)}} \right) \right\} = 1 \Leftrightarrow \min_{\ell = 1}^n \left\{ \sigma_K + \rho_K \ln \left( \ddfrac{\ddfrac{X^{K+1}(\ell)}{X^0(\ell)}}{\ddfrac{X^{K+1}(p)}{X^0(p)}} \right) \right\} = 1,
\end{align*}
and
\begin{align*}
\max_{\ell = 1}^n \left\{ \sigma_K - \rho_K \ln \left( \ddfrac{\ddfrac{X^{K+1}(p)}{X^0(p)}}{\ddfrac{X^{K+1}(\ell)}{X^0(\ell)}} \right) \right\} \leq 2 \Leftrightarrow \max_{\ell = 1}^n \left\{ \sigma_K + \rho_K \ln \left( \ddfrac{\ddfrac{X^{K+1}(\ell)}{X^0(\ell)}}{\ddfrac{X^{K+1}(p)}{X^0(p)}} \right) \right\} \leq 2.
\end{align*}
By the assumption
\begin{align*}
\liminf\limits_{K \rightarrow \infty} \left\{ \ln \left( \ddfrac{X^{K+1}(p)}{X^0(p)} \right) \right\} > - \infty
\end{align*}
the quantity
\begin{align*}
\ln \left( \ddfrac{X^{K+1}(p)}{X^0(p)} \right)
\end{align*}
is bounded away from $-\infty$. Choosing $\rho_K < \gamma \sigma_K$, to satisfy the previous inequalities it suffices to find $\rho_K$, $\sigma_K$, and $\gamma$ such that
\begin{align}
1 + \gamma \min_{\ell = 1}^n \left\{ \ln \left( \ddfrac{\ddfrac{X^{K+1}(p)}{X^0(p)}}{\ddfrac{X^{K+1}(\ell)}{X^0(\ell)}} \right) \right\} \geq 0,\label{a1}
\end{align}
\begin{align}
\sigma_K + \rho_K \min_{\ell = 1}^n \left\{ \ln \left( \ddfrac{\ddfrac{X^{K+1}(\ell)}{X^0(\ell)}}{\ddfrac{X^{K+1}(p)}{X^0(p)}} \right) \right\} = 1,\label{a2}
\end{align}
and
\begin{align}
\sigma_K \left(1 + \gamma \max_{\ell = 1}^n \left\{ \ln \left( \ddfrac{\ddfrac{X^{K+1}(\ell)}{X^0(\ell)}}{\ddfrac{X^{K+1}(p)}{X^0(p)}} \right) \right\} \right) \leq 2.\label{a3}
\end{align}
Choosing $\gamma$ arbitrarily close to $0$ (but bounded away from $0$) satisfies the first inequality, since
\begin{align*}
\liminf\limits_{K \rightarrow \infty} \left\{ \min_{\ell = 1}^n \left\{ \ln \left( \ddfrac{\ddfrac{X^{K+1}(p)}{X^0(p)}}{\ddfrac{X^{K+1}(\ell)}{X^0(\ell)}} \right) \right\} \right\} > - \infty.
\end{align*}
Further choosing $\sigma_K \leq 2$ also satisfies the third inequality, since
\begin{align*}
\max_{\ell = 1}^n \left\{ \ln \left( \ddfrac{\ddfrac{X^{K+1}(\ell)}{X^0(\ell)}}{\ddfrac{X^{K+1}(p)}{X^0(p)}} \right) \right\} \geq 0.
\end{align*}
Further choosing $\rho_K$ arbitrarily close to $0$ and $\sigma_K$ arbitrarily close to $1$ from above satisfies the second inequality. This completes the proof.\\

{\em (ii)} It follows by straight algebra that \eqref{a1}, \eqref{a2}, and \eqref{a3} are also satisfied.\\

{\em (iii)} That \eqref{cond2} implies \eqref{cond} follows by the symmetry of $f(x) = x^2$. That \eqref{cond2} is equivalent to \eqref{cond3} follows by Lemma \ref{nonuniform_lemma} as used in the following equivalence:

\begin{align*}
\max_{\ell=1}^n \left\{ \ln \left( \ddfrac{\ddfrac{X^{K+1}(\ell)}{X^0(\ell)}}{\ddfrac{X^{K+1}(p)}{X^0(p)}} \right) \right\} \leq \max_{\ell=1}^n \left\{ \ln \left( \ddfrac{\ddfrac{X^{K+1}(p)}{X^0(p)}}{\ddfrac{X^{K+1}(\ell)}{X^0(\ell)}} \right) \right\}
\end{align*}

\begin{align*}
\max_{i=1}^n \left\{ (C\bar{Y})_i \right\} - (C\bar{Y})_p \leq (C\bar{Y})_p - \min_{i=1}^n \left\{ (C\bar{Y})_i \right\}
\end{align*}

\begin{align*}
(C\bar{Y})_p \geq \frac{1}{2} \left( \max_{i=1}^n \left\{ (C\bar{Y})_i \right\} + \min_{i=1}^n \left\{ (C\bar{Y})_i \right\} \right)
\end{align*}

\end{proof}

\section{The multiplicative weights convexity lemma}

Let $(C, C^T)$ be a symmetric bimatrix game,
\begin{align*}
T_i(X) = X(i) \cdot \frac{\exp\left\{ \alpha (CX)_i \right\}}{ \sum_{j=1}^n X(j) \exp \left\{ \alpha (CX)_j \right\} } \quad i = 1, \ldots, n,
\end{align*}
and $RE(\cdot, \cdot)$ denote the relative entropy function between probability vectors. We then have that

\begin{lemma}[\citep{Avramopoulos2}]
\label{convexity_lemma}
\begin{align*}
\forall X \in \mathbb{\mathring{X}}(C) \mbox{ } \forall Y \in \mathbb{X}(C) : RE(Y, T(X)) \mbox{ is a convex function of }\alpha.
\end{align*}
\end{lemma}

\begin{proof}
Let $\hat{X} \equiv T(X)$. We have
{\allowdisplaybreaks
\begin{align*}
\frac{d}{d \alpha}& RE(Y, \hat{X}) = \\
  &= \frac{d}{d \alpha} \left( \sum_{i \in \mathcal{C}(Y)} Y(i) \ln\left( \frac{Y(i)}{\hat{X}(i)} \right) \right)\\
  &= \frac{d}{d \alpha} \left( \sum_{i \in \mathcal{C}(Y)} Y(i) \ln\left( Y(i) 
  \cdot \frac{\sum_{j = 1}^n X(j) \exp\{ \alpha (CX)_j \}}{ X(i) \exp\{ \alpha (CX)_i \}} \right) \right)\\
  &= \frac{d}{d \alpha} \left( \sum_{i \in \mathcal{C}(Y)} Y(i) \ln\left( \frac{\sum_{j = 1}^n X(j) \exp\{ \alpha (CX)_j \}}{X(i)\exp\{ \alpha (CX)_i \}} \right) \right)\\
  &= \sum_{i \in \mathcal{C}(Y)} Y(i) \frac{d}{d \alpha} \left( \ln\left( \frac{\sum_{j = 1}^n X(j) \exp\{ \alpha (CX)_j \}}{X(i)\exp\{ \alpha (CX)_i \}} \right) \right).
\end{align*}}
Furthermore, using $(\cdot)'$ as alternative notation (abbreviation) for $d/d\alpha(\cdot)$,
\begin{align*}
\frac{d}{d \alpha} \left( \ln\left( \frac{\sum_{j = 1}^n X(j) \exp\{ \alpha (CX)_j \}}{X(i)\exp\{ \alpha (CX)_i \}} \right) \right) = \frac{X(i)\exp\{ \alpha (CX)_i \}}{\sum_{j = 1}^n X(j) \exp\{ \alpha (CX)_j \}} \left( \frac{\sum_{j = 1}^n X(j) \exp\{ \alpha (CX)_j \}}{X(i)\exp\{ \alpha (CX)_i \}} \right)'
\end{align*}
and
\begin{align*}
\left( \frac{\sum_{j = 1}^n X(j) \exp\{ \alpha (CX)_j \}}{X(i)\exp\{ \alpha (CX)_i \}} \right)'  &= \frac{\sum_{j = 1}^n X(j) (CX)_j \exp\{ \alpha (CX)_j \} X(i) \exp\{ \alpha (CX)_i \}}{\left( X(i) \exp\{ \alpha (CX)_i \} \right)^2} -\\
  &- \frac{X(i) (CX)_i \sum_{j = 1}^n X(j) \exp\{ \alpha (CX)_j \} \exp\{ \alpha (CX)_i \}}{\left( X(i) \exp\{ \alpha (CX)_i \} \right)^2} =\\
  &= \frac{\sum_{j = 1}^n X(j) (CX)_j \exp\{ \alpha (CX)_j \} \}}{ X(i) \exp\{ \alpha (CX)_i \} } -\\
  &- \frac{(CX)_i \sum_{j = 1}^n X(j) \exp\{ \alpha (CX)_j \} \}}{ X(i) \exp\{ \alpha (CX)_i \} }.
\end{align*}
Therefore,
\begin{align}
\frac{d}{d \alpha} RE(Y, \hat{X}) = \frac{\sum_{j = 1}^n X(j) (CX)_j \exp\{ \alpha (CX)_j \}}{\sum_{j = 1}^n X(j) \exp\{ \alpha (CX)_j \}} - Y \cdot CX.\label{valentine}
\end{align}
Furthermore,
\begin{align*}
\frac{d^2}{d \alpha^2} RE(Y, \hat{X}) &= \frac{\left( \sum_{j = 1}^n X(j) ((CX)_j)^2 \exp\{ \alpha (CX)_j \} \right) \left( \sum_{j = 1}^n X(j) \exp\{ \alpha (CX)_j \} \right)}{\left( \sum_{j = 1}^n X(j) \exp\{ \alpha (CX)_j \} \right)^2} -\\
  &- \frac{\left( \sum_{j = 1}^n X(j) ((CX)_j) \exp\{ \alpha (CX)_j \} \right)^2}{\left( \sum_{j = 1}^n X(j) \exp\{ \alpha (CX)_j \} \right)^2}.
\end{align*}
Jensen's inequality implies that
\begin{align*}
\frac{\sum_{j = 1}^n X(j) ((CX)_j)^2 \exp\{ \alpha (CX)_j \}}{\sum_{j = 1}^n X(j) \exp\{ \alpha (CX)_j \}} \geq \left( \frac{\sum_{j = 1}^n X(j) ((CX)_j) \exp\{ \alpha (CX)_j \}}{\sum_{j = 1}^n X(j) \exp\{ \alpha (CX)_j \}} \right)^2,
\end{align*}
which is equivalent to the numerator of the second derivative being nonnegative as $X$ is a probability vector. Note that the inequality is strict unless 
\begin{align*}
\forall i, j \in \mathcal{C}(X) : (CX)_i = (CX)_j.
\end{align*}
This completes the proof.
\end{proof}

\begin{figure}[tb]
\centering
\includegraphics[width=10cm]{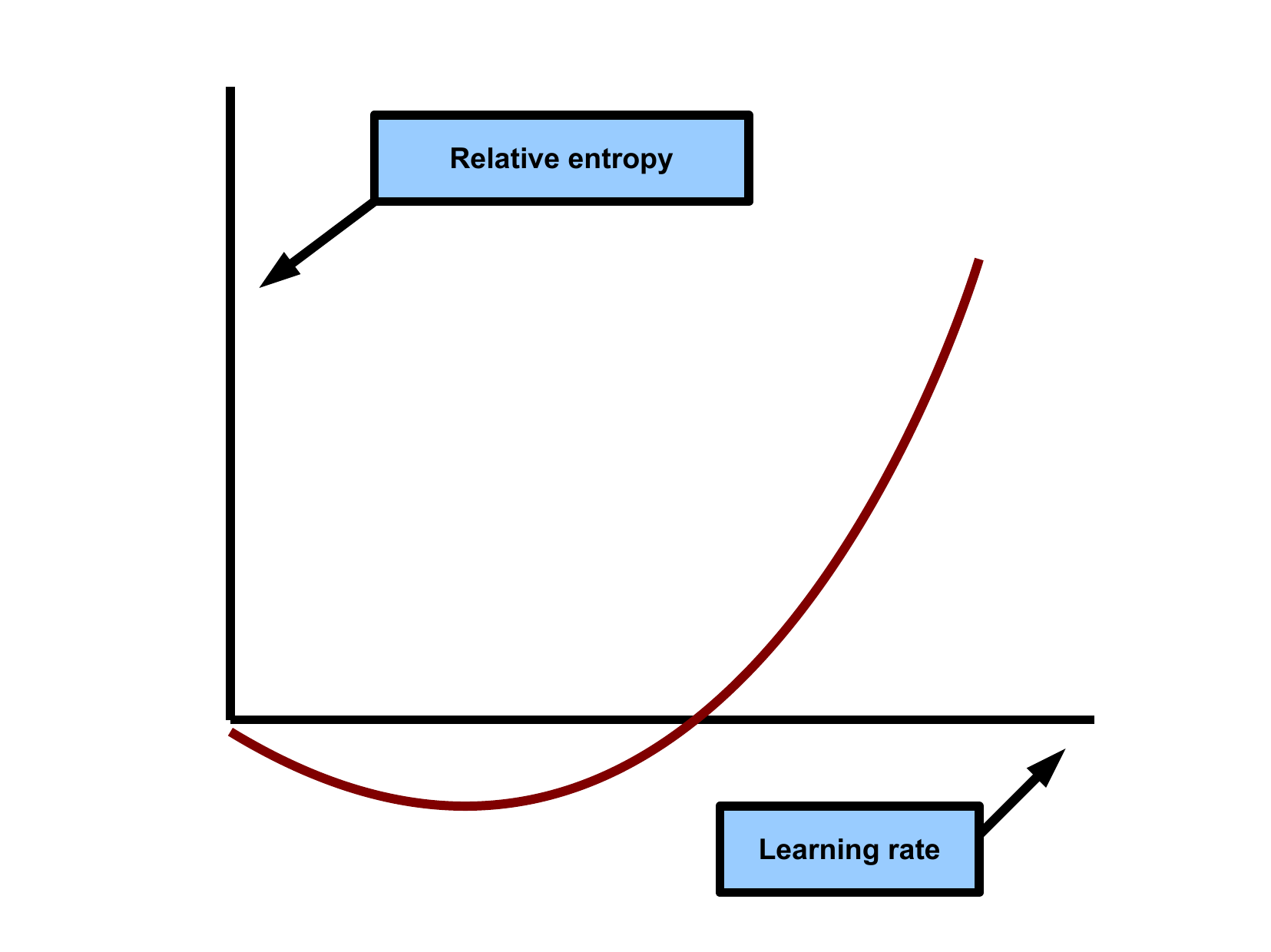}
\caption{\label{qpwoeirxwoexururxs}
$RE(X^*, T(X)) - RE(X^*, X)$ as a function of the learning rate $\alpha$.}
\end{figure}

\begin{lemma}
\label{convexity_lemma_2}
Let $C \in \mathbb{\hat{C}}$. Then, for all $Y \in \mathbb{X}(C)$ and for all $X \in \mathbb{\mathring{X}}(C)$, we have that
\begin{align*}
\forall \alpha > 0 : RE(Y, T(X)) \leq RE(Y, X) - \alpha (Y-X) \cdot CX + \alpha (\exp\{\alpha\} - 1) \bar{C},
\end{align*}
where $0 < \bar{C} < 1$ can be chosen independent of $X$ and $Y$.
\end{lemma}

\begin{proof}
Since, by Lemma \ref{convexity_lemma}, $RE(Y, T(X)) - RE(Y, X)$ is a convex function of $\alpha$, we have by the aforementioned secant inequality that, for $\alpha > 0$,
\begin{align}
RE(Y, T(X)) - RE(Y, X) \leq \alpha \left( RE(Y, T(X)) - RE(Y, X) \right)' = \alpha \cdot \frac{d}{d \alpha} RE(Y, T(X)).\label{ooone}
\end{align}
Straight calculus (cf. Lemma \ref{convexity_lemma}) implies that
\begin{align*}
\frac{d}{d \alpha} RE(Y, T(X)) = \frac{\sum_{j = 1}^n X(j) (CX)_j \exp\{ \alpha (CX)_j \}}{\sum_{j = 1}^n X(j) \exp\{ \alpha (CX)_j \}} - Y \cdot CX.
\end{align*}
Using Jensen's inequality in the previous expression, we obtain 
\begin{align}
\frac{d}{d \alpha} RE(Y, T(X)) \leq \frac{\sum_{j = 1}^n X(j) (CX)_j \exp\{ \alpha (CX)_j \}}{\exp\{ \alpha X \cdot CX \}} - Y \cdot CX.\label{vbvbvb}
\end{align}
Note now that
\begin{align}
\exp\{ \alpha x \} \leq 1 + (\exp\{ \alpha \} - 1) x, x \in [0, 1],\label{freund_schapire_in}
\end{align}
an inequality used in \cite[Lemma 2]{FreundSchapire2}. Using $C \in \mathbb{\hat{C}}$, \eqref{vbvbvb} and \eqref{freund_schapire_in} imply that
\begin{align*}
\frac{d}{d \alpha} RE(Y, T(X)) \leq \frac{X \cdot CX}{\exp\{ \alpha X \cdot CX \}} - Y \cdot CX + \left( \exp\{ \alpha \} - 1 \right) \frac{\sum_{j=1}^n X(j) (CX)_j^2}{\exp\{ \alpha X \cdot CX \}}
\end{align*}
and since $\exp\{\alpha X \cdot CX\} \geq 1$ (again by the assumption that $C \in \mathbb{\hat{C}}$), we have
\begin{align*}
\frac{d}{d \alpha} RE(Y, T(X)) \leq X \cdot CX - Y \cdot CX + \left( \exp\{ \alpha \} - 1 \right) \sum_{j=1}^n X(j) (CX)_j^2.
\end{align*}
Choosing $\bar{C} = \max \left\{\sum X(j) (CX)_j^2 \right\}$ and combining with \eqref{ooone} yields the lemma.
\end{proof}

\section{Repelling non-equilibrium fixed points}
\label{repelling}

\subsection{Repelling fixed points}

Recall that $X^*$ is a fixed point of $T$ if $T(X^*) = X^*$. We may assume $T$ is as in Section \ref{Hedge}. We are interested in the behavior of $T$ starting from an interior strategy of the probability simplex.

\begin{definition}
Let $T : \mathbb{X} \subset \mathbb{R}^n \rightarrow \mathbb{X}$ such that $X^*$ is a fixed point of $T$. $X^*$ is {\em repelling} under $T$ if, for all $X_0$ in the relative interior of $\mathbb{X}$,
\begin{align*}
\lim_{k \rightarrow \infty} T^k(X_0) \neq X^*.
\end{align*}
\end{definition}

\begin{lemma}
\label{repelling_lemma}
Suppose $X^*$ is a fixed point of $T$. If $\dot{V}$ and $V$ are positive definite with respect to $T$ and $X^*$, $X^*$ is repelling under $T$.
\end{lemma}

\begin{proof}
Suppose there exists $X_0$ such that $T^k(X_0)$ converges to $X^*$. Since $V$ is continuous, $V(T^k(X_0))$ also converges, in fact, to $0$. Consider any $k$ large enough that such $T^k(X_0)$ is close to $X^*$. Since, for all $X$ in a neighborhood of $X^*$, $\dot{V}(X) > 0$, for all $\hat{k} > k$, $V(T^{\hat{k}}(X_0)) > V(T^k(X_0))$, which is a positive constant bounded away from $0$, contradicting convergence of $V(T^k(X_0))$ to $0$.
\end{proof}

\subsection{Non-equilibrium fixed points are repelling}

Starting off with preliminaries, the following propositions are standard \citep{Econ_analysis}.

\begin{proposition}
Let $f : \mathbb{O} \rightarrow \mathbb{O}'$ be a continuous function between the topological spaces $\mathbb{O}$ and $\mathbb{O}'$. Then if $\mathbb{O}$ is connected, $f(\mathbb{O}) \subset \mathbb{O}'$ is connected.
\end{proposition}

The following proposition is known as the {\em intermediate value theorem}.

\begin{proposition}
If $f : |a, b| \subset \mathbb{R} \rightarrow \mathbb{R}$, where $|a, b|$ is an interval, is continuous, $y, y'' \in f(|a, b|)$ and $y < y' < y''$, then there exists $x \in |a, b|$ such that $f(x) = y'$.
\end{proposition}

A slightly more general version can be stated as follows.

\begin{proposition}
\label{gv_int_v_theorem}
Suppose $\mathbb{O}$ is connected and $f : \mathbb{O} \rightarrow \mathbb{R}$ is continuous. If $a, b \in \mathbb{O}$ and $f(a) < y < f(b)$, there exists $x \in \mathbb{O}$ such that $f(x) = y$. 
\end{proposition}

The following lemma is a straightforward implication of the previous propositions.

\begin{lemma}
\label{fp_instability_lemma}
Let $f : O \subset \mathbb{R}^n \rightarrow \mathbb{R}$ be a continuous map such $O$ is a neighborhood of $X^* \in \mathbb{R}^n$. If $f(X^*) > 0$, there exists a neighborhood $O' \subset O$ of $X^*$ such that, for all $X \in O'$, $f(X) > 0$.
\end{lemma}

\begin{proof}
If, for all $X \in O$, $f(X) \neq 0$, the lemma is trivially true by letting $O' = O$. Suppose, therefore, there exists $X \in O$ such that $f(X) = 0$ and let
\begin{align*}
\mathbb{F} = \{X \in O | f(X) > 0 \}.
\end{align*}
$\mathbb{F}$ is an open set containing $X^*$ and, therefore, there exists a neighborhood $O' \subset \mathbb{F}$ of $X^*$ such that, for all $X \in O'$, $f(X) > 0$. Since $f$ is continuous and $O'$ is connected, $f(O')$ is connected, and the intermediate value theorem implies that, for all $X \in O'$, $f(X) > 0$. For if $0 \in f(O')$, then, since $f(O')$ is open (since $O'$ is open and $f$ is continuous) and connected, there exists $X \in O'$ such that $f(X) < 0$, and Proposition \ref{gv_int_v_theorem} implies $\exists Y \in O' : f(Y) = 0$, contradicting the existence of $O'$.
\end{proof}

\begin{lemma}
\label{yannis_lemma}
Non-equilibrium fixed points are repelling under $T$ for small enough values of the learning rate.
\end{lemma}

\begin{proof}
Let $X^*$ be a non-equilibrium fixed point of $T$ and
\begin{align*}
i \in \arg\max \{ E_j \cdot CX^* | j \in \mathcal{K}(C) \}.
\end{align*}
By the assumption $X^* \not\in NE^+(C, C^T)$, $E_i \cdot CX^* > X^* \cdot CX^*$. Now let
\begin{align*}
V(X) = X(i) - X^*(i) = X(i)
\end{align*}
since $X^*(i) = 0$. Note $V(X^*) = 0$. Letting $\hat{X} = T(X)$, 
\begin{align*}
\dot{V}(X) &= V(\hat{X}) - V(X)\\
  &= \hat{X}(i) - X(i)\\
  &= X(i) \cdot \frac{\exp\left\{ \alpha E_i \cdot CX \right\}}{ \sum_{j=1}^n X(j) \exp \left\{ \alpha E_j \cdot CX \right\} } - X(i).
\end{align*} 
Let 
\begin{align*}
f(\alpha) \equiv \frac{\exp\left\{ \alpha E_{i} \cdot CX \right\}}{ \sum_{j=1}^n X(j) \exp \left\{ \alpha E_j \cdot CX \right\} } \equiv \frac{g(\alpha)}{h(\alpha)}.
\end{align*}
Letting $(CX)_j \equiv E_j \cdot CX$, we have
\begin{align}
\frac{d f}{d \alpha} = \frac{ (CX)_i g(\alpha) h(\alpha) - g(\alpha) \left( \sum_{j=1}^n X(j) (CX)_j \exp \left\{ \alpha (CX)_j \right\} \right) }{ h^2(\alpha)},\label{cutie_eq}
\end{align}
which implies
\begin{align*}
\left. \frac{df}{d \alpha} \right|_{\alpha = 0} = E_{i} \cdot CX - X \cdot CX.
\end{align*}
Therefore, since $E_i \cdot CX^* > X^* \cdot CX^*$, Lemma \ref{fp_instability_lemma} implies
\begin{align*}
\exists O \mbox{ } \forall X \in O/\{X^*\} : E_{i} \cdot CX > X \cdot CX.
\end{align*}
where $O$ is a neighborhood of $X^*$. Therefore, for all $X \in O$,
\begin{align*}
\left. \frac{d f}{d \alpha} \right|_{\alpha = 0} > 0.
\end{align*} 
Therefore, 
\begin{align*}
\forall X \in O \mbox{ } \exists \hat{\alpha} > 0 \mbox{ } \forall 0 < \alpha \leq \hat{\alpha}  : \dot{V}(X) > 0
\end{align*}
and Lemma \ref{repelling_lemma} completes the proof. 
\end{proof}

\subsection{Considering our maximum-clique computation algorithm}

Considering Ariadne, our maximum-clique computation algorithm, to ensure $T$ does not converge to a fixed point whose maximum payoff is equal to or greater than $C_{00} + \epsilon$, we ensure that
\begin{align*}
\dot{V}(X) = X(i) \cdot \ddfrac{\exp\left\{ \alpha \max_{i=1}^n \left\{ (CX)_i \right\} \right\}}{ \sum_{j=1}^n X(j) \exp \left\{ \alpha E_j \cdot CX \right\} } - X(i) > 0
\end{align*} 
or, equivalently, that
\begin{align*}
\exp\left\{ \alpha \max_{i=1}^n \left\{ (CX)_i \right\} \right\} - \sum_{j=1}^n X(j) \exp \left\{ \alpha (CX)_j \right\} > 0.
\end{align*} 
To that end, it suffices to choose $\alpha$ such that 
\begin{align*}
1 + \alpha \max_{i=1}^n \left\{ (CX)_i \right\} - 1 - (\exp\{\alpha\}-1) X \cdot CX > 0
\end{align*}
where in the previous inequality we have used \eqref{freund_schapire_in}. Cancelling terms, we obtain
\begin{align*}
\alpha \max_{i=1}^n \left\{ (CX)_i \right\} - (\exp\{\alpha\}-1) X \cdot CX > 0
\end{align*}
and using our bounds, we finally obtain
\begin{align}
\alpha (C_{00} + \epsilon) - (\exp\{\alpha\}-1) C_{00} > 0.\label{aaaaaaaaaa}
\end{align}
That is, provided $\alpha$ satisfies the latter inequality, $T$ cannot be attracted by a fixed point $X^*$ such that 
\begin{align*}
\max_{i=1}^n \left\{ (CX)_i \right\} \geq C_{00} + \epsilon.
\end{align*}

\section{Pseudocode}
\label{pseudocode}

In the pseudocode given below, Ariadne starts with the largest possible value of the Nisan parameter $k$ (which is equal to $n$, the number of vertices of the underlying graph $G$) and iteratively decreases the value of $k$ upon failure to compute a maximal clique clique of size equal to $k$. Once the sequence of iterates enters the effective interior of the lower feasibility set (that is, in other words, upon successful initialization of Ariadne's primary dynamical system), Ariadne uses an adaptive learning rate $\alpha$ and an adaptive parameter $\mathsf{C}$. It also computes a sequence of multipliers by a transformation of the (primary) sequence of iterates and using these multipliers computes a secondary sequence of iterates. To avoid cluttering notation we denote approximate multipliers by $Y$ (instead of $\tilde{Y}$). Ariadne's dynamical system iterates until either an equilibrium is computed by the empirical average of the sequence of multipliers or until the upper bound on the number of iterations to attain a desired equilibrium approximation error is violated. Both conditions are met in a polynomial number of iterations (cf. Lemma \ref{theorem1}). Unless the bound is violated, Ariadne computes an approximate well-supported equilibrium (cf. Proposition \ref{cgood}) using the empirical average of the sequence of multipliers and checks if the carrier of the well-supported equilibrium carries a fixed point of the replicator dynamic and if that fixed point is the characteristic vector of a clique of size equal to the Nisan parameter. If the Nisan parameter is equal to the clique number this test cannot fail (cf. Lemma \ref{theorem1}). Upon computation of an equilibrium that is a clique of size smaller than the Nisan parameter or a non-clique equilibrium or upon violation of the upper bound on the number of iterations, Ariadne sets $k \leftarrow k-1$ (decreases the value of the Nisan parameter by one) and repeats. Note that, as discussed earlier, if the Nisan parameter is equal to the clique number $\omega(G)$, Ariadne switches from the primary to the secondary system as the potential function $\mathsf{P}$ approaches the value $C_{00}$. The mechanism by which the simplified secondary system replaces the primary is simple: The primary system is replaced upon detection of an iterate whose potential value is greater than $\mathsf{C}_\ell$. To avoid cluttering the pseudocode, we do not invoke the secondary system below. The parameter $\alpha_\epsilon$ is chosen so as to satisfy \eqref{aaaaaaaaaa}. Typically $\alpha_\epsilon$ would remain constant given $C_{00}$.

\begin{algorithm}[]
\caption{{\sc Main-Body}($C$, $n$)}
\label{MBC}
\begin{algorithmic}[1]

\State{$k \leftarrow n$}

\While{true}
\State{$X$ = {\sc Compute-Initial-Condition}($C$)}
\State{$\bar{Y}$ = {\sc Compute-Equilibrium}$\left(C, X, \frac{1}{4} \left( \frac{1}{k-1} - \frac{1}{k} \right)\right)$}

\If{$\bar{Y}$ = NULL}

\State{$k \leftarrow k-1$}

\Else

\State{Compute well-supported approximate equilibrium $\hat{Y}$ using $\bar{Y}$}

\If{$\hat{Y}$ is in the carrier of a clique $Y^*$ of size equal to $k$}

\State{ \Return{$Y^*$} }

\Else

\State{$k \leftarrow k-1$}

\EndIf

\EndIf

\EndWhile
  
\end{algorithmic}
\end{algorithm}

The initial condition $X^0$ is set such that it is strictly upper feasible, such that
\begin{align*}
X^0 \cdot CX^0 \geq \frac{1}{2} \left( 1 + \frac{1}{2} \right),
\end{align*} 
and such that
\begin{align*}
\min_{i=1}^n \left\{ X^0(i) \right\} \geq 2^{-n}.
\end{align*}
The first requirement ensures that we can use the discrete-time replicator dynamic to initialize the primary dynamical system. The second requirement ensures that uniform equalizers that are not cliques cannot attract the iterates. The third requirement (plugged in the formula which gives the equilibrium approximation error) ensures that the upper bound on the number of iterations until either a maximum-clique is computed or until the Nisan parameter is abandoned is polynomial.

 \begin{algorithm}[]
\caption{{\sc Compute-Initial-Condition}($C$)}
\label{CIC}
\begin{algorithmic}[1]

\State{Find an edge that belongs to a clique of size four.}
\State{Call the characteristic vector of this edge $X^*$.}
\State{Call the barycenter of $\mathbb{Y}$ $X^0$.}
\While{$X^0$ is not strictly upper feasible and $X^0 \cdot CX^0 < (1/2)(1 + 1/2)$}

\State{$X^0 \leftarrow 1/2 (X^0 + X^*)$}

\EndWhile

\State{$X = X^0$}

\State{$\bar{Y} = X$}

\While{$X$ is above the effective interior}

\State{$X = J(X)$}

\State{Compute multiplier $Y$, use learning rate $\alpha' \leftarrow \alpha_{\epsilon}$, and {\sc Update-and-Extend}($\bar{Y}$, $Y$, $\alpha'$)}

\EndWhile

\If{$X$ is inside the effective interior}

\State{Compute multiplier $Y$, use learning rate $\alpha' \leftarrow \alpha_{\epsilon}$, and {\sc Update-and-Extend}($\bar{Y}$, $Y$, $\alpha'$)}

\State{ \Return{$X$} }
 
\Else

\State{Choose intermediate point between last and second-to-last iterate such that}

		\State{$X$ is strictly lower feasible and $\ddfrac{X \cdot CX - \mathsf{C}}{\max_{i=1}^n \left\{ (CX)_i \right\} - \mathsf{C}_\ell} > \mathsf{G}^*_0$ and $(CX)_{\max} < C_u - \epsilon$}

		\State{Compute multiplier $Y$, use learning rate $\alpha' \leftarrow \alpha_{\epsilon}$, and {\sc Update-and-Extend}($\bar{Y}$, $Y$, $\alpha'$)}
\State{ \Return{$X$} }

\EndIf
\end{algorithmic}
\end{algorithm}

\newpage

Recall that the empirical average of the sequence of multipliers is defined as
\begin{align*}
\bar{Y}^K = \frac{1}{A_K} \sum_{k = 0}^K \alpha_k Y^k \quad \mbox{where} \quad A_K = \sum_{k=0}^K \alpha_k.
\end{align*}
The formula for updating the empirical average in the next subroutine is obtained by straight algebra from the previous definition. But the next subroutine not only updates the empirical average, but also generates the secondary sequence of iterates and extends the sequence of multipliers to prevent convergence of the secondary sequence of iterates to a pure strategy.

\begin{algorithm}[H]
\caption{{\sc Update-and-Extend}($\bar{Y}$, $Y$, $\alpha$)}
\label{UEA}
\begin{algorithmic}[1]

\State{ $\bar{Y} \leftarrow \frac{\alpha}{A + \alpha} Y + \frac{A}{A + \alpha} \bar{Y}$; $A \leftarrow A + \alpha$ }

\State{ \begin{align*}
\tilde{X}(i) \leftarrow \tilde{X}(i) \ddfrac{\exp\{ \alpha (C Y)_i  \}}{\sum_{j=1}^n X(j) \exp\{ \alpha (C Y)_j \}} \quad i =1, \ldots, n.
\end{align*}
}

\While{$(C\tilde{X})_{\max} \geq C_{00} + \epsilon$}
		
		\State{$\tilde{X} = T(\tilde{X})$; $\bar{Y} \leftarrow \frac{\alpha}{A + \alpha} \tilde{X} + \frac{A}{A + \alpha} \bar{Y}$; $A \leftarrow A + \alpha$ }
	
\EndWhile

\end{algorithmic}
\end{algorithm}

Note that the computation of the empirical average $\bar{Y}$ commences at $X^0$ (before an iterate becomes strictly lower feasible). The reason for this is technical, useful in upper-bounding the number of iterations of the dynamical system until computation of a maximum clique by a polynomial.

\begin{algorithm}[H]
\caption{{\sc Invoke-upper-mechanism}$(C, X)$}
\label{MBB}
\begin{algorithmic}[1]

\State{X = J(X); Compute multiplier $Y$; use rate $\alpha' \leftarrow \alpha_{\epsilon}$; {\sc Update-and-Extend}($\bar{Y}$, $Y$, $\alpha'$)}
		
		\While{$(CX)_{\max} \geq \mathsf{C}_u - \epsilon$}
		
		\State{X = J(X); Compute multiplier $Y$; use rate $\alpha' \leftarrow \alpha_{\epsilon}$; {\sc Update-and-Extend}($\bar{Y}$, $Y$, $\alpha'$)}
		
		\EndWhile
		
		\If{$X$ is not in the effective interior}

		\State{Choose intermediate point between last and second-to-last iterate such that}

		\State{$X$ is strictly lower feasible and $\ddfrac{X \cdot CX - \mathsf{C}}{\max_{i=1}^n \left\{ (CX)_i \right\} - \mathsf{C}_\ell} > \mathsf{G}^*_0$ and $(CX)_{\max} < \mathsf{C}_u - \epsilon$}

		\State{Compute multiplier $Y$, use learning rate $\alpha' \leftarrow \alpha_{\epsilon}$, and {\sc Update-and-Extend}($\bar{Y}$, $Y$, $\alpha'$)}
		
		\State{\Return{$X$}}

		\EndIf
  
\end{algorithmic}
\end{algorithm}

In the following subroutine $\bar{k}$ is a configurable constant parameter. Its goal is to configure parameters $\mathsf{C}$ (of the first and second primary map) such that the potential function $\mathsf{P}$ increases. Parameter $\bar{k}$ limits the number of times these parameters need to be reconfigured.

\begin{algorithm}[H]
\caption{{\sc Wind}($C$, $X$, $k$, $\alpha$)}
\label{MB}
\begin{algorithmic}[1]

\If {$k \geq \bar{k}$}

\State{$\hat{X} = T_{\mathsf{G}}(X)$; $\hat{X'} = T_{\mathsf{G'}}(\hat{X})$}

\State{$\mathsf{C} = X \cdot CX$; $\mathsf{C'} = \hat{X} \cdot C\hat{X}$}

\Else

\For {$i = 1$ to $k$}

\State{$\hat{X} = T_{\mathsf{G}}(X)$; $\hat{X'} = T_{\mathsf{G'}}(\hat{X})$}

\State{$\mathsf{C} = \frac{1}{2} \left(\mathsf{C} + X \cdot CX \right)$; $\mathsf{C'} = \frac{1}{2} \left(\mathsf{C'} + \hat{X} \cdot C\hat{X} \right)$}

\EndFor

\EndIf

\State{\Return{$(\hat{X}, \hat{X'})$}}
  
\end{algorithmic}
\end{algorithm}

The following subroutine implements the {\em leapfrogging mechanism} that circumvents non-equilibrium fixed points. Note that as a fixed point is circumvented using this mechanism, the computation of the sequences of iterates and multipliers does not start anew, but rather treats the ``leapfrogging step'' as a step in the corresponding dynamical system. There is a technical reason for this, namely, to retain the favorable properties of $X^0$ in the bound on the equilibrium approximation error.

\begin{algorithm}[H]
\caption{{\sc Circumvent-Non-Equilibrium-Fixed-Points}$(C, X)$}
\label{MBA}
\begin{algorithmic}[1]

\State{Sort $X$ in descending order to produce probability vector $R_X$}

\For {$i = n$ to $1$}

\If{top-$i$ elements of $R_X$ are a clique, say $X^*$, that is inside the effective interior}

\If{$X$ is in the probability sector of $X^*$ and $\mathsf{P}(X^*) > \mathsf{P}(X)$}

\State{Find $\hat{X}$ in the effective interior such that $\mathsf{P}(\hat{X}) > \mathsf{P}(X^*)$ and set $X \leftarrow \hat{X}$}

\State{Compute multiplier $Y$; use rate $\alpha' \leftarrow \alpha_{\epsilon}$; {\sc Update-and-Extend}($\bar{Y}$, $Y$, $\alpha'$)}

\State{\Return{$X$}}

\EndIf

\Else

\If{top-$i$ elements of $R_X$ are a clique $X^*$ that is on the boundary the effective interior}

\If{$X$ is the probability sector of $X^*$ and $\mathsf{P}(X^*) > \mathsf{P}(X)$}

\State{Slightly increase $\epsilon$ or $\mathsf{G}^*_0$ until the sequence exceeds $\mathsf{P}(X^*)$}

\State{\Return{$X$}}

\EndIf

\EndIf

\EndIf

\EndFor
  
\end{algorithmic}
\end{algorithm}

Finally, one point is worth discussing in the next subroutine in relation to line 46, namely, that it is possible that the equilibrium approximation error of the empirical average of the extended sequence of multipliers may attain the value $\epsilon_a^2/8$ more than once. But once the upper bound on the equilibrium approximation error (as that is computed in the proof of Lemma \ref{theorem1}) is less than $\epsilon_a^2/8$, then this is no longer possible as this upper bound monotonically diminishes.

\begin{algorithm}[H]
\caption{{\sc Compute-Equilibrium}($C$, $X$, $\epsilon_a$)}
\label{CE}
\begin{algorithmic}[1]

\State{$\alpha_h \leftarrow \alpha_{\epsilon}$, $\alpha_\ell \leftarrow 0.001$}

 \While{true}
 
 \If{$X \cdot CX \geq \mathsf{C}_\ell + \epsilon$}
 
 \State{$\alpha_h \leftarrow d$, $\alpha_\ell \leftarrow d/1000$}
 
 \EndIf
 
 \State{$k=1$}
 
 \While{true}
 
 \State{$(X_1, X_2)$  = {\sc Wind}($C$, $X$, $k$, $\alpha_h$)}
 
 \If{$X_1 \cdot CX_1 - X \cdot CX > 0$ and $X_2 \cdot CX_2 - X_1 \cdot CX_1 > 0$}
 
 \State{break}
 
 \Else
 
 \State{$k \leftarrow k+1$}
 
 \EndIf
 
 \EndWhile

 \If{$(CX_2)_{\max} < \mathsf{C}_u - \epsilon$ and $\mathsf{G}^*(X_2) < \mathsf{G}^*_0$}
 
 \State{$X = X_2$; Compute multiplier $Y$; use rate $\alpha' \leftarrow \alpha_{\epsilon}$; {\sc Update-and-Extend}($\bar{Y}$, $Y$, $\alpha'$)}
 
 \Else \If{$(CX_2)_{\max} > \mathsf{C}_u - \epsilon$}

\State{Find $\alpha$ such that $(C (T_{\mathsf{G'}} \circ T_{\mathsf{G}})(X_2))_{\max} = \mathsf{C}_u - \epsilon$ using the bisection method}

\State{If find point whose potential $\leq X \cdot CX$ goto line $7$ using $k \leftarrow k+1$}

\If{$\alpha > \alpha_\ell$}

\State{$\alpha \leftarrow \alpha / 2$; $X = (T_{\mathsf{G'}} \circ T_{\mathsf{G}})(X_2)$;}

\State{Compute multiplier $Y$; use rate $\alpha' \leftarrow \alpha_{\epsilon}$; {\sc Update-and-Extend}($\bar{Y}$, $Y$, $\alpha'$)}

\Else

\State{$X = (T_{\mathsf{G'}} \circ T_{\mathsf{G}})(X_2)$ using $\alpha$}

\State{Compute multiplier $Y$; use rate $\alpha' \leftarrow \alpha_{\epsilon}$; {\sc Update-and-Extend}($\bar{Y}$, $Y$, $\alpha'$)}

\State{$X$ = {\sc Invoke-upper-mechanism}$(C, X)$}

\EndIf

\Else \If{$\mathsf{G}^*(X_2) < \mathsf{G}^*_0$}

\State{Find $\alpha$ such that $\mathsf{G}^*((T_{\mathsf{G'}} \circ T_{\mathsf{G}})(X_2)) = \mathsf{G}^*_0$ using the bisection method}

\State{If find point whose potential $\leq X \cdot CX$ goto line $7$ using $k \leftarrow k+1$}

\If{$\alpha > \alpha_\ell$}

\State{$\alpha \leftarrow \alpha / 2$; $X = (T_{\mathsf{G'}} \circ T_{\mathsf{G}})(X_2)$}

\State{Compute multiplier $Y$; use rate $\alpha' \leftarrow \alpha_{\epsilon}$; {\sc Update-and-Extend}($\bar{Y}$, $Y$, $\alpha'$)}

\Else

\State{$X = (T_{\mathsf{G'}} \circ T_{\mathsf{G}})(X_2)$ using $\alpha$}

\State{Compute multiplier $Y$; use rate $\alpha' \leftarrow \alpha_{\epsilon}$; {\sc Update-and-Extend}($\bar{Y}$, $Y$, $\alpha'$)}

\State{$X = T_{\mathsf{G}}^{OPT}$ using some $\alpha$ such that $X$ is in the effective interior}

\State{Compute multiplier $Y$; use rate $\alpha' \leftarrow \alpha_{\epsilon}$; {\sc Update-and-Extend}($\bar{Y}$, $Y$, $\alpha'$)}

\EndIf

\EndIf

\EndIf

\EndIf

 \State{$X$ = {\sc Circumvent-Non-Equilibrium-Fixed-Points}$(C, X)$}

 \State{If $(C\bar{Y})_{\max} - \bar{Y} \cdot C \bar{Y}$ is upper bounded by $\epsilon_a^2/8$, then return $\bar{Y}$}
 
 \State{If the equilibrium approximation bound has been violated, then return NULL}

\EndWhile

 \end{algorithmic}
  \end{algorithm}

\end{document}